\documentclass[10pt,sigconf]{acmart}
\setcopyright{none}
\allowdisplaybreaks
\usepackage{calrsfs}
\DeclareMathAlphabet{\pazocal}{OMS}{zplm}{m}{n}
\usepackage{epsfig,endnotes,xspace,diagbox}
\AtBeginDocument{}
\usepackage{amsmath,amsthm,amssymb}
\usepackage{tikz}
\usepackage{enumerate}
\usepackage{varioref}
\usepackage{graphicx}
\usepackage{epstopdf}
\usepackage{color}
\usepackage{xcolor}
\usepackage{multirow}
\usepackage[labelformat=simple]{subcaption}

\usepackage{comment}
\usepackage[utf8]{inputenc}
\usepackage{mathtools}
\usepackage{wrapfig}
\usepackage{algorithm, algorithmic}
\usepackage{caption}
\newcommand{\greekbf}[1]{\text{\boldmath $#1$}}

\newcommand{\z}{\mathbf{z}}
\newcommand{\E}{\mathbb{E}}

\DeclareMathOperator*{\argmin}{arg\,min}
\DeclareMathOperator*{\argmax}{arg\,max}

\renewcommand{\Pr}[1]{\ensuremath{\mathsf{Pr}\left[#1\right]}\xspace}
\newcommand{\sqrtee}{e^{\epsilon/2}}

\newcommand{\mypara}[1]{\vspace*{0.06in}\noindent\textbf{#1 }}

\newcommand{\Domain}{\mathcal{D}\xspace}
\newcommand{\Codomain}{\Tilde{\mathcal{D}}\xspace}

\newcommand{\rdv}{\tilde{v}}
\newcommand{\rdvv}{\tilde{\mathbf{v}}}
\newcommand{\vv}{\mathbf{v}}
\newcommand{\RV}{\tilde{V}}
\newcommand{\estx}{\hat{\mathbf{x}}}
\newcommand{\truex}{\mathbf{x}}
\newcommand{\noisy}{\Tilde{\mathbf{v}}}

\newcommand{\originCDF}{\mathbf{P}}
\newcommand{\M}{\mathbf{M}}
\newcommand{\x}{\mathbf{x}}
\newcommand{\w}{\mathbf{w}}

\newcommand{\y}{\mathbf{y}}
\newcommand{\iter}[1]{\ensuremath{^{(#1)}}\xspace}
\newcommand{\vmu}{\greekbf{\mu}}
\newcommand{\vnu}{\greekbf{\nu}}
\newcommand{\veta}{\greekbf{\eta}}

\newcommand{\perturb}{\ensuremath{{\Psi}}\xspace}

\newcommand{\olh}{\ensuremath{\mathsf{OLH}}\xspace}

\newcommand{\grr}{\ensuremath{\mathsf{GRR}}\xspace}
\newcommand{\SR}{\ensuremath{\mathsf{SR}}\xspace}
\newcommand{\PM}{\ensuremath{\mathsf{PM}}\xspace}

\newcommand{\SW}{\ensuremath{\mathsf{SW}}\xspace}
\newcommand{\GW}{\ensuremath{\mathsf{GW}}\xspace}

\newcommand{\fo}{\ensuremath{\mathsf{FO}}\xspace}

\newcommand{\CFO}{\ensuremath{\mathsf{CFO}}\xspace}
\newcommand{\tuple}[1]{\ensuremath{\langle #1 \rangle}}
 
\title{Estimating Numerical Distributions under Local Differential Privacy}

\begin{document}
\author{Zitao Li$^1$, Tianhao Wang$^1$, Milan Lopuhaä-Zwakenberg$^2$, }
\author{Boris Skoric$^2$, Ninghui Li$^1$}
\affiliation{
  \institution{$^1$Purdue University, $^2$Eindhoven University of Technology}
}
\email{
	{ li2490, wang2842, ninghui}@purdue.edu, 
	{m.a.lopuhaa, b.skoric}@tue.nl
}

\begin{CCSXML}
<ccs2012>
<concept>
<concept_id>10002978.10002991.10002995</concept_id>
<concept_desc>Security and privacy~Privacy-preserving protocols</concept_desc>
<concept_significance>500</concept_significance>
</concept>
</ccs2012>
\end{CCSXML}

\ccsdesc[500]{Security and privacy~Privacy-preserving protocols}

\keywords{local differential privacy, density estimation}

\begin{abstract}
When collecting information, local differential privacy (LDP) relieves the concern of privacy leakage from users' perspective, as user's private information is randomized before sent to the aggregator.  
We study the problem of recovering the distribution over a numerical domain while satisfying LDP. 
While one can discretize a numerical domain and then apply the protocols developed for categorical domains, 
we show that taking advantage of the numerical nature of the domain results in better trade-off of privacy and utility. 
We introduce a new reporting mechanism, called the square wave (\SW) mechanism, which exploits the numerical nature in reporting.  
We also develop an Expectation Maximization with Smoothing (EMS) algorithm, which is applied to aggregated histograms from the \SW mechanism to estimate the original distributions.  Extensive experiments demonstrate that our proposed approach, \SW with EMS, consistently outperforms other methods in a variety of utility metrics.
\end{abstract}

\sloppypar

\maketitle

\section{Introduction}

Differential privacy~\cite{tcc:DworkMNS06} has been accepted as the \textit{de facto} standard for data privacy.  
Recently, techniques for satisfying differential privacy (DP) in the local setting, which we call {LDP}, have been studied and deployed.  
In the local setting for DP, there are many \emph{users} and one \emph{aggregator}.  
Each user sends randomized information to the aggregator, who attempts to infer the data distribution based on users' reports.  
LDP techniques enable the gathering of statistics while preserving privacy of every user, without relying on trust in a single trusted third party.  
LDP techniques have been deployed by companies like
Apple~\cite{url:apple}, Google~\cite{ccs:ErlingssonPK14}, and Microsoft~\cite{nips:DingKY17}.

Most existing work on LDP focuses on the situations where the attributes that one wants to collect are categorical.  
Existing research~\cite{ccs:ErlingssonPK14,stoc:BassilyS15,uss:WangBLJ17,arXiv:AcharyaSZ18,tiot:YeB18} has developed frequency oracle (\fo) protocols for categorical domains, where the aggregator can estimate the frequency of any chosen value in the specified domain (fraction of users with that private value). We call these Categorical Frequency Oracle (\CFO) protocols.

Many attributes are ordinal or numerical in nature, e.g., income, age, the amount of time viewing a certain page, the amount of communications, the number of times performing a certain actions, etc.  A numerical domain consists of values that have a meaningful total order.
One natural approach for dealing with ordinal and numerical attributes under LDP is to first apply binning and then use \CFO protocols.  That is, one treats all values in a range as one categorical value when reporting.  This approach faces the challenge of finding the optimal number of bins, which depends on both the privacy parameter and the data distribution.  One improvement over this approach is to apply Hierarchical Histogram-based approaches~\cite{pvldb:HayRMS10,pvldb:QardajiYL13,icde:XiaoWG10}, which uses multiple granularities at the same time, and exploit the natural consistency relationships between estimations at different granularities.  Recently, Kulkarni et al.~\cite{pvldb:KulkarniCD18} studied the accuracy of answering range queries using this approach.

We note that the stronger privacy guarantee offered by LDP (as compared with DP) comes with the cost of significantly higher noises.  As a result, many estimated frequencies will be negative.  Existing approaches (such as~\cite{pvldb:KulkarniCD18}) do not correct this, and are sub-optimal.  We propose to apply Alternating Direction Method of Multipliers (ADMM) optimization~\cite{book:boyd11distributed} to improve Hierarchical Histograms, utilizing the constraints that all estimations are non-negative and sum up to 1.  Experiments show that the improved version of hierarchy histogram, which we call HH-ADMM, has significantly better utility. 

The above methods still use \CFO protocols in a blackbox fashion, and existing \CFO protocols ignore any semantic relationship between different values.  An intriguing research question is whether one can design frequency oracle protocols that directly utilize the ordered nature of the domain and produce better estimations.  In this paper, we answer this affirmatively.  We propose an approach that combines what we call a Square Wave reporting mechanism with post-processing using Expectation Maximization and Smoothing.  

The key intuition under the Square Wave mechanism is that given input $v$, one should report a value close to $v$ with higher probability than a value farther away from $v$.  More specifically, assuming the input domain of numerical values is $\Domain=[0,1]$, the output domain of Square Wave mechanism is $\Codomain=[-b,1+b]$, where $b$ is a parameter depending on the privacy parameter $\epsilon$.  
A user with value $v \in \Domain$ reports a value $\rdv$ randomly drawn from a distribution with probability density function $\M_v$. 
For any $\rdv \in [v-b, v+b]$, probability density is $\M_v(\rdv) = p$, and any $\rdv \in [-b, 1+b] \setminus [v-b, v+b]$, probability density is $\M_v(\rdv) = q$, where $\frac{p}{q}=e^\epsilon$.  
We define and studied different wave shapes of General Wave mechanism other than the above Square Wave, and concluded that Square Wave has the best utility. 
We also studied how to determine the key parameter $b$, the width of the wave.  
We propose to choose $b$ to maximize the upper bound of mutual information between the input and the output variable, and can compute $b$ when given the privacy parameter $\epsilon$.  Experiments demonstrate the effectiveness of this approach. 

Conceptually, the aggregator, after observing the reported values, without any prior knowledge of the input distribution, should perform Maximum Likelihood Estimation (MLE) to infer the input distribution, which can be carried out by the Expectation Maximization (EM) algorithm.  Through experiments, we have observed that the result of applying EM is highly sensitive to the parameter controlling terminating condition.  This is because the observed distribution is a combination of the true distribution and the effect of random noise.  When EM terminates too early, the result does not fit the true distribution well.  When EM terminates too late, the result fits both the true distribution and the effect of noises.  It is unclear how one can set the parameter so that one fits the distribution, but not the noise, across different datasets and privacy parameters. 

To deal with this challenge, we propose to use smoothing together with the EM algorithm.  In each iteration, after the E step and the M step, we add an S (smoothing) step, which averages each estimation with its nearest neighbours, by binomial coefficients.  The Expectation Maximization with Smoothing approach was developed in the context of positron emission tomography and image reconstruction~\cite{nychka1990some, silverman1990smoothed}, and was shown to be equivalent to adding a regularization term penalizing the spiky estimation~\cite{nychka1990some}.  Intuitively, EMS uses the prior knowledge that the observation is affected by noise and prefer a smoother distribution to a jagged one.  In the experiment, we observe that EMS is stable under different settings, and requires no parameter tuning.

To compare different algorithms for reconstructing distributions of numerical attributes, we propose to use a number of metrics.  We use two metrics measuring the distance of reconstructed cumulative distribution from the true one, namely the Wasserstein distance and Kolmogorov–Smirnov distance (KS distance).  In addition, we also consider accuracy for answering range queries, and accuracy of estimations of different statistics from the reconstructed distributions such as mean, variance and quantiles.

The contributions of this paper are as follows.
\begin{enumerate}
    \item We define the problem of reconstructing distributions of numerical attributes under LDP (with non-negativity and sum-up-to-1 constraints) and propose multiple metrics for comparing competing algorithms. 
    \item We introduce HH-ADMM, which improves upon existing hierarchy histogram based methods.
    \item We introduce the method of combining Square Wave (\SW) reporting with Expectation Maximization and Smoothing (EMS), and showed that Square Wave is preferable to other wave shapes, and introduced techniques to choose the bandwidth parameter $b$ using mutual information. 
    \item We conduct extensive experimental evaluations, comparing the proposed methods with state-of-the-art methods (e.g.,~\cite{pvldb:KulkarniCD18}).  Results demonstrate that \SW with EMS and HH-ADAM significantly out-perform existing methods.  In addition, \SW with EMS generally performs the best under a wide range of metrics, and HH-ADMM performs better than \SW-EMS on a very spiking distribution under some of the metrics.
\end{enumerate}

\mypara{Roadmap.}
In Section~\ref{sec:background}, we review the LDP definition and existing LDP protocols. 
In Section~\ref{sec:problemdef}, we discuss metrics for measuring the quality of the reconstructed distribution.
We describe \CFO with binning and HH-ADMM in Section~\ref{sec:adapted_method}. 
\SW reporting and EMS reconstruction are introduced in Section~\ref{sec:method}. 
We show our experimental results in Section~\ref{sec:experiments}, discuss the related work in Section~\ref{sec:related}, and conclude in Section~\ref{sec:conc}.

 \section{Background}
\label{sec:background}

Assume there are $n$ \emph{users} and one \emph{aggregator}.  
Each user possesses a value ${v} \in \Domain$, and the aggregator wants to learn the distribution of values from all users. 
To protect privacy, each user randomizes the input value $v$ using an algorithm $\perturb(\cdot) : \Domain \rightarrow \Codomain$, where $\Codomain$ is the set of all possible outputs, and sends $\tilde{v}=\perturb(v)$ to the aggregator.

\begin{definition}[$\epsilon$-Local Differential Privacy] \label{def:dlp}
	An algorithm $\perturb(\cdot) : \Domain \rightarrow \Codomain$ satisfies $\epsilon$-local differential privacy ($\epsilon$-LDP), where $\epsilon \geq 0$,
	if and only if for any input $v_1,v_2 \in \Domain$, we have
	\begin{equation*}
	\forall{T\subseteq\! \Codomain}:\; \Pr{\perturb(v_1)\in T} \leq e^{\epsilon}\, \Pr{\perturb(v_2)\in T},\end{equation*}
	where $\mathit{Range}(\perturb)$ denotes the set of all possible outputs of $\perturb$.
\end{definition}

Since a user never reveals $v$ to the aggregator and reports only $\rdv = \perturb(v)$, the user's privacy is still protected even if the aggregator is malicious.

\mypara{Notational Conventions.}
Throughout the paper, we use bold letters to denote vectors.  
For example, $\mathbf{v} = \tuple{v_1, \ldots, v_n}$ is all users' values, and $\mathbf{x} = \tuple{x_1, \ldots, x_d}$ is frequencies of all values (i.e., $x_i = |\{j\mid v_j = i\}| / n$).
If the notation is associated with a tilde (e.g., $\noisy$), it is the value after LDP perturbation; and a hat (e.g., $\estx$) denotes the value computed by the aggregator.  Capital bold letters denote matrices and functions that take more than one input. 
Table~\ref{tbl:notations} gives some of the frequently used symbols.
\begin{table}
\begin{center}
\resizebox{0.48\textwidth}{!}{\begin{tabular}{||c|c||}
        \hline
        Symbol & Description \\
        \hline
         $v$ & Private input\\
         $\rdv$ & Randomized output \\
         \hline
         $\Domain$ & Domain of private input \\
         $\Codomain$ & Domain of the randomized output \\
\hline
         $\x$ & True private input frequencies \\
         $\estx$ & Estimate of private input frequencies (normalized) \\
         $\noisy$ & Randomized output frequencies (normalized)\\
         \hline
         $\originCDF$ & Cumulative distribution function (CDF) \\
         $\M_v$ &  Probability density function given input $v$ \\
         \hline
    \end{tabular}
    }
    \caption{Notations.}
    \label{tbl:notations}
\end{center}
\vspace{-1cm}
\end{table}

\subsection{Categorical Frequency Oracles}
\label{subsec:fo}

A \textit{frequency oracle ($\fo$)} protocol enables the estimation of the frequency of any value $v \in \Domain$ under LDP.
Existing protocols are designed for situations where $\Domain$ is a categorical domain.  We call them \emph{categorical frequency oracle} (\CFO) protocols in this paper.
The following are two commonly used \CFO protocols.

\mypara{Generalized Randomized Response (\grr).}
This \CFO protocol generalizes the \emph{randomized response} technique~\cite{jasa:Warner65}, and uses $\Codomain=\Domain$.   
It uses as input perturbation function $\mathsf{GRR}(\cdot)$, where $\mathsf{GRR}(v)$ outputs the true value $v$ with probability $p=\frac{e^\epsilon}{e^\epsilon + d - 1}$, and any value $v'\ne v$ with probability $q=\frac{1-p}{d-1}=\frac{1}{e^\epsilon + d - 1}$, where $d=|\Domain|$ is the domain size.
To estimate the frequency of $v\in \Domain$ (i.e., the ratio of the users who have $v$ as private value to the total number of users), one counts how many times $v$ is reported, and denote the count as $C(v)$, and then computes
\begin{align}
\tilde{x}_v =  \frac{(C(v)/n)-q}{p-q}\nonumber \ ,
\end{align}
where $n$ is the total number of users.  
In~\cite{sp:wang2018locally}, it is shown that this is an unbiased estimate of the true count, and the variance for this estimate is
\begin{equation}\label{var_grr}
\mathrm{Var}[\tilde{x}_v]=\frac{d-2+e^\epsilon}{(e^\epsilon-1)^2\cdot n} \ .
\end{equation}
The variance given in~\eqref{var_grr} is linear to $d$; thus when the domain size $d$ increases, the accuracy of this protocol is low.

\mypara{Optimized Local Hashing (\olh)~\cite{sp:wang2018locally}.}
This protocol deals with a large domain size $d=|\Domain|$ by first using a hash function to map an input value into a smaller domain of size $g$ (typically $g\ll |\Domain|$), and then applying randomized response to the hashed value (which leads to $p=\frac{e^\epsilon}{e^\epsilon + g - 1}$).  In this protocol, both the hashing step and the randomization step result in information loss. The choice of the parameter $g$ is a tradeoff between losing information during the hashing step and losing information during the randomization step.  In~\cite{sp:wang2018locally}, it is found that the optimal choice of $g$ that leads to minimal variance is $(e^\epsilon+1)$.

In \olh, one reports $\tuple{H,\grr(H(v))}$
where $H$ is randomly chosen from a family of hash functions that hash each value in $\Domain$ to $\{1\ldots g\}$, and 
$\grr(\cdot)$ is the perturbation function for Generalized Randomized Response, while operating on the domain $\{1\ldots g\}$.
Let $\tuple{H^j,y^j}$ be the report from the $j$'th user.
For each value $v\in \Domain$, to compute its frequency, one first computes $C(v)=|\{j\mid H^j(v) = y^j\}|$, and then transforms $C(v)$ to its unbiased estimate
\begin{align*}
\tilde{x}_v = \frac{(C(v)/n) - (1/g)}{p-1/g}.
\end{align*}

The approximate variance of this estimate is
\begin{align*}
\mathrm{Var}[\tilde{x}_v]=\frac{4e^\epsilon}{(e^\epsilon-1)^2\cdot n}.
\end{align*}
Compared with \eqref{var_grr}, the factor $d-2+e^\epsilon$ is replaced by $4 e^\epsilon$.  This suggests that for smaller $|\Domain|$ (such that $|\Domain|-2<3e^\epsilon$), \grr is better; but for large $|\Domain|$, \olh is better and has a variance that does not depend on $|\Domain|$.

\subsection{Handling Numerical Attributes}\label{sec:background:mean}

Two methods have been proposed for mean estimation under LDP for numerical attributes.  Note that using these methods one can estimate the mean, and not the distribution.

\mypara{Stochastic Rounding (\SR)~\cite{jasa:DuchiJW18}.}  
The main idea of Stochastic Rounding (\SR) is that, no matter what is the input value $v$, each user reports one of two extreme values, with probabilities depending on $v$.  Here we give an equivalent description of the protocol.  
Following ~\cite{jasa:DuchiJW18}, we assume that the input domain is $[-1, 1]$.
Given a value $v \in [-1, 1]$, let $p=\frac{e^\epsilon}{e^\epsilon+1}$ and $q=1-p=\frac{1}{e^\epsilon+1}$, the \SR method outputs a random variable $v'$, which takes the value $-1$ with probability $q+\frac{(p-q)(1-v)}{2}$ and value $1$ with probability $q+\frac{(p-q)(1+v)}{2}$.  Since   
\begin{align*}
 \E[v'] & =  (-1)\left(q+\frac{(p-q)(1-v)}{2}\right) +  q +\frac{(p-q)(1+v)}{2} \\
        & =  (p-q)v \ .
\end{align*}
Let $\tilde{v} = \frac{v'}{p-q}$, we have $\E[\tilde{v}] = v$; thus the mean of $\tilde{v}$ provides an unbiased estimate of the mean for the distribution.

\mypara{Piecewise Mechanism (\PM)~\cite{icde:WangXYHSSY18}.}
In the Piecewise Mechanism, the input domain is $[-1,1]$, and the output domain is $[-s, s]$, where $s=\frac{e^{\epsilon/2}+1}{e^{\epsilon/2}-1}$.  For each $v \in [-1,1]$, there is an associated range $[\ell(v),r(v)]$ where $-s\le \ell(v) < r(v) \le s$, such that with input $v$, a value in the range $[\ell(v),r(v)]$ will be reported with higher probability than a value outside the range.  
More precisely, we have $\ell(v)=\frac{\sqrtee\cdot v - 1}{\sqrtee - 1}$ and $r(v)=\frac{\sqrtee\cdot v + 1}{\sqrtee - 1}$.  
The width of the range is $r(v) - \ell(v) = \frac{2}{\sqrtee - 1}$, and the center is $\frac{\ell(v)+r(v)}{2} = \frac{\sqrtee}{\sqrtee - 1}\cdot v$.  Specifically, $\PM$ works as follows:
\begin{align*}
    \Pr{\PM(v)=\rdv} & = \frac{\sqrtee}{2} \cdot \frac{\sqrtee - 1}{\sqrtee + 1} \mbox{ if } \rdv \in [\ell(v), r(v)],\\
    \Pr{\PM(v)=\rdv} & = \frac{1}{2\sqrtee} \cdot \frac{\sqrtee - 1}{\sqrtee + 1} \mbox{ otherwise}.
\end{align*}
It is shown that $\tilde{v}$ is unbiased, and has better variance than \SR when $\epsilon$ is large~\cite{icde:WangXYHSSY18}.

\section{Utility Metrics} 
\label{sec:problemdef}

When the private values are in a numerical domain, we need utility metrics that are different from those in categorical domains.  
In particular, the metrics should reflect the ordered nature of the underlying domain.

\subsection{Metrics based on Distribution Distance}

We want a metric to measure the distance between the recovered density distribution and the true distribution.  
However, since the distribution is over a metric space, we do not want to use point-wise distance metrics such as the $L_1$ and $L_2$ distance or the Kullback–Leibler (KL) divergence. 
For a simple example, consider the case where $\Domain = \{1,2,3,4\},$ the true distribution is $\x = [0.7, 0.1, 0.1, 0.1]$.  The two estimations  $\estx_1 = [0.1, 0.7, 0.1, 0.1]$ and $\estx_2 = [0.1, 0.1, 0.1, 0.7]$ have the same $L_1$, $L_2$, and KL distance from $\x$, but the distance between $\estx_1$ and $\x$ should be smaller than the distance between $\estx_2$ and $\x$ when we consider the numerical nature.  To capture this requirement, we propose to use two popular distribution distances as metrics.

\mypara{Wasserstein Distance (aka. Earth Mover Distance).}
Wasserstein distance measures the cost of moving the probability mass (or density) from distribution to another distribution.
In this paper, we use the one dimensional Wasserstein distance.
For discrete domain, define the cumulative function $\originCDF:[0,1]^d\times \Domain \mapsto [0, 1]$ that takes a distribution $\x$ and a value $v$, and output $\originCDF(\x, v) = \sum_{i=1}^v x_v$.  
Let $\x$ and $\estx$ be two distributions. The one dimensional Wasserstein distance is the $L_1$ difference between their cumulative distributions:
    \begin{align*}
        W_1(\x, \estx) = 
\sum_{v \in \Domain}|\originCDF(\x, v) - \originCDF(\estx, v)| \ .
    \end{align*}
For continuous domain, $\x$ is the probability density function with support on $[0,1]$, $\originCDF(\x, v) = \int_{t=0}^v x(t)dt$.  The one dimensional Wasserstein distance is 
    \begin{align*}
        W_1(\x, \estx) = 
        \int_{v \in \Domain}|\originCDF(\x, v) - \originCDF(\estx, v)|\ dv \ .
    \end{align*}

\mypara{Kolmogorov-Smirnov (KS) Distance .}
KS distance is the maximum absolute difference at any point between the cumulative functions of two distributions:
\begin{align*}
    d_{KS}(\x, \estx) = \sup_{v\in \Domain}\left|\originCDF(\x, v) - \originCDF(\estx, v)\right| \ .
\end{align*}
Intuitively, Wasserstein distance measures the area between the two CDF curves, and KS-distance the maximum height difference between them.

\subsection{Semantic and Statistical Quantities}

Range queries have been used as the main utility metrics for research in this area~\cite{pvldb:HayRMS10, pvldb:KulkarniCD18, wang2019answering, arXiv:Wang19LLLS}.
Also, we consider the basic statistics from the estimated data distributions and check whether they are accurate.

\mypara{Range Query.}
Define the range query function $\mathbf{R}(\x, i, \alpha) = \originCDF(\x, i + \alpha) - \originCDF(\x, i)$, where $\alpha$ specifies the range size.
Given the true distribution $\x$ and the estimated distribution $\estx$, range queries reflect the quality of estimate with randomly sampling $i$ and calculating the following:
\begin{align*}
    |\mathbf{R}(\x, i, \alpha) - \mathbf{R}(\estx, i, \alpha)| \ .
\end{align*}

\mypara{Mean.}
We denote $\mu$ to denote the mean of the true distribution, and $\hat{\mu}$ the estimated mean.  
To measure mean accuracy, we use the absolute value of the difference between these two, i.e. $|\mu -\hat{\mu}|$.

\mypara{Variance.}
We use $\sigma^2$ to denote the variance of the true distribution, and $\hat{\sigma}^2$ for the variance from the reconstructed distribution.
To measure variance accuracy, we use the absolute value of the difference between these two, i.e. $|\sigma^2 -\hat{\sigma}^2|$.

\mypara{Quantiles.}
Quantiles are cut points dividing the range of a probability distribution into intervals with equal probabilities.  Formally, $\mathbf{Q}(\x, \beta) = \argmax_v\{\originCDF(\x, v) \leq \beta \}$.
In the experiment, define $B = \{ 10\%, 20\%, \ldots, 90\%\}$, we measure the following: 
\begin{align*}
    \frac{1}{|B|}\sum_{\beta\in B}|\mathbf{Q}(\x, \beta) - \mathbf{Q}(\estx, \beta)| \ .
\end{align*}

 \section{Using \CFO Protocols for Numerical Domains} 
\label{sec:adapted_method}

In this section, we present two approaches that use \CFO protocols to reconstruct distributions over an discrete numerical domain $\Domain=\{1,2\cdots,d\}$.  
Continuous numerical domains can be buckized into discrete ones.

\subsection{\CFO with Binning}

Given a numerical domain, one can make it discrete using binning, and then have each user report which bin the private value is in using a \CFO protocol.  
For a given domain size and privacy parameter $\epsilon$, one chooses either \olh or \grr, based on which one gives lower estimation variance. 
After obtaining density estimations for all the bins, one computes a density distribution for the domain by assuming uniform distribution within each bin. However, some estimated values may be negative, which does not lead to valid cumulative distribution functions on the domain. 
In~\cite{arXiv:Wang19LLLS}, it is shown that a post-processing method called Norm-Sub can be applied to improve estimation.  Norm-sub converts negative estimates to $0$ and subtracts the same amount to all the positive estimates so that they sum up to $1$.  If some positive estimates become negative after the subtraction, the process is repeated.  This results in an estimation such that each estimation is non-negative and all estimations sum up to $1$.  It can thus be interpreted as a probability distribution.

\mypara{Challenge of Choosing Bin Size.}
When using binning, there are two sources of errors: noise and bias due to grouping values together.  More bins lead to greater error due to noises.   
Fewer bins lead to greater error due to biases.  Choosing the bin size is a trading-off of the above two sources of errors, and the effect of each choice depends both on the privacy parameter $\epsilon$, and on property of the distribution.   
For example, when a distribution is smooth, one would prefer using less bins, as the bias error is small, and when a distribution is spiky, using more bins would perform better. 
In our experiments, we observe that even if we could choose the optimal bin size empirically for each dataset and $\epsilon$ value (which is infeasible to do in practice due to privacy), the result would still be worse than the method to be proposed in Section~\ref{sec:method}.  We thus chose not to develop ways to choose bin size based on $\epsilon$, and just report results of this method under several different bin sizes.

\subsection{Hierarchy-based Methods}
Hierarchy-based methods, including Hierarchy Histogram (HH) in~\cite{pvldb:HayRMS10,pvldb:QardajiYL13} and Haar in~\cite{icde:XiaoWG10}, were first proposed in the centralized setting of DP. 
In \cite{pvldb:KulkarniCD18}, Kulkarni et al.~studied the HH method and the Haar in the context of LDP.
In order to adapt Haar method to the local setting, they used Hadamard random response (HRR) as the frequency oracle. 
HRR is simliar to Local Hashing method introduced in the Section~\ref{subsec:fo}, but fixing $g=2$ and 
using a Hadamard matrix as the family of hash functions.  To make it clear in the context, we call the LDP version of Haar as HaarHRR.

\mypara{HH in LDP.}
Given a positive integer $\beta$ and a discrete, ordered domain with size $d=|\Domain|$, one can construct a $\beta$-ary tree with $d$ leaves corresponding to values in $\Domain$.  There are $(h+1)$ layers in the tree, where $h=\log_\beta d$ (for simplicity, we assume that $\log_\beta d$ is an integer).  The $(h+1)$-th layer is the root.  
A user with value $v$ chooses a layer $\ell \in \{1, \ldots, h\}$ uniformly at random, and then reports $\ell$ as well as the perturbed value of $v$'s ancestor node at layer $\ell$. 
For each node in the tree, the aggregator can obtain an estimate of its frequency.  
Assuming that the distribution differences among the $h$ groups are negligible, for each parent-child relation, one expects that the sum of child estimations equals the that of the parent.  Constrained inference techniques~\cite{pvldb:HayRMS10} are applied to ensure this property. 

\mypara{HaarHRR.}
Similar to HH, one can use a binary tree to estimate distribution with Discrete Haar Transform~\cite{pvldb:KulkarniCD18}.
Specifically, each leaf represents the frequency of a value.
Define the height of a leaf node as $0$; and the height of an inner nodes $a$ is denotes as $h(a)$.
Each inner node now represents the Haar coefficient $c_a = \frac{C^{(a)}_l-C^{(a)}_r}{2^{h(a)/2}}$, where $C^{(a)}_l$ (or $C^{(a)}_r$) is the sum of all leaves of left (or right) subtree of node $a$.

In the LDP setting, for a user with value $v$, the Haar coefficients on each layer has exactly one element equal to $-1$ or $1$, while others are all zeros.
Similar to HH, each user chooses a layer $\ell \in \{1, \ldots,h\}$ uniformly at random, then apply Hadamard randomized response (HRR) on layer $\ell$ which depends on Hadamard matrix $\phi \in \{-1,1\}^{2^{h-\ell}\times 2^{h-\ell}}$.
With HRR reports from users, the aggregator can calculate unbiased estimates for the Haar coefficients on layer $\ell$.
Due the limit of space, more details can be found in~\cite{pvldb:KulkarniCD18}.

\mypara{Difference from the Centralized Setting.}
When using hierarchy-based method, there are two ways to ensure the privacy constraint.  One is to divide the privacy budget, where one builds a single tree for all values.  
Since each value affects the counts at every level, one splits the privacy budget among the levels.  
The other is to divide the population among the layers, where each value contributes to the estimation of a single layer, and one can use the whole privacy budget for each count.  
When dividing the population, the absolute level of noise is less than the case of dividing privacy budget; however, the total count also decreases, magnifying the impact of noise.  
In addition, dividing the population introduces sampling errors, as users are divided into different groups, which may have different distribution from the global one.

In the centralized setting, because the amount of added noise is low, it is better to divide the privacy budget, as one avoids sampling errors.  In~\cite{pvldb:QardajiYL13}, it was found that in the centralized setting, the optimal branching factor for HH is around 16.  And this results in better performance than using the Haar method, which can be applied only to a binary hierarchy.  In the LDP setting, because the amount of noise is much larger, sampling errors can be mostly ignored, and it is better to divide the population instead of privacy budget.  As a result, the optimal branching factor for HH is around 5, making it similar to the Haar method.  This was theoretically proved and empirically demonstrated in~\cite{pvldb:KulkarniCD18,wang2019answering}.

\subsection{HH-ADMM}

We notice that there are other ways to improve hierarchy-based mechanism in the LDP setting.  First, the larger noise in the LDP setting results in negative estimates.  We can exploit the prior knowledge that the true counts are non-negative to improve the negative estimates.  Second, the total true count is known, as LDP protects privacy of reported values and not the fact that one is reporting.   These are not exploited in \cite{pvldb:KulkarniCD18}. 
We propose to use the Alternating Direction Method of Multipliers (ADMM) algorithm~\cite{book:boyd11distributed} to post-process the hierarchy estimation. The usage of ADMM was proposed in~\cite{kdd:lee2015maximum} for the centralized setting.  Our method applies this to LDP, and has two additional differences from~\cite{kdd:lee2015maximum}.  First, we use $L_2$ norm in the objective function because the noise by \CFO is well approximated by Gaussian noise, and minimizing $L_2$ norm achieves MLE.  In the centralized setting, Laplace noise is used, and $L_1$ norm is minimized in~\cite{kdd:lee2015maximum}.  Second, we pose an additional constraint that the estimates sum up to $n$, which is known in LDP setting.  In the setting considered in~\cite{kdd:lee2015maximum}, $n$ is unknown.

\mypara{The HH-ADMM Algorithm.}
Given a constant vector $\tilde{\x}$, ADMM is an efficient algorithm that aims to find $\estx$ that satisfies the following optimization problem:
\begin{align}
    \text{minimize} & \quad \frac{1}{2}\|\estx - \tilde{\x}\|_2^2 \label{eq:hh-object}\\ 
    \text{subject to}  &\quad \mathbf{A} \estx = 0, \quad \estx \succcurlyeq 0, \quad \estx_0 = 1 \nonumber
\end{align}

In the hierarchy histogram case of LDP, $\tilde{\x}$ represents the concatenation of estimates from all the layers, where $\tilde{\x}_0$ is the root. 
$\estx$ is the post-processed estimates.
The hierarchical constraints state that the estimate of each internal node should be equal to the sum of estimates of its children nodes.  
This can be represented by an equation $\mathbf{A} \estx = 0$, where $\mathbf{A}$ has one row for each internal node and one column for each node, and $a_{ij}$ is defined as:
\begin{align*}
    a_{ij} = \begin{cases}
    1,  & \text{if } i = j\\
    -1, &  \text{node } j \text{ is a child of } \text{node }i\\
    0, &\text{otherwise}
\end{cases}
\end{align*}

The optimization problem~\eqref{eq:hh-object} improves the estimation by enforcing the non-negativity ($\estx \succcurlyeq 0$) and sum-up-to-1 ($\estx_0=1$) compared with~\cite{pvldb:KulkarniCD18}.
Because of the limit of space, we refer the readers who want to know the detail of derivation to ~\cite{kdd:lee2015maximum} for more information.

 \section{Square Wave and Expectation Maximization with Smoothing}
\label{sec:method}

The methods we presented in Section~\ref{sec:adapted_method} use \CFO protocols as black-boxes and do not fully exploit the ordered nature of the domains.
We propose a new approach that uses a Square Wave reporting mechanism with post-processing conducted using Expectation Maximization with Smoothing (EMS).

\subsection{General Wave Reporting}
\label{subsec:gw}

We first study a family of randomized reporting mechanisms that we call General Wave mechanisms.  
The intuition behind this approach is to try to increase the probability that a noisy reported value carries meaningful information about the input.  This is also the implicit goal driving the development of \CFO protocols beyond \grr.  
In \grr, one reports a value in $\Domain$.  
Intuitively, if the reported value is the true value, then the report is a ``useful signal'', as it conveys the extract correct information about the true input.  If the reported value is not the true value, the report is in some sense noise that needs to be removed.  The probability that a useful signal is generated is $p=\frac{e^\epsilon}{e^\epsilon+d-1}$, where $d=|\Domain|$ is the size of the domain.  When $d$ is large, $p$ is small, and \grr performs poorly.  The essence of \olh and other \CFO protocols is that one reports a randomly selected set of values, where one's true value has a higher probability of being selected than other values.  In some sense, each ``useful signal'' is less sharp, since it is a set of values, but there is a much higher probability that a useful signal is transmitted. 

Exploiting the ordinal nature of the domain, we note that a report that is different from but close to the true value $v$ also carries useful information about the distribution.  
Therefore, given input $v$, we can report values closer to $v$ 
with a higher probability than values that are farther away from $v$.  

Without loss of generality, we assume that $\Domain = [0,1]$ consists of floating point numbers between 0 and 1.
The random reporting mechanism can be defined by a family of probability density functions (PDF) over the output domain, with one PDF for each input value.  
We denote the output probability density function for $v$ as $\M_v(\rdv) = \Pr{\perturb(v) = \rdv}$.

Following the above intuition, we want $\M_v(\rdv)$ to satisfy the property that $\M_v(\rdv)=q$ when $|\rdv-v|>b$, and $q \leq \M_v(\rdv) \le e^\epsilon q$ when $|\rdv-v|\le b$, where $b$ is a parameter to be chosen. 
To ensure that for values close to the two ends, the range of near-by values is the same, we enlarge the output domain $\Codomain = [-b, 1+b]$.  
We formalize the idea as the following general wave mechanism.

\begin{definition}[General Wave Mechanism (\GW) ]
With input domain $\Domain=[0,1]$ and output domain $\Codomain=[-b, 1 + b]$, a randomization mechanism $\perturb:\Domain \rightarrow \Codomain$ is an instance of general wave mechanism if for all $v\in \Domain$, there is a wave function $W\colon \mathbb{R} \rightarrow [q,e^\epsilon q]$ with constants $q>0$ and $\epsilon > 0$, such that the output probability density function $\M_v(\rdv) = W(\rdv - v)$ :
\begin{enumerate} 
    \item $W(z) = q$ for $|z| > b$ ;
    \item $\int_{-b}^b W(z)\ dz = 1-q$ \ . 
\end{enumerate}
\end{definition}

\begin{theorem}
\GW satisfies $\epsilon$-LDP.
\end{theorem}
\begin{proof}
 For any two possible input value $v_1, v_2 \in \Domain$ and any set of possible output $T \subseteq \Codomain$ of \GW, we have
\begin{align*}
     \frac{\Pr{\GW(v_1)\in T}}{\Pr{\GW(v_2)\in T}} = \frac{\int_{\rdv\in T}\Pr{\GW(v_1) = \rdv} d\rdv }{\int_{\rdv\in T}\Pr{\GW(v_2) = \rdv}d\rdv} \ .
 \end{align*}
By definition of \GW, for all $v_1,v_2 \in \Domain$ and $T \subset \Codomain$ we have
\begin{align*}
    \frac{\Pr{\GW(v_1)\in T}}{\Pr{\GW(v_2)\in T}} \leq \frac{\int_{\rdv\in T}e^\epsilon q \ d\rdv}{\int_{\rdv\in T}q \ d\rdv} = e^\epsilon \ .
\end{align*}
\end{proof}

\subsection{The Square Wave mechanism}
\label{subsec:sw}

\GW can have different wave shapes.  An intriguing question is what shape should be used. 
Following the same intuition in~\cite{arXiv:AcharyaSZ18}, given different values $v \neq v'$, if $\M_v$ and $\M_{v'}$ are identical, then there is no way to distinguish those different input values. 
Therefore, the hope is that the farther apart $\M_{v}$ and $\M_{v'}$ are, the easier it is to tell them apart.
We use the difference between two output distributions, Wasserstein (a.k.a., earth-mover) distance as the utility metric.
Based on this, we find the Square Wave mechanism, where supports for $[v-b, v+b]$ are the same, is optimal.
We also empirically compare \GW of other shapes with Square Wave mechanism in Section~\ref{subsec:exp_different_shape}.  
The experimental results support our intuition. 

\mypara{Specification of Square Wave Reporting.}
The Square Wave mechanism $\SW$ is defined as: 
\begin{align}
\forall v\in \Domain, \rdv \in \Codomain,  \;\Pr{\SW(v)=\rdv}  \!=\! \left\{
\begin{array}{lr}
\!p, & \mbox{if} \; |v-\rdv|\le b  \ , \\
\!q, & \mbox{ otherwise} \ . \\
\end{array}\label{eq:squarewave}
\right.
\end{align}

By maximizing the difference between $p$ and $q$ while satisfying the total probability adds up to $1$, the values $p, q$ can be derived as:
\begin{align*}
p =  \frac{e^\epsilon}{2be^\epsilon + 1} \ ,\quad
q =  \frac{1}{2b e^\epsilon + 1} \ . 
\end{align*}

For each input $v$, the probability mass distribution for the perturbed output looks like a square wave, with the high plateau region centered around  $v$. We thus call it the Square Wave (\SW) reporting mechanism.

\begin{theorem}\label{thm:optimal_wave}
    For any fixed $b$ and $\epsilon$, the \SW is the \GW that maximizes the Wasserstein distance between any two output distributions of two different inputs.
\end{theorem}
Theorem~\ref{thm:optimal_wave} can be proved by using the following Lemma~\ref{lemma:wass} and Lemma~\ref{lemma:min_q}.

\begin{lemma} \label{lemma:wass}
Given $v_1,v_2 \in \Domain$ as inputs to general wave mechanism, where $v_2 > v_1$ and let $\Delta = v_2 - v_1 > 0$, the Wasserstein distance between the output distributions of general wave mechanism is $\Delta(1 - (2b+1)q)$.
\end{lemma}

\begin{proof}
Given two different input values $v_1$ and $v_2$ which satisfy $v_2 - v_1 = \Delta > 0$, let $\M_{v_1}$ and $\M_{v_2}$ are the corresponding output distributions.
Define a function $\textsc{diff}(z)$ as the following:
\begin{align*}
    \textsc{diff}(z) = 
    & \begin{cases} 
    0\ , & \text{if } z \leq -b \\
    1 - (2b + 1 )q\ , &\text{if } z \geq b \\
    \int_{-b}^{z} (W(z')-q)\ dz' \ , &\text{otherwise.}
    \end{cases}
\end{align*}

The cumulative function of \SW can be written as
\begin{align*}
    \originCDF(\M_{v}, \rdv) = (b + \rdv) q + \textsc{DIFF}(\rdv - v)
\end{align*}
Therefore, 
\begin{align*}
    \int_{-b}^{1+b}\originCDF(\M_{v}, \rdv) d\rdv =& \frac{q}{2}(1+2b)^2  + \int_{-b}^b \textsc{DIFF}(z)dz \\
    & + (1-(2b+1)q)(1-v) \ .
\end{align*}
Following the definition of Wasserstein distance of one dimensional data with $\ell_1$ norm in Section~\ref{sec:problemdef}, and as $\originCDF(\M_{v_1}, \rdv) \geq \originCDF(\M_{v_2}, \rdv)$ for all $\rdv$, it follows that
\begin{align*}
    W_1(\M_{v_1}, \M_{v_2}) &= 
    \int_{\Codomain}|\originCDF(\M_{v_1}, \rdv) - \originCDF(\M_{v_2}, \rdv)| d\rdv \\
  & = \int_{ - b}^{1 + b}\left(\originCDF(\M_{v_1}, \rdv) - \originCDF(\M_{v_2}, \rdv)\right) d\rdv \\
    &= (1-(2b+1)q)\Delta \ .
\end{align*}
\end{proof}

Lemma~\ref{lemma:wass} shows that we need to minimize $q$ if we want to maximize the Wasserstein distance between any two output distributions.
Thus, we have the following lemma.
\begin{lemma}
    \label{lemma:min_q}
    For any fixed $b$ and $\epsilon$, the minimum $q$ for general wave mechanism is $q = \frac{1}{2be^\epsilon  + 1}$, which can be achieved if and only if the mechanism is \SW.
\end{lemma}

\begin{proof}
By criteria of the definition of \GW, we have 
    \begin{align*}
        & 1 = q+\int_{-b}^b W(z)dz \leq 1+(2b)e^\epsilon q \\
        & \Rightarrow q \geq \frac{1}{2be^{\epsilon}+1}
    \end{align*}
    We have equality iff $\M_v(\rdv) = e^{\epsilon}q$ for all $\rdv \in[v-b, v+b]$, which turns out to be \SW.
\end{proof}

\mypara{Comparison with \PM Mechanism.}
Square Wave (\SW) reporting is similar to the Piecewise Mechanism (\PM) for mean estimation~\cite{icde:WangXYHSSY18} (see Section~\ref{sec:background:mean}).  \PM directly sums up the randomized reports to estimate the mean of distribution, while the outputs of \SW are used to reconstruct the whole distribution  (the reconstruction method will be described in Subsection~\ref{subsec:em}).
Driven by the different focus, the reporting mechanisms are also different. 
\PM has to be unbiased for mean estimation, so the input values are not always at the center of high probability region.  
For example, given input $v=-1$, the high probability range in \PM is $[-\frac{e^{\epsilon/2} + 1}{e^{\epsilon/2}-1}, -1]$.

\subsection{Choosing \textit{b}}
\label{subsec:choose_b}
An important parameter to choose for the Square Wave reporting mechanism is $b$. 
In Square Wave reporting, a value that is within $b$ of true input is reported with a probability that is $e^\epsilon$ times the probability that a ``far'' value is reported.  
The optimal choice of $b$ depends on the privacy parameter $\epsilon$.  For a larger $\epsilon$, a smaller $b$ is preferred.  When $\epsilon$ goes to infinity, a value of $b\rightarrow 0$ leads to total recovery of input distribution, and any $b>0$ leads to information loss.  
Intuitively, the optimal choice of $b$ also depends on the input distribution.  For a distribution with probability density concentrated at one point, one would prefer smaller $b$.  For a distribution with more or less evenly distributed probability density, one would prefer a larger $b$.  However, since we do not know the distribution of the private values, we want to choose a $b$ value independent of the distribution, but can perform reasonably well over different distributions. 

In this paper, we choose $b$ to maximize the upper bound of mutual information between the input and output of the Square Wave reporting.  We also empirically study the effect of varying $b$ (see Section~\ref{subsec:exp_different_shape}).  The experimental results show that choosing $b$ by this method results in optimal or close to optimal choices of $b$.

Let $V$ and $\RV$ be the input and output random variables representing the input and output of \SW, respectively.
The mutual information between $V$ and $\RV$ can be represented by the difference between differential entropy and conditional differential entropy of $V$ and $\RV$:
\begin{align*}
    I(V, \RV) = h(V) - h(V|\RV) = h(\RV) - h(\RV|V) \ .
\end{align*}
The quantity $I(V, \RV)$ depends on the input distribution, which we want to avoid.  Therefore, we consider an upper bound of $I(V, \RV)$, which is achieved when $\RV$ is uniformly distributed on $\Codomain$.  Let $U$ be the random variable that is uniformly distributed in $\Codomain$. 
Because $h(\RV)\leq h(U)$, we have:
\begin{align}
    I(V, \RV) \leq h(U) - h(\RV|V) . \label{eq:mi_bound}
\end{align}
In~\eqref{eq:mi_bound}, the first term of RHS is 
\begin{align*}
    h(U) = \log(2b+1) .
\end{align*}
The second term of RHS only depends on \SW:
\begin{align*}
    \begin{split}
\quad h(\RV|V)    &= -\int_{v}\Pr{V = v} (2bp\log p + q \log q) \\
    & = -(2bp\log p + q \log q)\\
    & = - \frac{2b\epsilon e^\epsilon}{2b e^\epsilon + 1} + \log (2be^\epsilon + 1) \ .
    \end{split}
\end{align*}
So the mutual information is determined by a function of $b$,
\begin{align*}
    \log\left(\frac{2b+1}{2be^\epsilon + 1} \right) + \frac{2b\epsilon e^\epsilon }{2b e^\epsilon + 1}  \ .
\end{align*}
By making its derivative to 0, we get
\begin{align*}
    b = \frac{\epsilon e^\epsilon - e^\epsilon +1}{2e^\epsilon(e^\epsilon - 1 - \epsilon)} \ .
\end{align*}
Note that $b$ is a non-increasing function with $\epsilon$.  When $\epsilon$ goes to $\infty$, $b$ goes to $0$. When $\epsilon$ goes to $0$,  $b$ goes to $1/2$, which leads to an output domain that doubles the size of the input domain, and for each input value, half of the output domain are considered ``close'' to the input value.

\subsection{Bucketizing}
\label{subsec:bucketize}

The aggregator receives perturbed reports from users and needs to reconstruct the distribution on $\Domain$.  Our approach performs this reconstruction on a discretized domain, i.e., histograms over the domain.  
The bucketization step can be performed either before or after applying the randomization step.  We discuss the two approaches below.  In experiments, we use the ``randomize before bucketize'' approach.

\mypara{``Randomize before bucketize'' (R-B).}
Here each user possesses a floating point number in $\Codomain = [0, 1]$, applies the Square Wave mechanism  in Section~\ref{subsec:sw}, and sends the result to the aggregator.  The aggregator receives values in $\Codomain = [-b, 1+b]$, discretizes the reported values into $\tilde{d}$ buckets in $\Codomain$, and constructs a histogram with $\tilde{d}$ bins.  Using the method in Section~\ref{subsec:em}, the aggregator can reconstruct an estimated input histogram of $d$ bins.   In experiments, we set $\tilde{d}=d$ for simplicity. 

We compare the results of choosing different $\tilde{d}$ in Section~\ref{subsec:exp_bucketize}, and found that the results are similar so long as $\tilde{d}$ does not deviate far from $\sqrt{N}$.

\mypara{``Bucketize before randomize'' (B-R) or discrete input domain.}
Alternatively, a user can perform the discretization step first, and then perform randomization.  The \SW mechanism can be naturally applied in a discrete domain as well. 
Assume input domain size is $d = |\Domain|$, discrete \SW mechanism has output domain size $\tilde{d} = |\Codomain| = d + 2b$, and randomizes input values as the following:
\begin{align*}
\forall v\in \Domain, \rdv \in \Codomain,  \;\Pr{\SW(v)=\rdv}  \!=\! \left\{
\begin{array}{lr}
\!p, & \mbox{if} \; |v-\rdv|\le b  \\
\!q, & \mbox{ otherwise, } \\
\end{array}\label{eq:squarewave_discrete}
\right.
\end{align*}
where $p = \frac{e^{\epsilon}}{(2b+1)e^{\epsilon} + d -1}$ and $q = \frac{1}{(2b+1)e^{\epsilon} + d -1}$.
In this case, one can set $b = \left\lfloor \frac{\epsilon e^\epsilon - e^\epsilon +1}{2e^\epsilon(e^\epsilon - 1 - \epsilon)} d \right\rfloor$.

The above discrete \SW mechanism can also be applied when the input domain is already discrete (e.g., age). 
We conducted experiments comparing doing R-B versus B-R, and found that they are very similar.  Detailed results are omitted due to space limitation.

\subsection{Estimating Distribution from Reports}\label{subsec:em}
The aggregator receives perturbed values and faces an estimation problem.  Without relying on any prior knowledge of the actual distribution, the natural approach is to conduct Maximum Likelihood Estimation (MLE).  
We use a $\tilde{d} \times d$ matrix $\M$ to characterize the randomization process. 
More specifically, the matrix $\M \in [0,1]^{\tilde{d} \times d}$ denotes the transformation probabilities, where $\M_{i, j}$ represents the probability of output value falling in bucket $\tilde{B}_j$, $j \in [\tilde{d}]$, given input in bucket $B_i$, $i \in [d]$, (assuming the input data fall uniformly at random within bucket $B_i$). Each column of $\M$ sums up to 1. 

\mypara{Expectation-Maximization (EM) Algorithm.}
Given the probability matrix $\M$ as defined above, we can use an Expectation-Maximization (EM) algorithm to reconstruct the distribution.
The aggregator receives $n$ randomized values from users, which are denoted as $\rdvv = \{\rdv_1, \ldots, \rdv_n\}$,  
and finds $\estx$ that maximizes the log-likelihood $L(\estx)= \ln \Pr{\rdvv | \estx}$.

Let $n_{j}$ be the number of values in $\tilde{B}_j$ is reported.
The EM algorithm for post-processing the square wave reporting is shown in Algorithm~\ref{algo:EM}.  
Note that there are existing works that use EM to post-process results of \CFO  (e.g.,~\cite{popets:FantiPE16,tifs:RenYYYYMY18}), but our proposed EM algorithm takes aggregated results and is thus more efficient.
Because of limitation of space, we omit the derivation of EM algorithm.

\begin{algorithm}[t]
\begin{algorithmic}
\STATE \textbf{Input:} $\M, \noisy$
\STATE \textbf{Output:} $\estx$
\WHILE{not converge}
\STATE E-step: $\forall i \in \{1, ..., d\}$,
        \begin{align*}
            P_i &= \estx_i \sum_{j \in [\tilde{d}]} n_{j}   \frac{\Pr{ \rdv \in \tilde{B}_j | v \in B_i , \estx}}{\Pr{\rdv \in \tilde{B}_j |\estx}}  \\
            &= \estx_i \sum_{j \in [\tilde{d}]} n_{j}   \frac{\M_{ j, i}}{\sum_{k=1}^{d} \M_{ j, k}\estx_k}
        \end{align*}
    \STATE M-step: $\forall i \in \{1, ..., d\}, $
        \begin{align*}
            \estx_i =  \frac{P_i}{\sum_{k'=1}^{d}P_{k'}}
        \end{align*}
\ENDWHILE
\STATE Return $\estx$
\end{algorithmic}
\caption{Post-processing EM algorithm }
\label{algo:EM}
\end{algorithm}

\begin{theorem}
The EM algorithm converges to the maximum-likelihood (ML) estimator of the true frequencies $\x$.
\label{thm:em_ml}
\end{theorem}
\begin{proof}
\label{proof:em_ml}
To prove EM algorithm converges to the maximum likelihood estimator, it is enough to show the log-likelihood function is concave~\cite{bilmes1998gentle}.
In the context of our problem 
\begin{align*}
    L(\x) &= \ln \Pr{\rdvv | \x} = \ln \prod_{k=1}^{n}\Pr{\rdv_k |\x } \\
    &=\sum_{k=1}^{n}\ln \left(\sum_{i=1}^{d} \x_i\Pr{\rdv_k |v\in B_i } \right),
\end{align*}
where $\Pr{\rdv_k |v\in B_i}$ are constants determined by \SW method.
Thus, $L(\x)$ is a concave function.
\end{proof}

\mypara{Stopping Criteria.}
Through experiments, we have observed that the result of applying EM is highly sensitive to the parameter controlling terminating condition.  If EM terminates too early, the reconstructed distribution is still far from the true one.  If EM terminates too late, while the reconstructed distribution does fit the observation better (higher likelihood), it is also getting farther away from the true distribution to fit the noise.  
One of the most common stopping criteria for EM algorithm is checking whether the relative improvement of log-likelihood is small~\cite{popets:FantiPE16}. 
Namely, when $\left | L(\estx\iter{t+1}) - L(\estx\iter{t}) \right|< \tau$ for some small positive number $\tau$, EM algorithm stops.
The choice of $\tau$ depends on many factors, including the smoothness of distribution and the amount of noise added by the square wave distribution.
Empirically, we find that if we set $\tau$ proportional to $e^\epsilon$, EM algorithm generally performs better than the one using a fixed $\tau$.
However, on some datasets that have a smoother distribution, the recovered result still over-fits the noise.  Several of our attempts at finding a stopping condition that make EM perform well consistently did not succeed.  This motivates us to apply smoothing in EM.

\mypara{EMS Algorithm.}
By the nature of numerical domain, adjacent numerical values' frequencies should not vary dramatically. 
With this observation, we can add a smoothing step after the M-step in the EM algorithm.
We call the EM algorithm with smoothing steps as EMS algorithm.
The idea of adding smoothing step into EM algorithm dates back to 1990s~\cite{nychka1990some, silverman1990smoothed} in the context of  positron emission tomography and image reconstruction.
The authors showed that a simple local smoothing method, the weight average with binomial coefficients of a bin value and the values of its nearest neighbours, could improve the estimation dramatically.  We adopt this smoothing method. 
That is, after the M-step, the smoothing step will average each estimate with its adjacent ones with binomial coefficients $(1, 2, 1)$:
\begin{align*}
    \estx_i = \frac{1}{2}\estx_i + \frac{1}{4}\left(\estx_{i-1} + \estx_{i+1}\right).
\end{align*}
It was proved that adding the smoothing step is equivalent to adding a regularization term penalizing the spiky estimation~\cite{nychka1990some}, which can be viewed as applying Bayesian inference with a prior that prefers smoother distribution to jagged ones~\cite{ormoneit1998averaging}.
In more recent work, the idea of EMS is also applied to spatial data~\cite{fan2011local} and biophysics data~\cite{huys2009smoothing}.

\begin{figure*}[h!]
    \centering
    \begin{subfigure}[b]{0.23\textwidth}
	    \label{fig:sample_beta}
		\includegraphics[width=\textwidth]{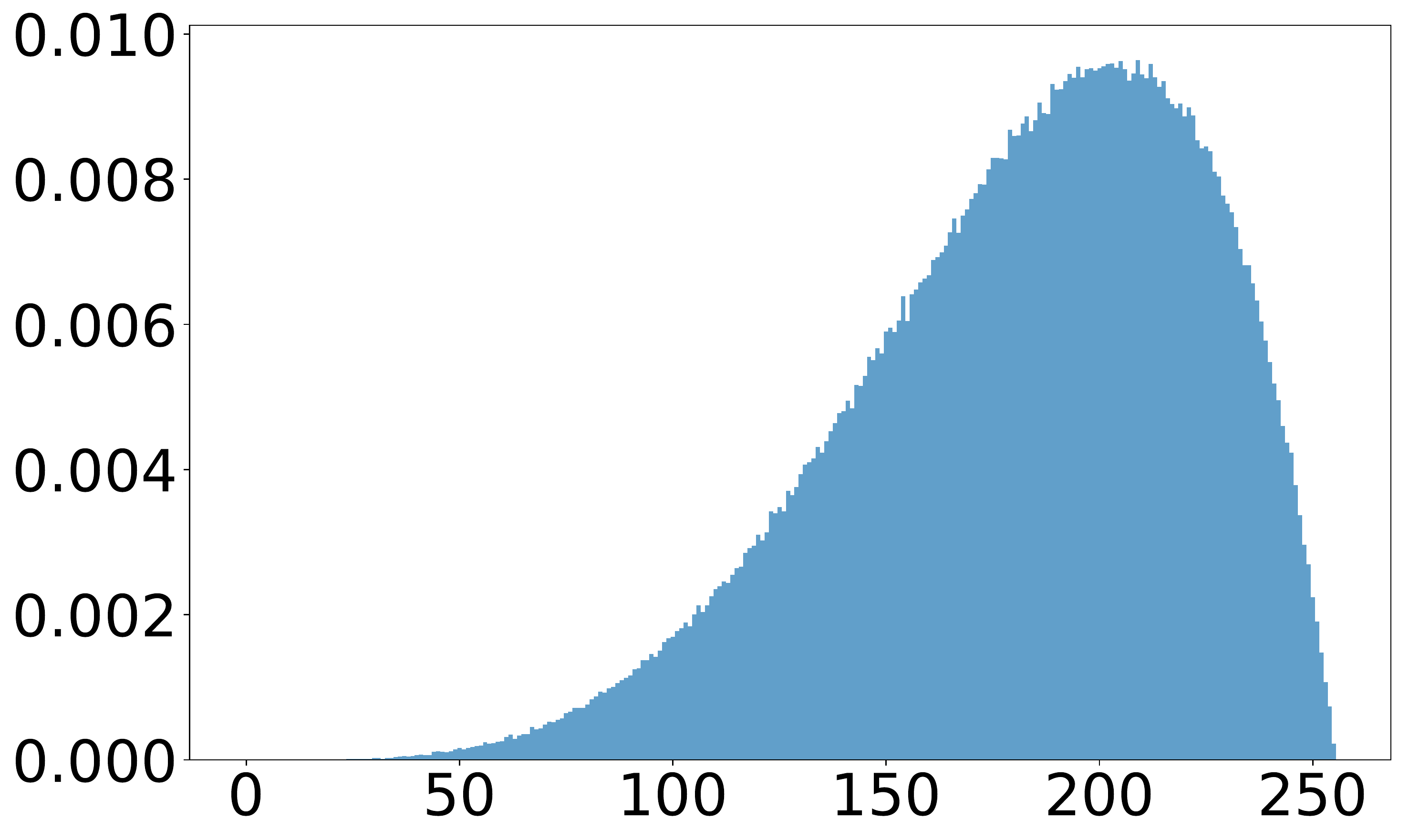}
		\vspace{-0.6cm}
		\caption{Beta(5, 2)}
	\end{subfigure}
	 \begin{subfigure}[b]{0.23\textwidth}
	    \label{fig:sample_PT}
		\includegraphics[width=\textwidth]{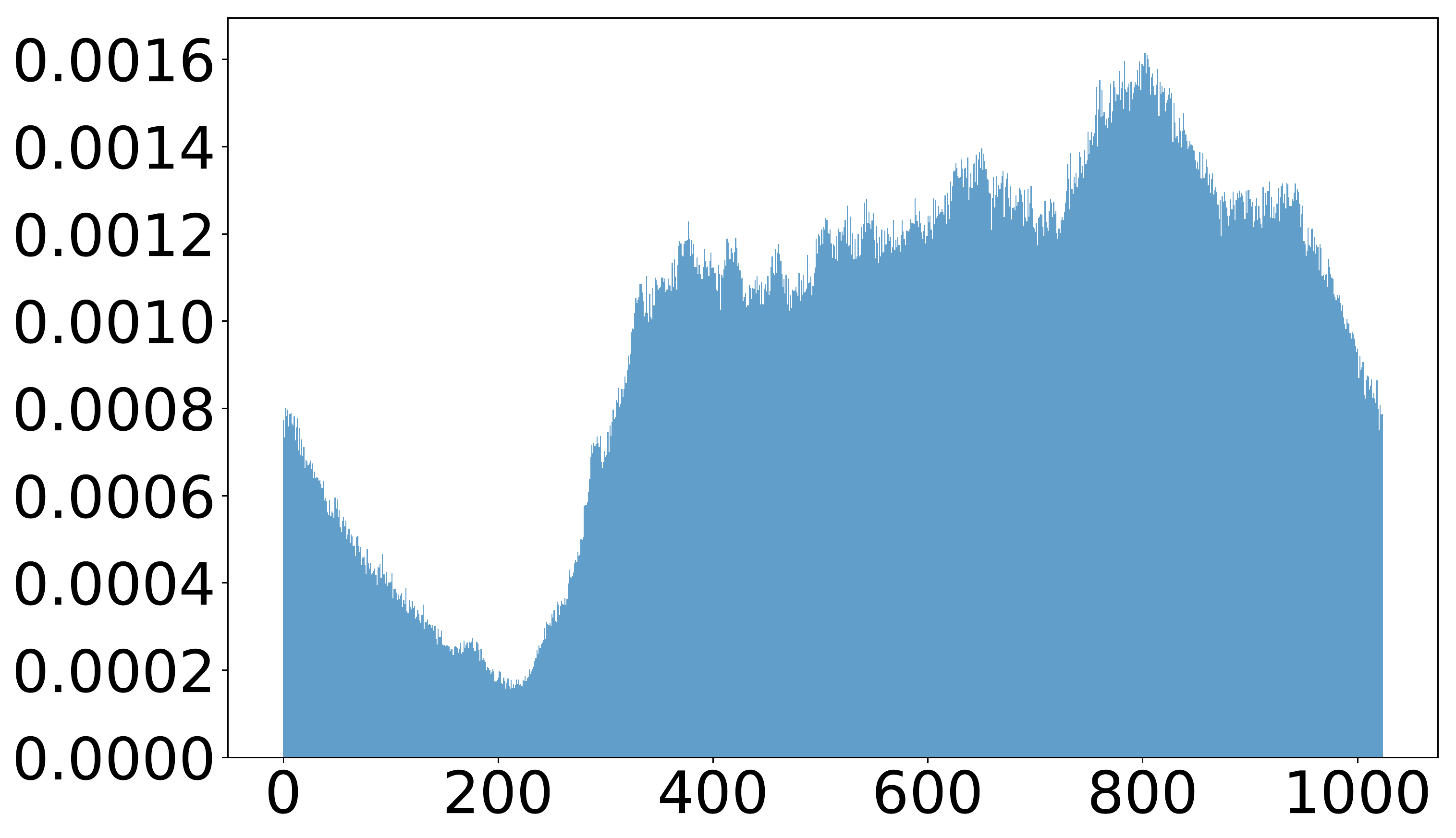}
		\vspace{-0.6cm}
		\caption{Taxi pickup time}
	\end{subfigure}
	 \begin{subfigure}[b]{0.23\textwidth}
	    \label{fig:sample_INC}
		\includegraphics[width=\textwidth]{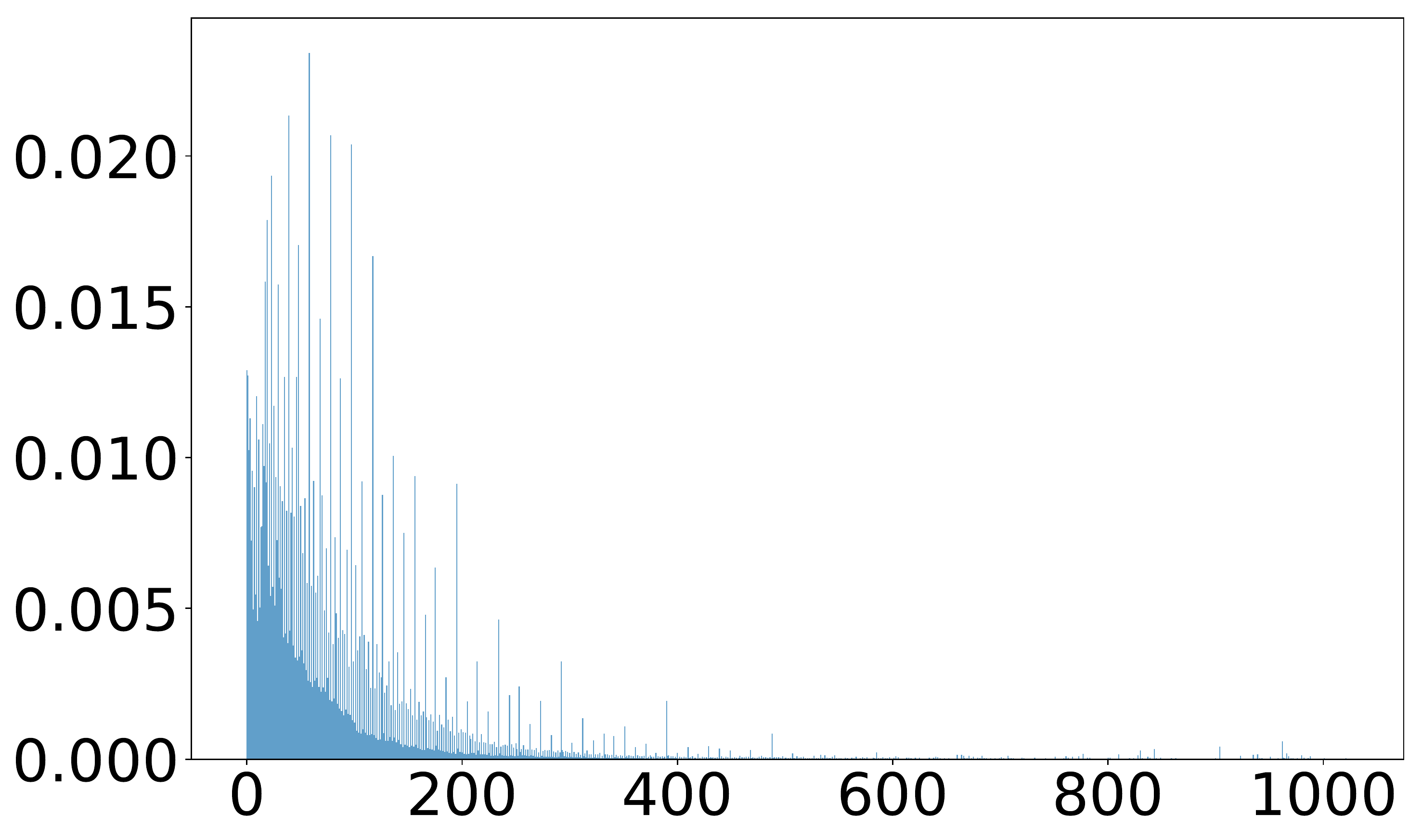}
		\vspace{-0.6cm}
		\caption{Income}
	\end{subfigure}
	 \begin{subfigure}[b]{0.23\textwidth}
	    \label{fig:sample_RT}
		\includegraphics[width=\textwidth]{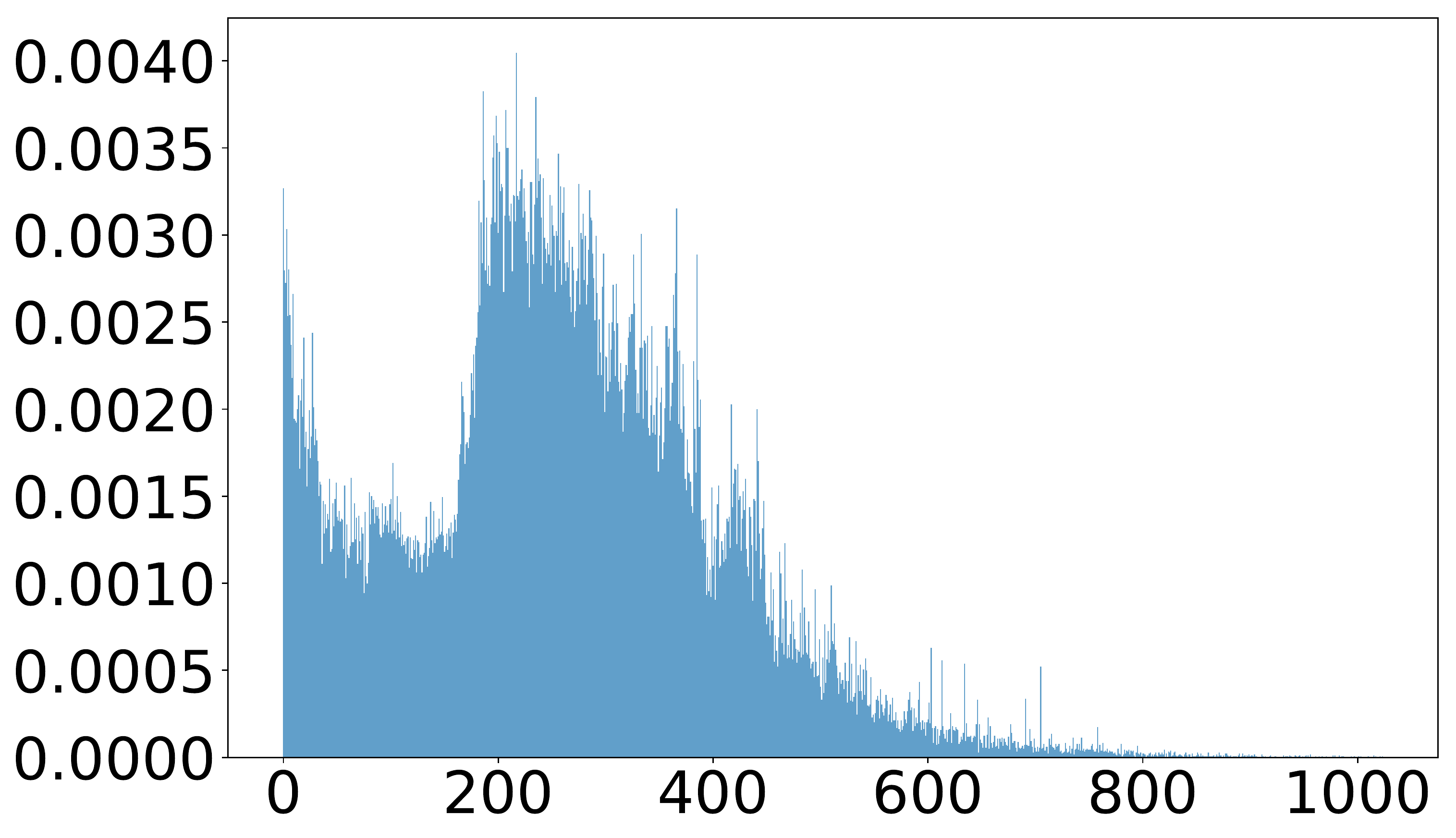}
		\vspace{-0.6cm}
		\caption{Retirement}
	\end{subfigure}\\
	\vspace{-0.4cm}
	\caption{
	Normalized frequencies of datasets for experiments.
	}
	\label{fig:samples}
\end{figure*}

\section{Experiments}
\label{sec:experiments}

\subsection{Experimental Setup}
\mypara{Datasets.}
We use the following datasets to conduct our experiments. 
One of them is synthetic, and the other three are real world datasets.
All of them consist of numerical values.
For \CFO based methods, we discretize the values to the same granularity as the output of \SW with EMS/EM method.
Also, in order to compare with HH and HH-ADMM, which have optimal branching factor close to 4~\cite{pvldb:KulkarniCD18}, we choose the granularity (number of buckets in histogram) to be power of 4.

\emph{Synthetic Beta}(5, 2) dataset.  Originally, the distribution is in the continuous domain $[0,1]$.  
    One hundred thousand samples are generated.
In experiments, we reconstruct the histogram with 256 buckets for all methods.
    
\emph{Taxi} dataset's attribute pick-up time. 
    Taxi pickup time dataset comes from 2018 January New York Taxi data~\cite{data:pt}.
    Originally, the dataset contains the pickup time in a day (in seconds).
    We map the values into $[0,1]$.
    There are $2,189,968$ samples in the dataset.
In experiments, all estimated histograms have 1024 buckets.
    
\emph{Income} dataset.
    We use the income information of the 2017 American Community Survey~\cite{data:ipums}.
    The data range is $[0, 1563000)$.
    We extract the values that are smaller than $524288$ (i.e., $2^{19}$) and map them into $[0,1]$.
There are $2,308,374$ samples after pre-processing.
    We choose to set the estimated histograms with 1024 buckets.

\emph{Retirement} dataset.
    The San Francisco employee retirement plans data~\cite{data:rt} contains integer values from $-28,700$ to $101,000$.
    We extract values that are non-negative and smaller than $60,000$, and map them into $[0,1]$.
There are $178,012$ samples after post-processing.
    In experiments, we reconstruct the histogram with 1024 buckets for all methods.

The income dataset is spiky because many people tend to report with precision up to hundreds or thousands (e.g., people are more likely to report \$3000 instead of more precise value like \$3050 or \$2980.)

\begin{figure*}[h!]
    \centering
    \begin{subfigure}[b]{\textwidth}
		\includegraphics[width=1\textwidth]{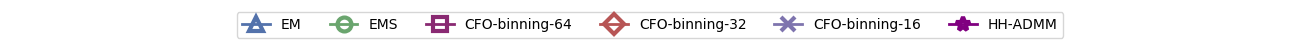}
		\vspace{0.00mm}
	\end{subfigure}\\
	\vspace{-0.5cm}
    \begin{subfigure}[b]{0.23\textwidth}
         \centering
         \includegraphics[width=\textwidth]{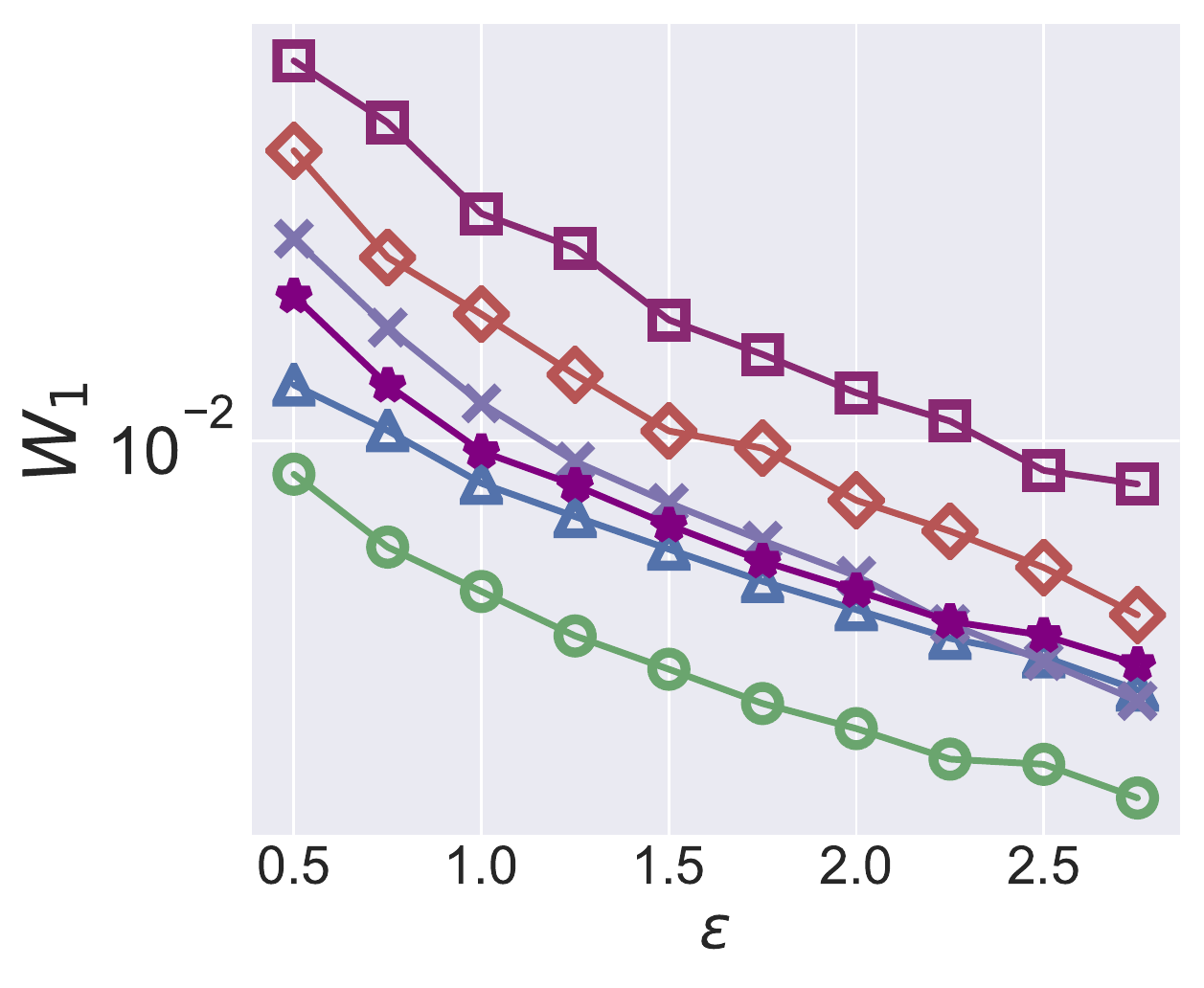}
         \vspace{-0.8cm}
         \caption{Beta(5, 2)}
         \label{wass_beta}
    \end{subfigure}
\begin{subfigure}[b]{0.23\textwidth}
		\includegraphics[width=\textwidth]{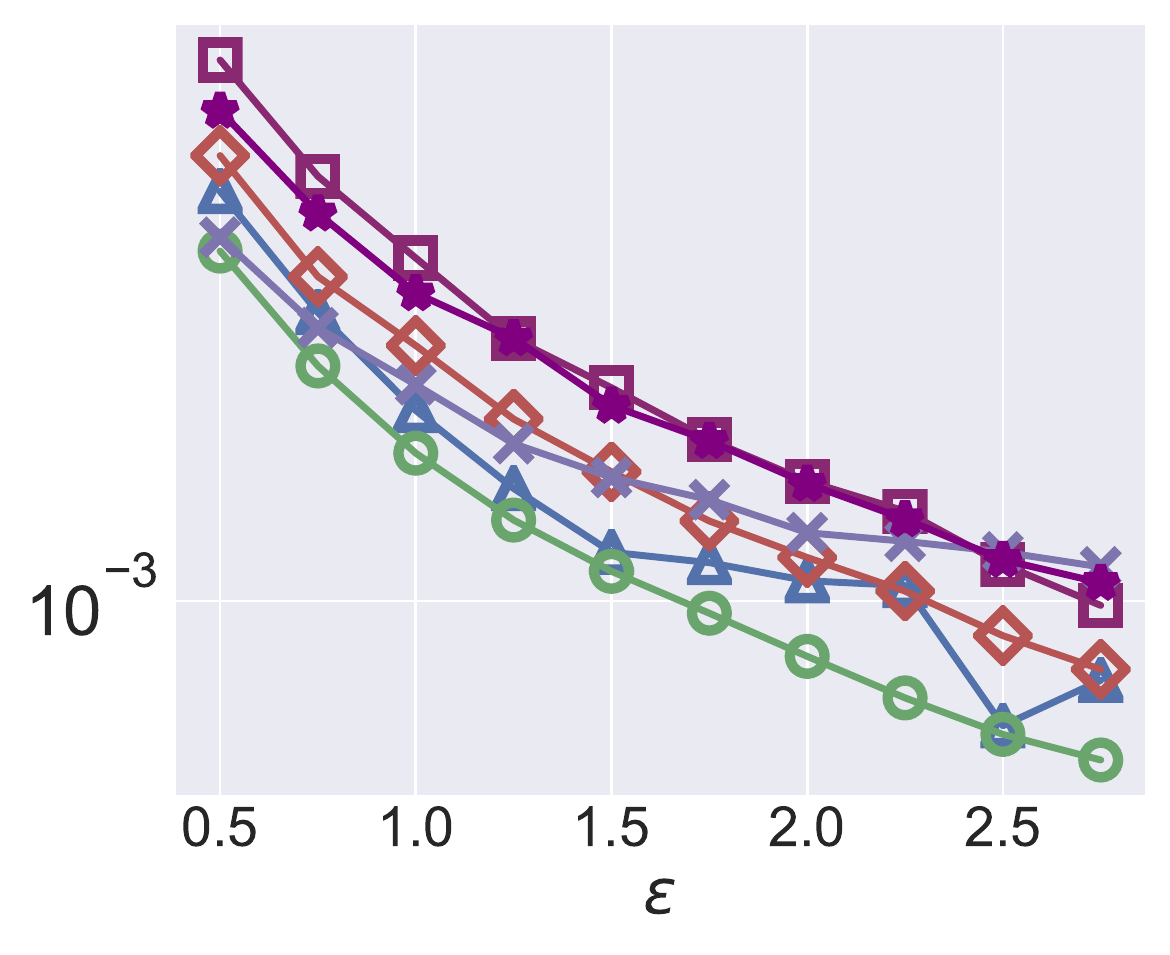}
		\vspace{-0.8cm}
		\caption{Taxi pickup time}
		\label{wass_PT}
	\end{subfigure}
\begin{subfigure}[b]{0.23\textwidth}
		\includegraphics[width=\textwidth]{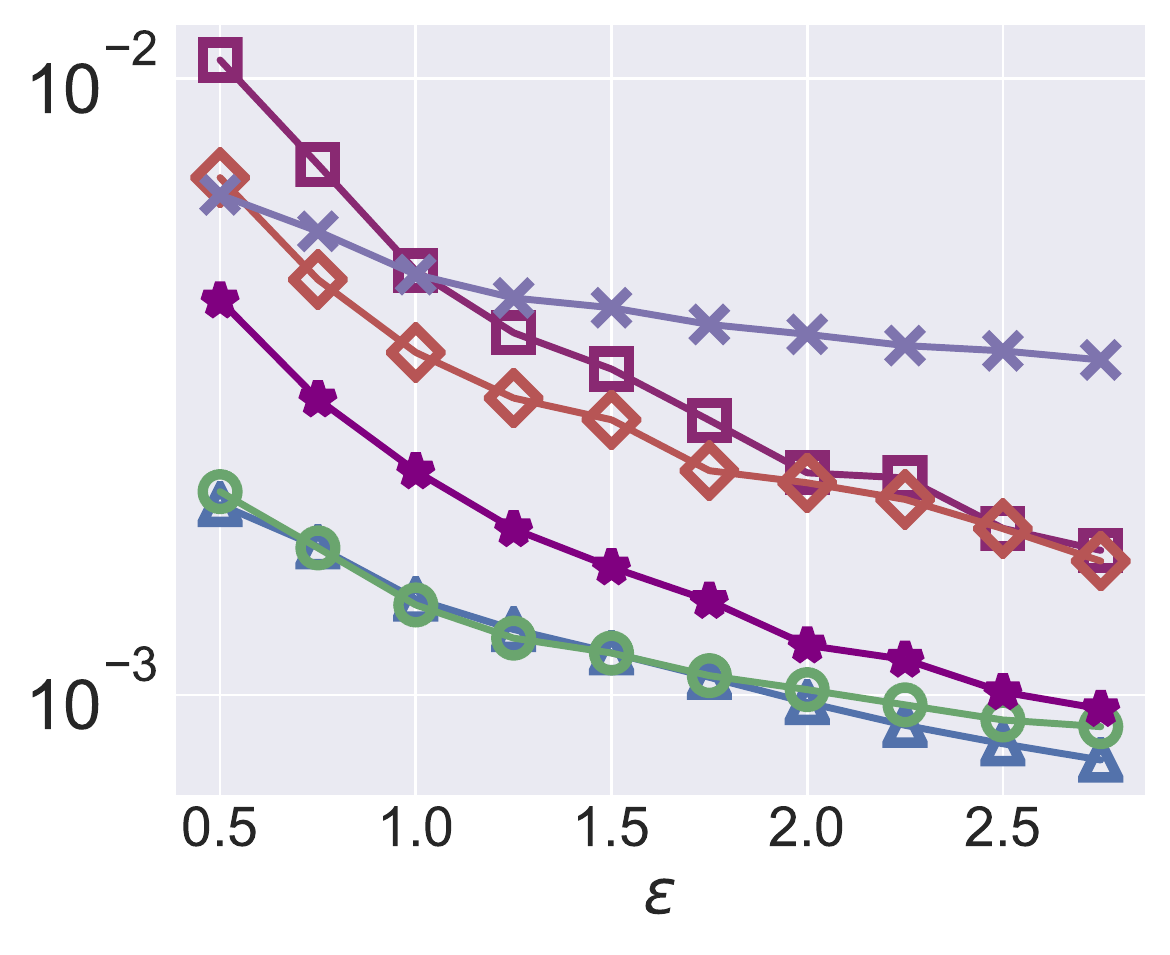}
		\vspace{-0.8cm}
		\caption{Income}
		\label{wass_INC}
	\end{subfigure}
\begin{subfigure}[b]{0.23\textwidth}
		\includegraphics[width=\textwidth]{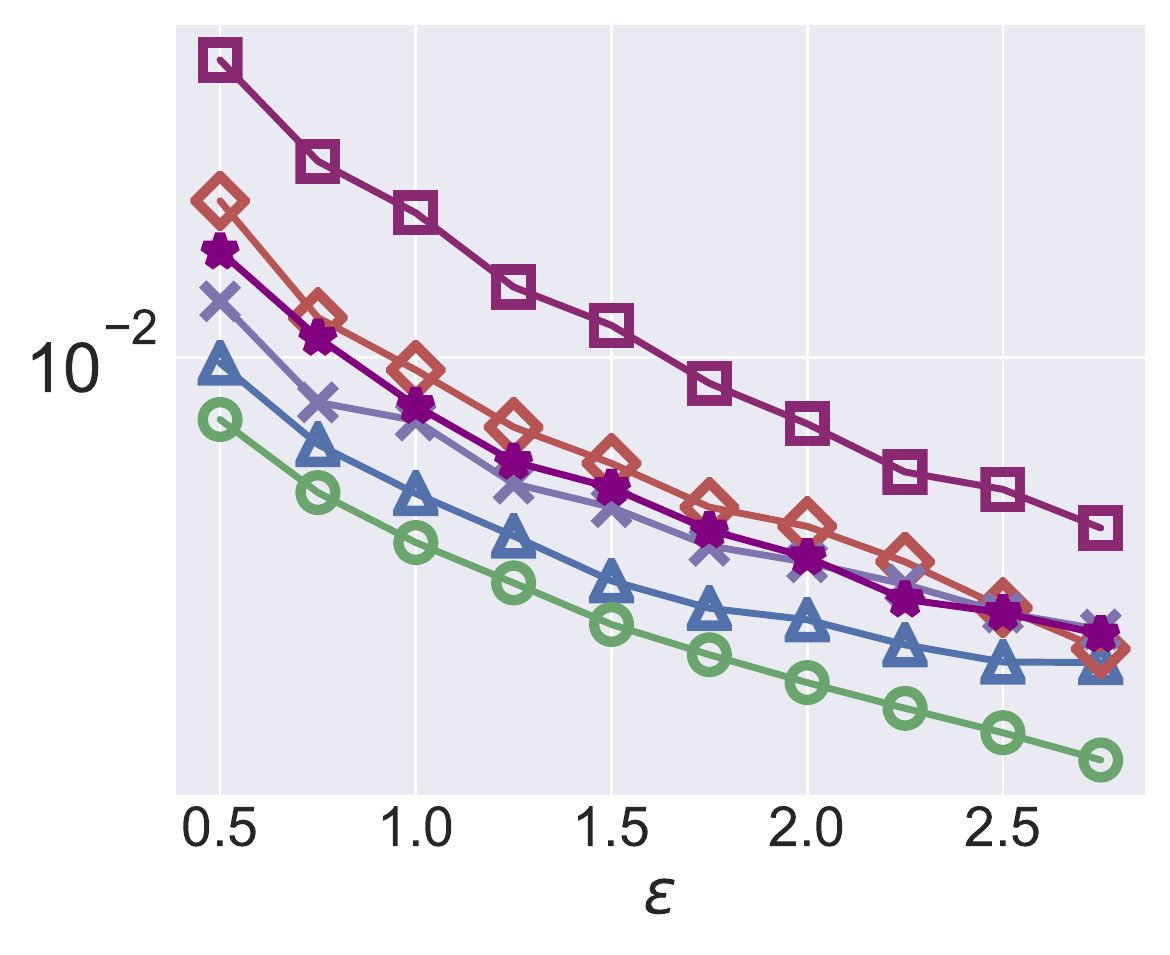}
		\vspace{-0.8cm}
		\caption{Retirement}
		\label{wass_RT}
	\end{subfigure}\\
\begin{subfigure}[b]{0.23\textwidth}
		\includegraphics[width=\textwidth]{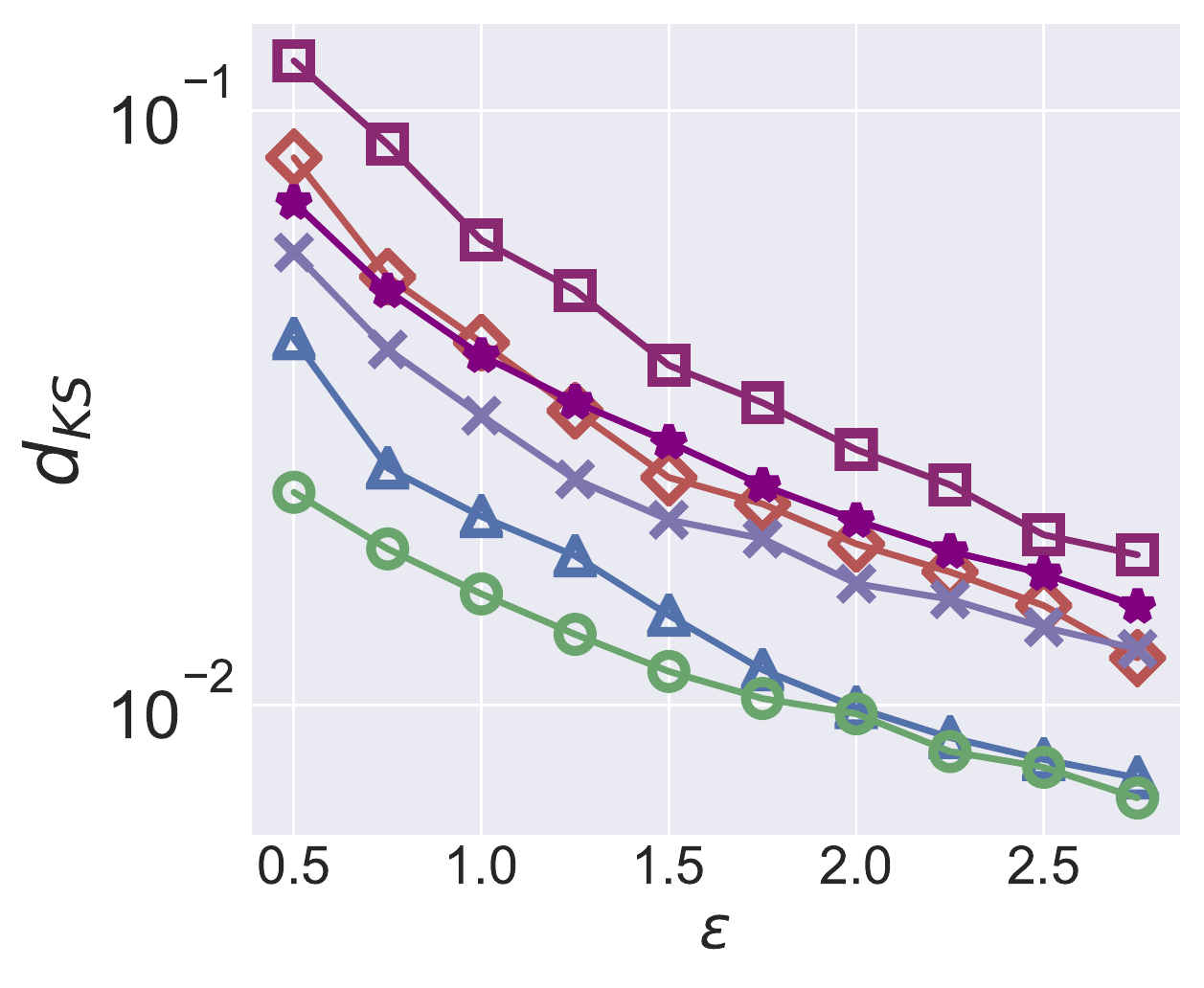}
		\vspace{-0.8cm}
		\caption{Beta(5,2)}
		\label{ks_beta}
	\end{subfigure}
\begin{subfigure}[b]{0.23\textwidth}
		\includegraphics[width=\textwidth]{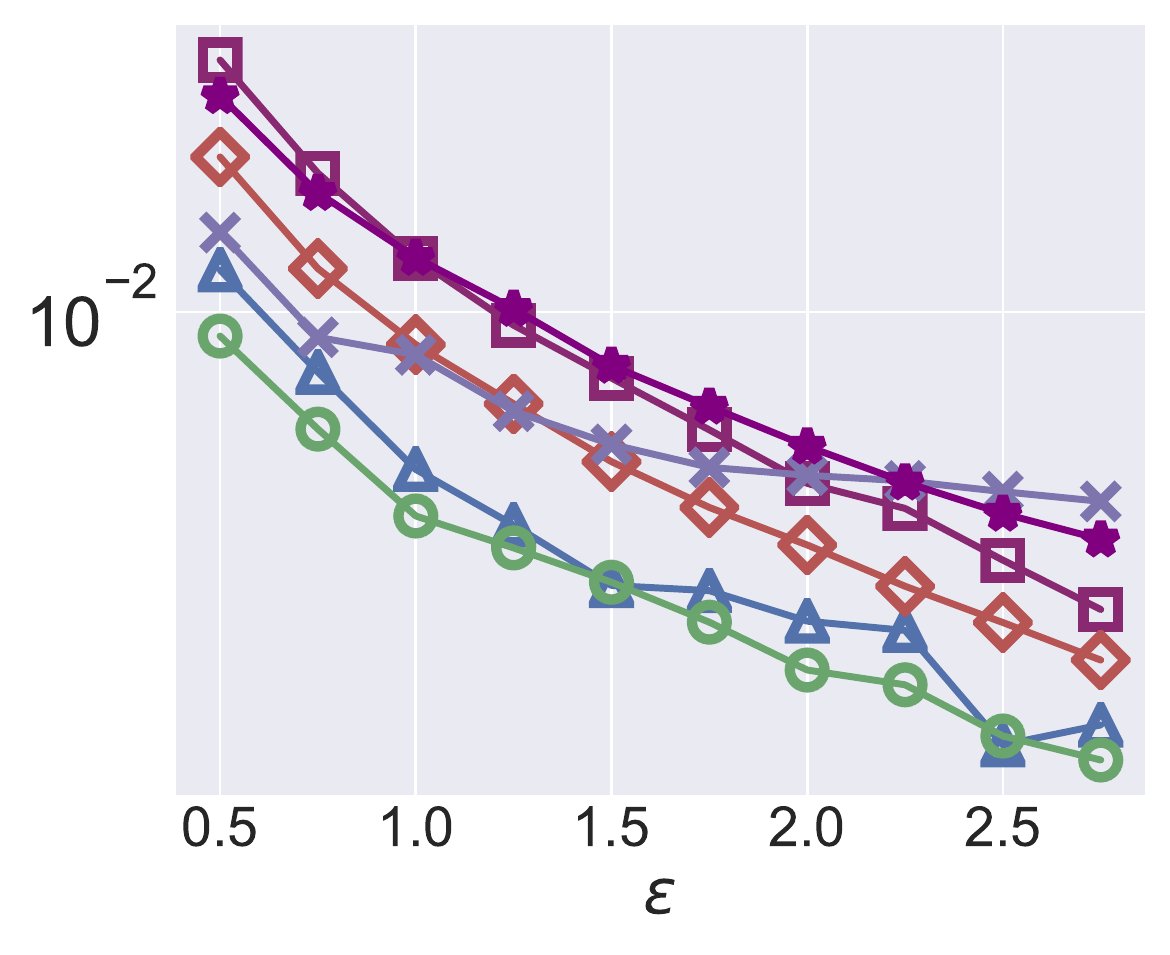}
		\vspace{-0.8cm}
		\caption{Taxi pickup time}
		\label{ks_PT}
	\end{subfigure}
\begin{subfigure}[b]{0.23\textwidth}
		\includegraphics[width=\textwidth]{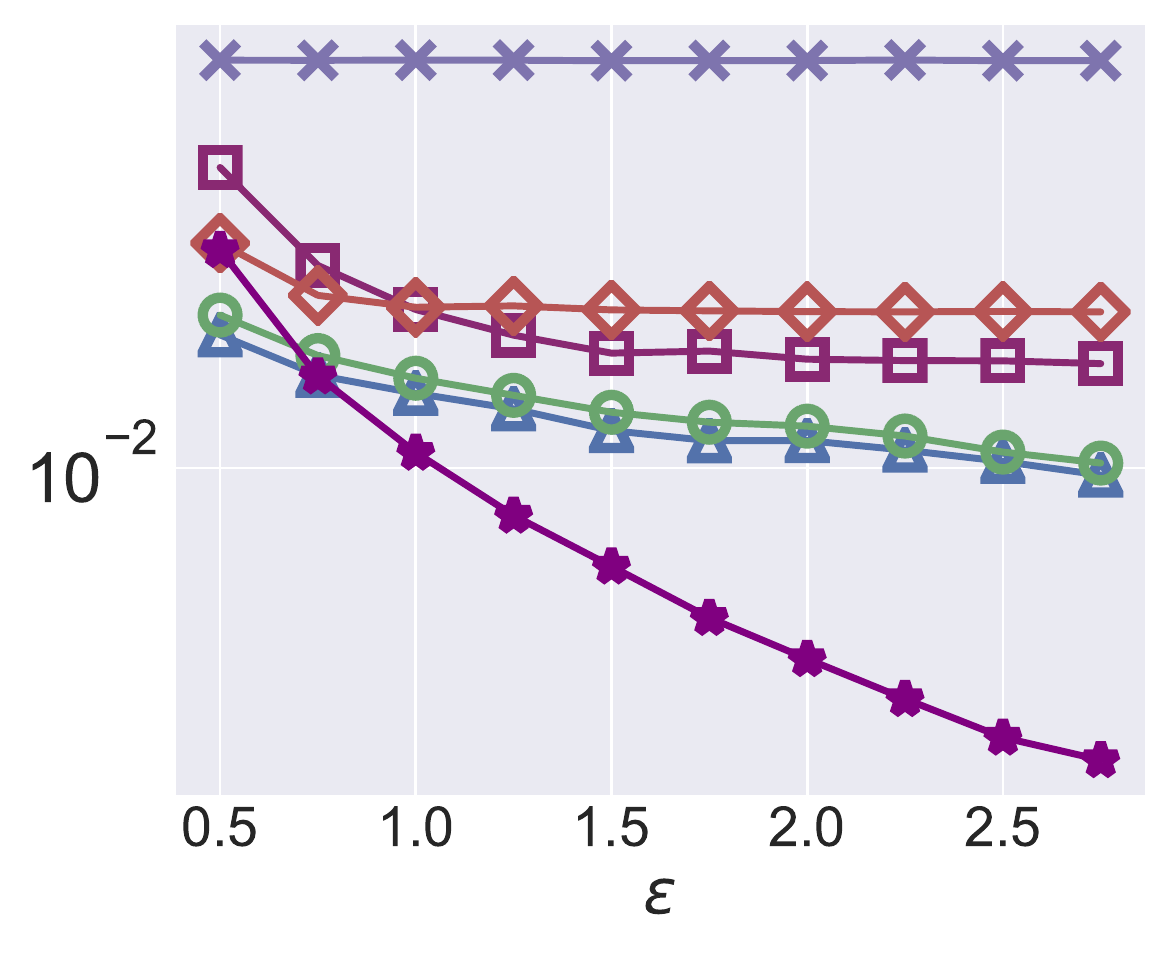}
		\vspace{-0.8cm}
		\caption{Income}
		\label{ks_INC}
	\end{subfigure}
\begin{subfigure}[b]{0.23\textwidth}
		\includegraphics[width=\textwidth]{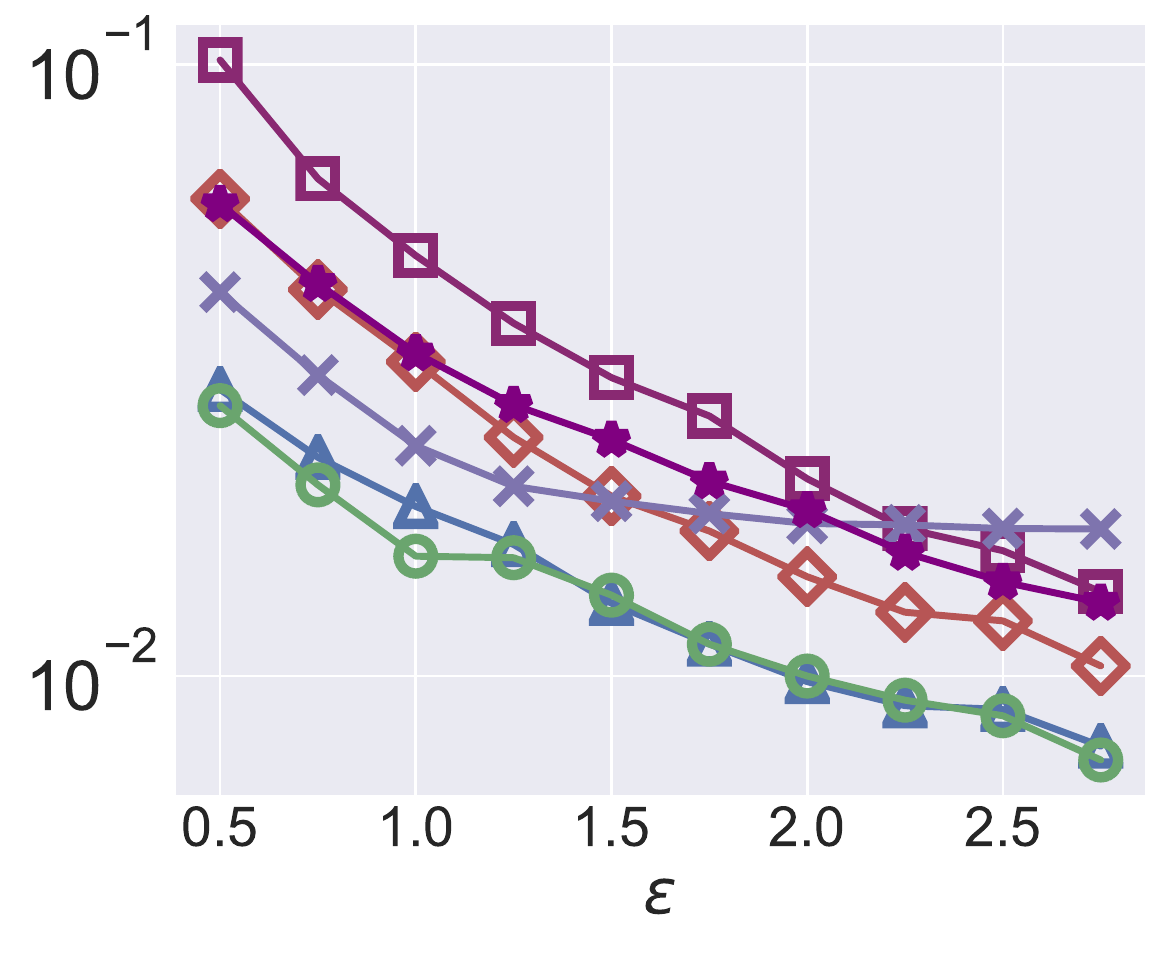}
		\vspace{-0.8cm}
		\caption{Retirement}
		 \label{ks_RT}
	\end{subfigure}\\
	\vspace{-0.3cm}
	\caption{
    Results of distribution distances (first row: Wasserstein distance, second row: KS distance), varying $\epsilon$.
    }
	\label{fig:distribution_distance}
	\vspace{-0.3cm}
\end{figure*}

\mypara{Competitors.}
In the experiments, we consider several existing methods, including methods that obtain mean (\PM, \SR) and Hierarchy-based Methods (HH, HaarHRR).  We also consider \CFO with binning methods, our proposed method HH-ADMM, and \SW with EMS/EM.
To the more specific, we summarize the methods and metrics evaluated in Table~\ref{tbl:method_metric}.

\begin{itemize}
    \item Piecewise Mechanism (\PM) and Stochastic Rounding (\SR) (See Section~\ref{sec:background:mean}) are only evaluated for mean and variance.  They were designed for mean, and we adapted them to also estimate variance. 
    \item For \CFO with binning, we partition $\Domain$ into $c$ consecutive, non-overlapping chunks.  We consider $c= 16, 32, 64$, which are the best performing $c$ values.
    \item For HH, HaarHRR and HH-ADMM, similar to~\cite{pvldb:KulkarniCD18}, we use a branching factor of $4$.  HH and HaarHRR are only evaluated for range queries as they produce estimation results with negative values, which are not valid probability distributions.  Other metrics are defined for probability distributions. 
    \item For \SW with EM and EMS as post-processing, we set $\tau = 10^{-3}e^\epsilon$ for EM and $\tau = 10^{-3}$ for EMS.
\end{itemize}

\begin{table}
\centering
\normalsize
\begin{center}
\resizebox{0.48\textwidth}{!}{\begin{tabular}{|p{3cm}|p{1.5cm}|p{1.1cm}|p{1.3cm}|p{1.1cm}|}
        \hline
        \multirow{3}{*}{\backslashbox{\quad \\ Methods }{Metrics\\}} & Wasserstein and KS distance & Range Query & Mean \& Variance & Quantile \\
        \hline
        \centering \begin{tabular}{c} \SW with EMS/EM \\(this paper)\end{tabular} &\checkmark & \checkmark & \checkmark & \checkmark \\
        \hline
        \centering \begin{tabular}{c} HH-ADMM \\ (this paper) \end{tabular} & \checkmark & \checkmark & \checkmark & \checkmark \\
        \hline
        \centering \CFO binning & \checkmark & \checkmark & \checkmark & \checkmark \\
        \hline
        \centering \begin{tabular}{c} HH~\cite{pvldb:KulkarniCD18} and \\ HaarHRR~\cite{pvldb:KulkarniCD18} \end{tabular} & & \checkmark & & \\
        \hline
        \centering \PM~\cite{icde:WangXYHSSY18} and \SR ~\cite{jasa:DuchiJW18} & & &\checkmark &  \\
        \hline
    \end{tabular}
    }
    \caption{Methods and evaluated metrics.}
    \label{tbl:method_metric}
\end{center}
\vspace{-1cm}
\end{table}

As a brief overview of the experiment results, \SW with EMS performs best with different privacy budgets and different metrics. 
HH-ADMM performs best on the income dataset under some of the metrics.
We also experimentally demonstrate the better utility of \SW over other wave shapes in \GW and the near-optimal choice of $b$ for \SW.

\mypara{Evaluation methodology.}
The algorithms are implemented using Python 3.6 and Numpy 1.15; the experiments are conducted on a server with Intel Xeon 4108 and 128GB memory.  
For each dataset and each method, we repeat the experiment $100$ times, with result mean and standard deviation reported.  The standard deviation is ignored because it is typically very small, and barely noticeable in the figures.

\subsection{Distribution Distance}
We first evaluate metrics that capture the quality of the recovered distributions.  Note that HH and haarHRR are not included (but HH-ADMM is) because HH or haarHRR does not result in valid distributions.

\mypara{Wasserstein Distance.}
Figure~\ref{wass_beta}-\ref{wass_RT} shows the Wasserstein distance $W_1$ of reconstructed distribution and the true distribution.  
In most cases, \SW with EMS performs best, followed by EM and HH-ADMM.
For the \CFO-binning methods, when $\epsilon$ is small, larger binning sizes (i.e., fewer number of bins) tend to give better performance.  The lines for larger binning sizes flatten as $\epsilon$ increases, showing that the errors are dominated by biases due to binning. 
When $\epsilon$ becomes larger, \CFO-binning with smaller bin sizes (i.e., more bins) becomes better.  
We observe that even if we could choose the optimal bin size empirically for each dataset and $\epsilon$ value, the result would still be worse than \SW with EMS.

\mypara{KS Distance.}
Figure~\ref{ks_beta}-\ref{ks_RT} show the K-S distance.  
For Beta, taxi pickup time and retirement datasets, \SW with EMS generally performs the best, followed by EM.
For the income dataset, HH-ADMM performs better than EM and EMS under this metric, especially under larger $\epsilon$ values.
This is because the income dataset is more spiky, due to the fact that people tend to report income using round numbers.  
HH-ADMM is better at preserving some of the spikes in the distribution, whereas \SW with EM or EMS will smooth the spikes.  
Since KS distance measures maximum difference at one point in CDF, HH-ADMM results in lower errors under KS distance, 
even though it produces higher error under Wasserstein Distance.
For similar reason, \CFO with larger bin size also perform poorly on the income dataset under KS distance.

\subsection{Semantic and Statistical Quantities}
We compare the results of different methods using the range query and statistic quantities including mean, variance, quantiles.
For mean and variance, we also consider the \SR and \PM, which were designed for mean estimations. 
All results are measured by Mean Absolute Error (MAE).

\begin{figure*}[h!]
    \centering
    \begin{subfigure}[b]{\textwidth}
		\includegraphics[width=1\textwidth]{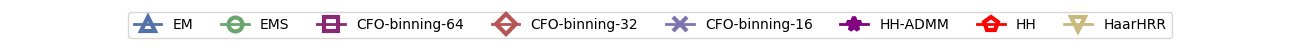}
	    \vspace{0.00mm}
	\end{subfigure}
	\\
\vspace{-0.5cm}
	\begin{subfigure}[b]{0.23\textwidth}
		\includegraphics[width=\textwidth]{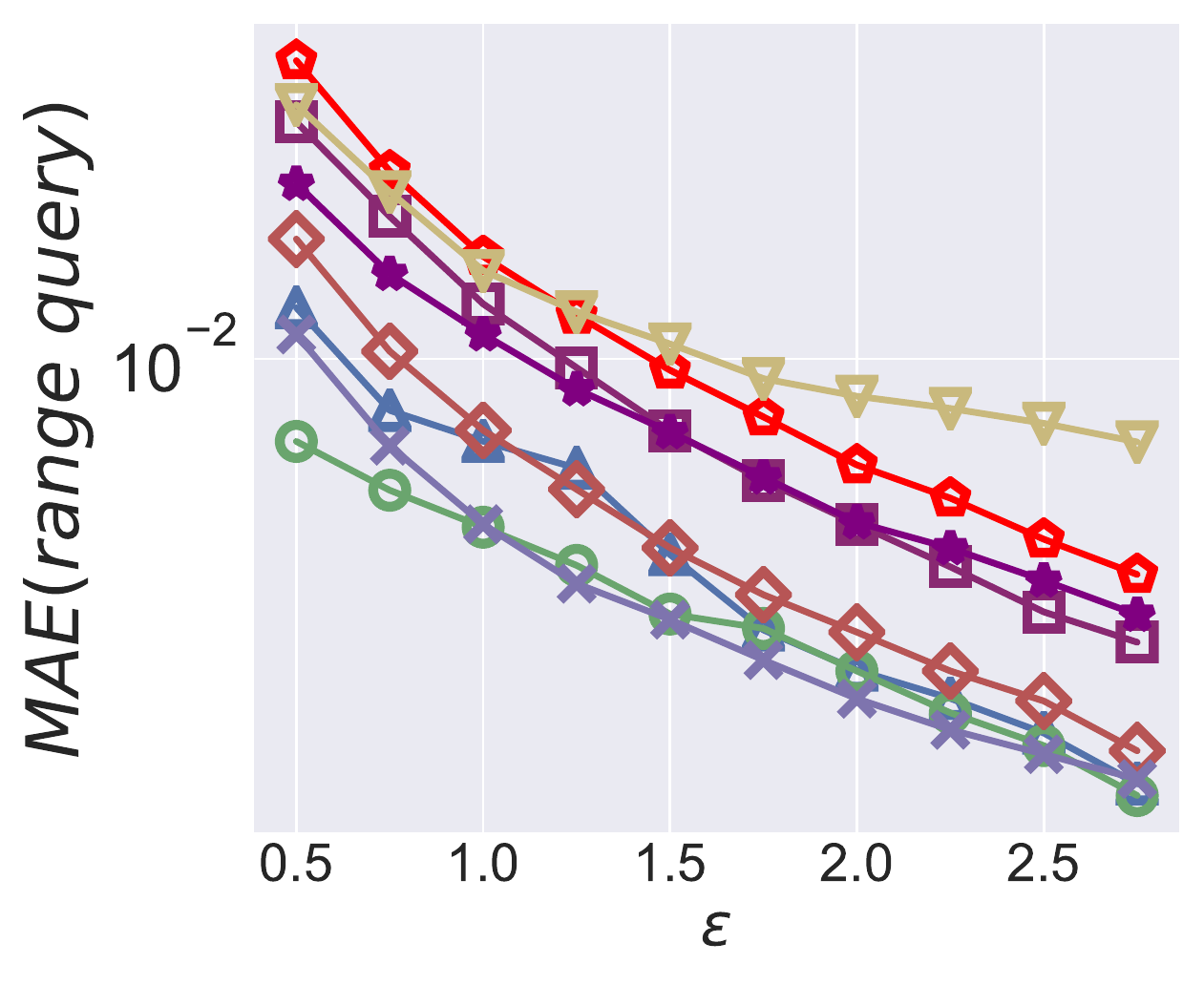}
		\vspace{-0.7cm}
		\caption{Beta(5,2) $\alpha= 0.1$}
		\label{rq_beta_0.1}
	\end{subfigure}
	\begin{subfigure}[b]{0.23\textwidth}
		\includegraphics[width=\textwidth]{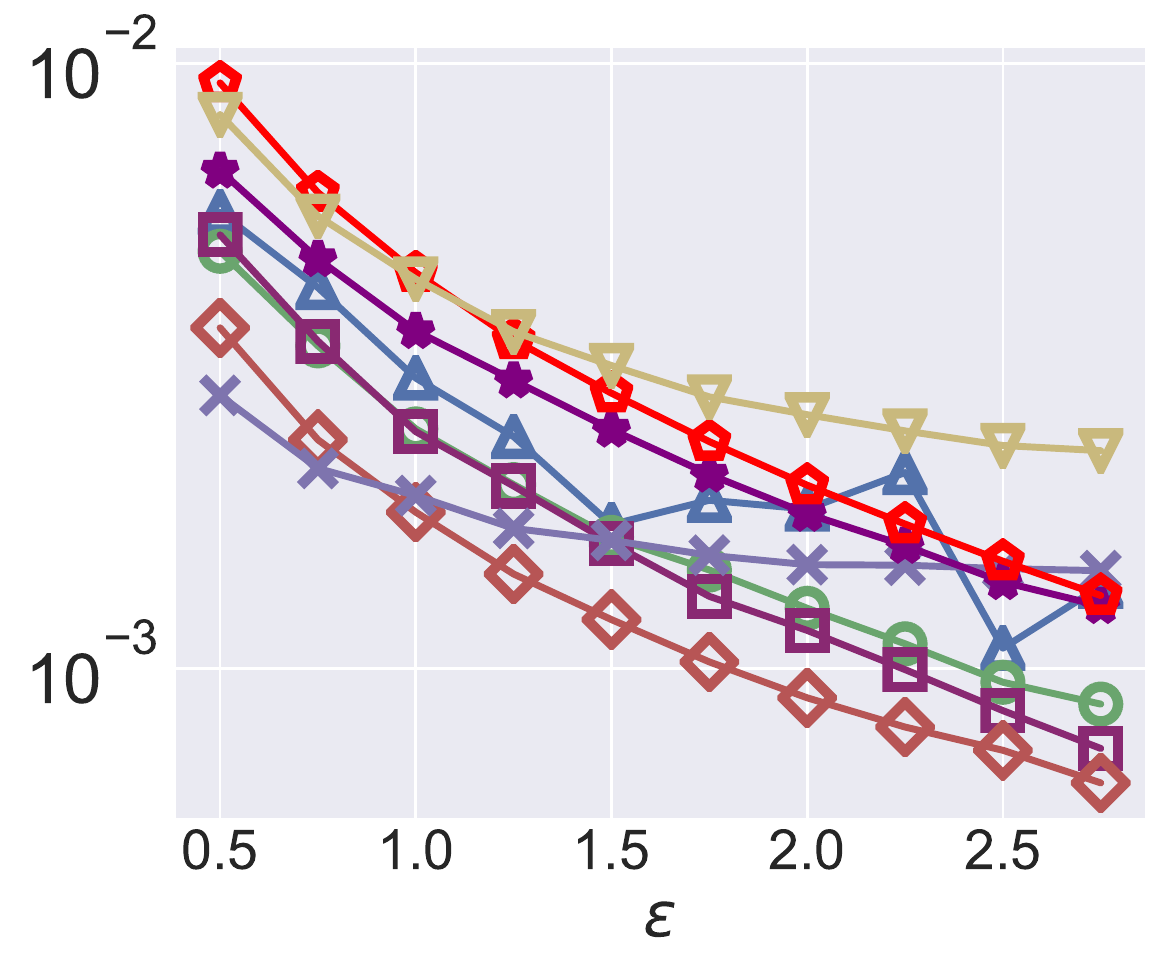}
		\vspace{-0.7cm}
		\caption{Taxi pickup time $\alpha= 0.1$}
		\label{rq_PT_0.1}
	\end{subfigure}
    \begin{subfigure}[b]{0.23\textwidth}
		\includegraphics[width=\textwidth]{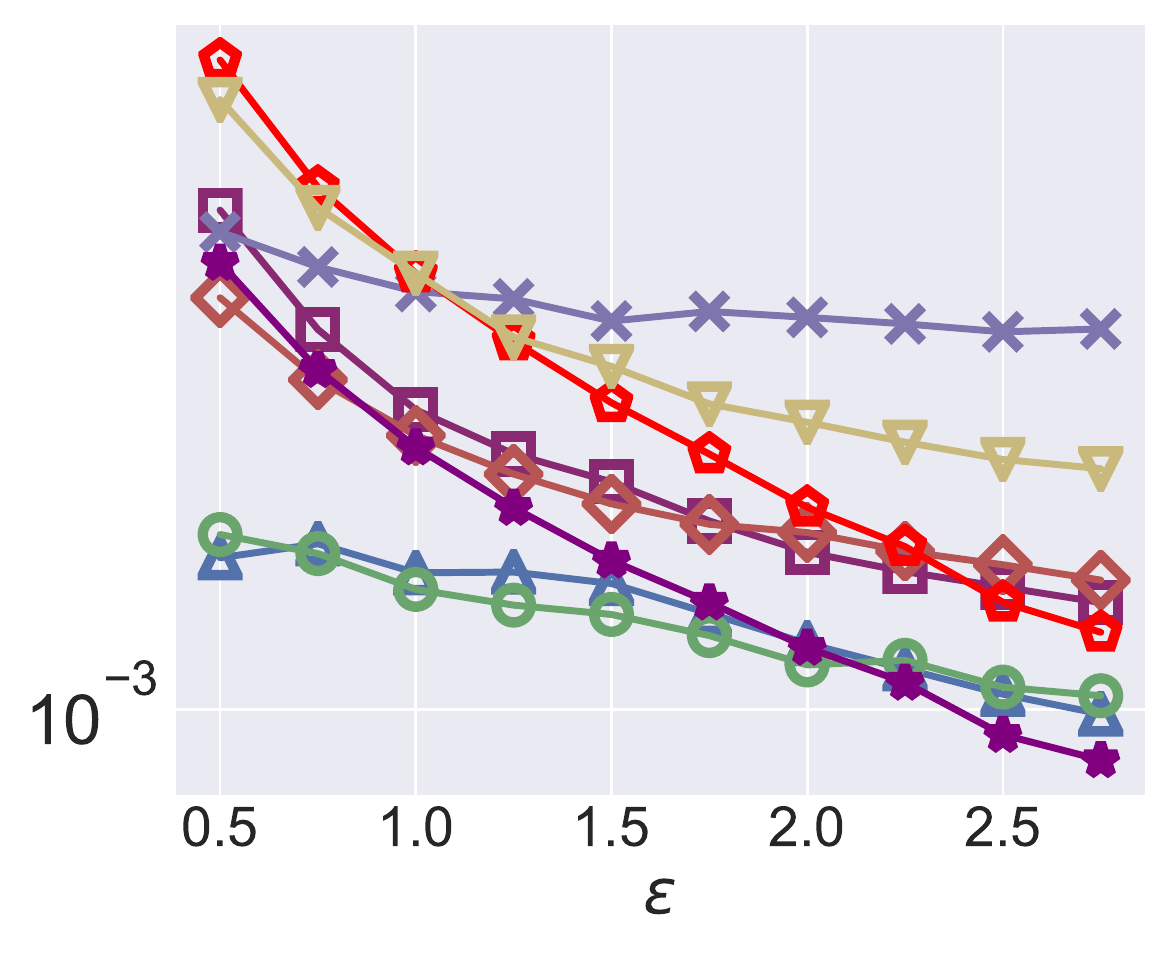}
		\vspace{-0.7cm}
		\caption{Income $\alpha= 0.1$}
		\label{rq_INC_0.1}
	\end{subfigure}
    \begin{subfigure}[b]{0.23\textwidth}
		\includegraphics[width=\textwidth]{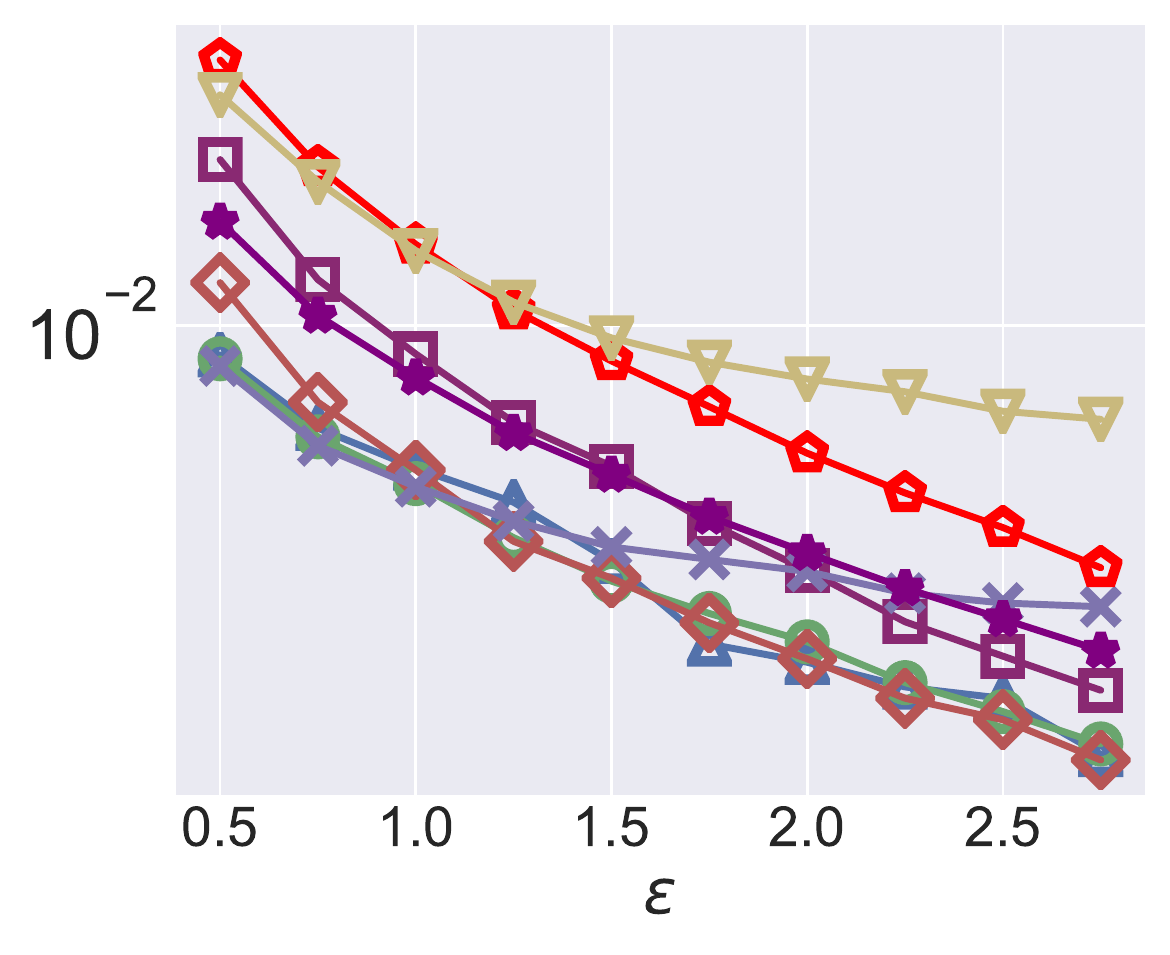}
		\vspace{-0.7cm}
		\caption{Retirement $\alpha= 0.1$}
		\label{rq_RT_0.1}
	\end{subfigure}\\
\begin{subfigure}[b]{0.23\textwidth}
		\includegraphics[width=\textwidth]{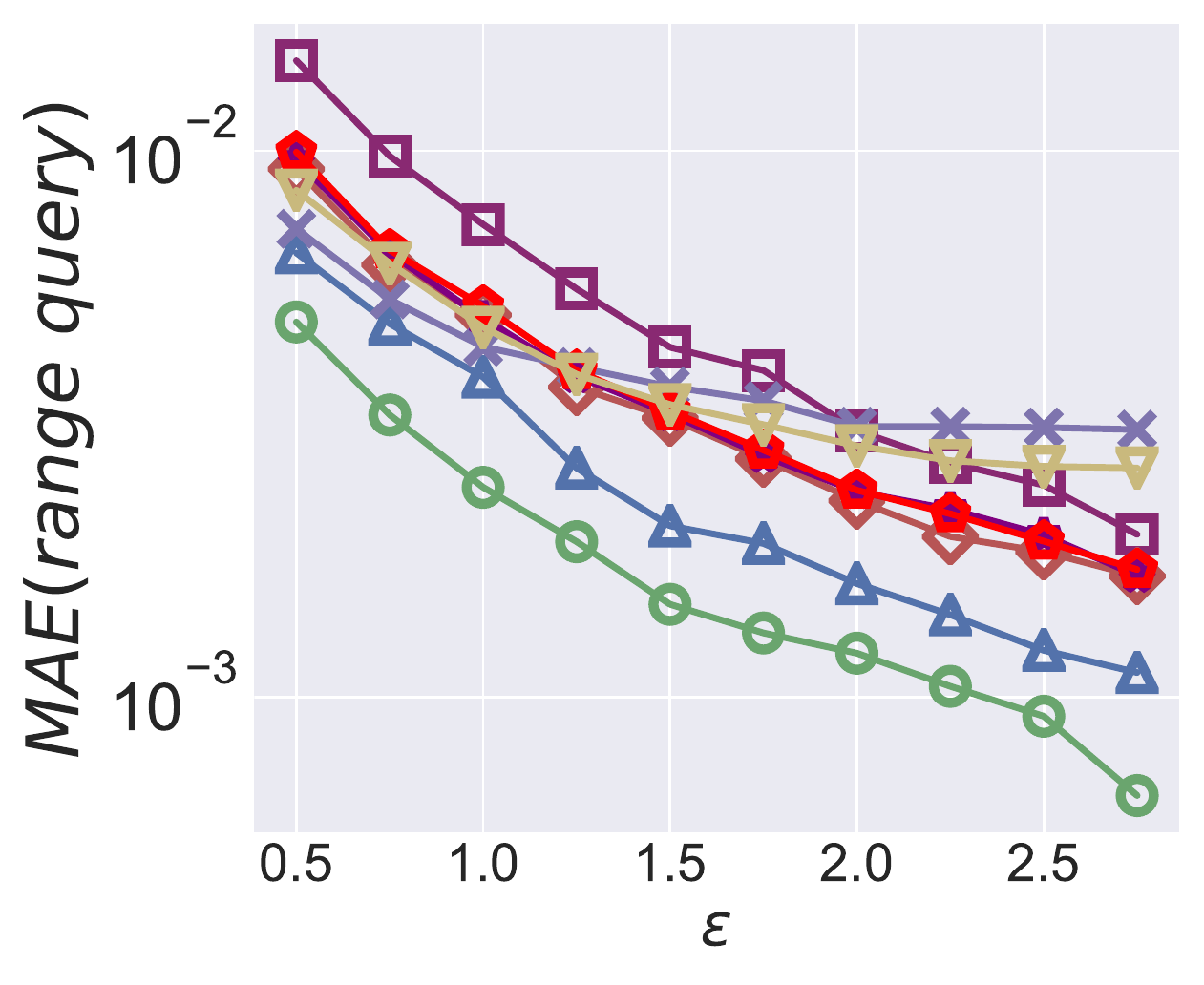}
		\vspace{-0.7cm}
		\caption{Beta(5,2) $\alpha= 0.4$}
		\label{rq_beta_0.4}
	\end{subfigure}
	\begin{subfigure}[b]{0.23\textwidth}
		\includegraphics[width=\textwidth]{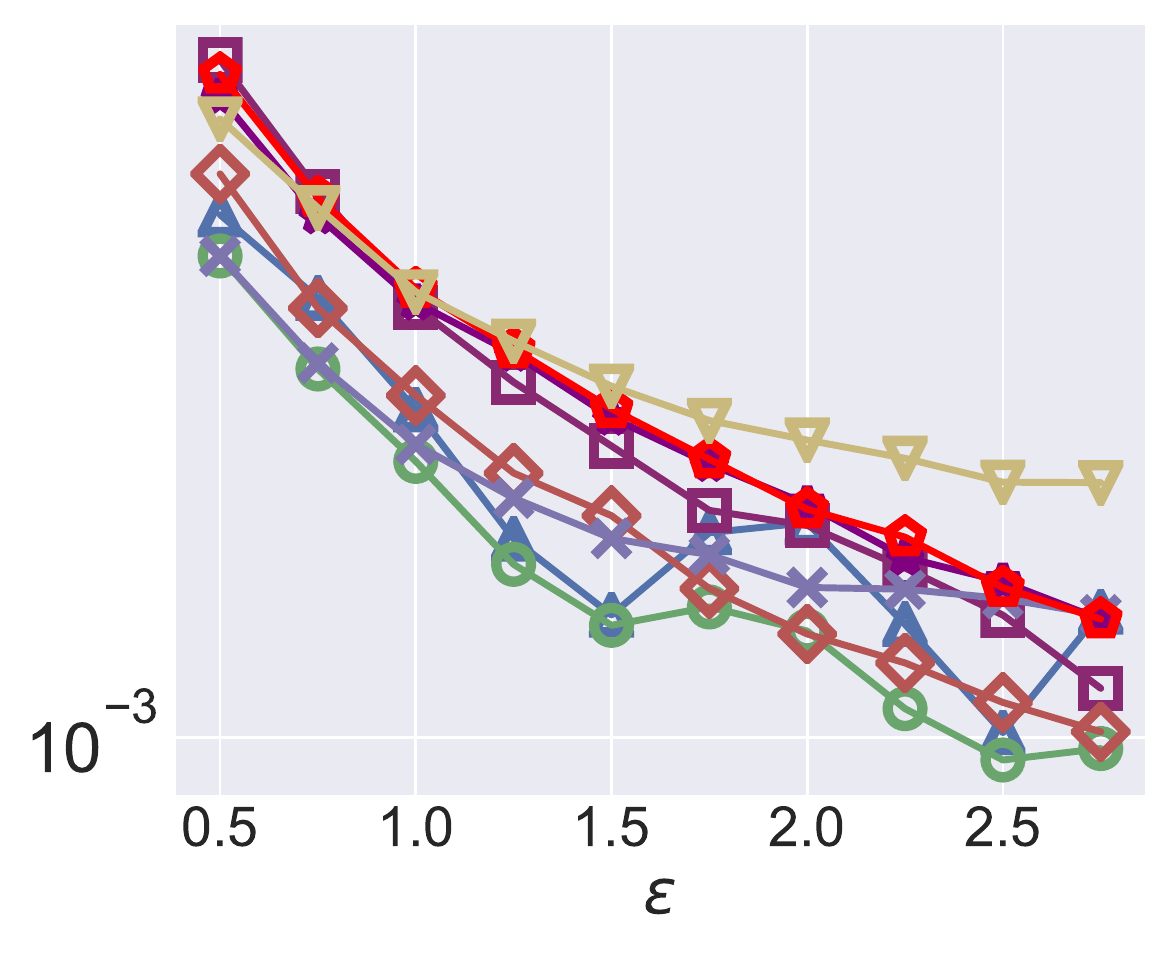}
		\vspace{-0.7cm}
		\caption{Taxi pickup time $\alpha= 0.4$}
		\label{rq_PT_0.4}
	\end{subfigure}
    \begin{subfigure}[b]{0.23\textwidth}
		\includegraphics[width=\textwidth]{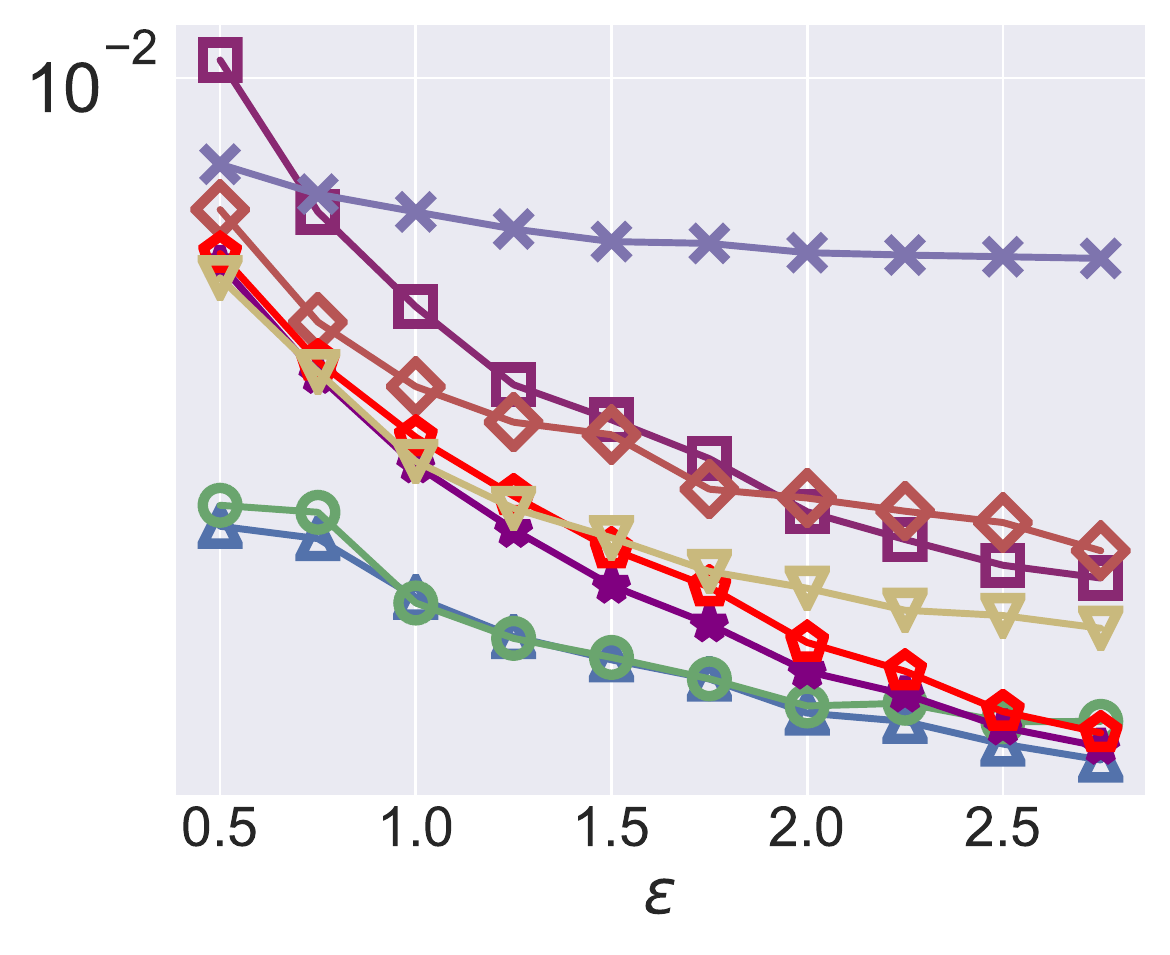}
		\vspace{-0.7cm}
		\caption{Income $\alpha= 0.4$}
		\label{rq_INC_0.4}
	\end{subfigure}
    \begin{subfigure}[b]{0.23\textwidth}
		\includegraphics[width=\textwidth]{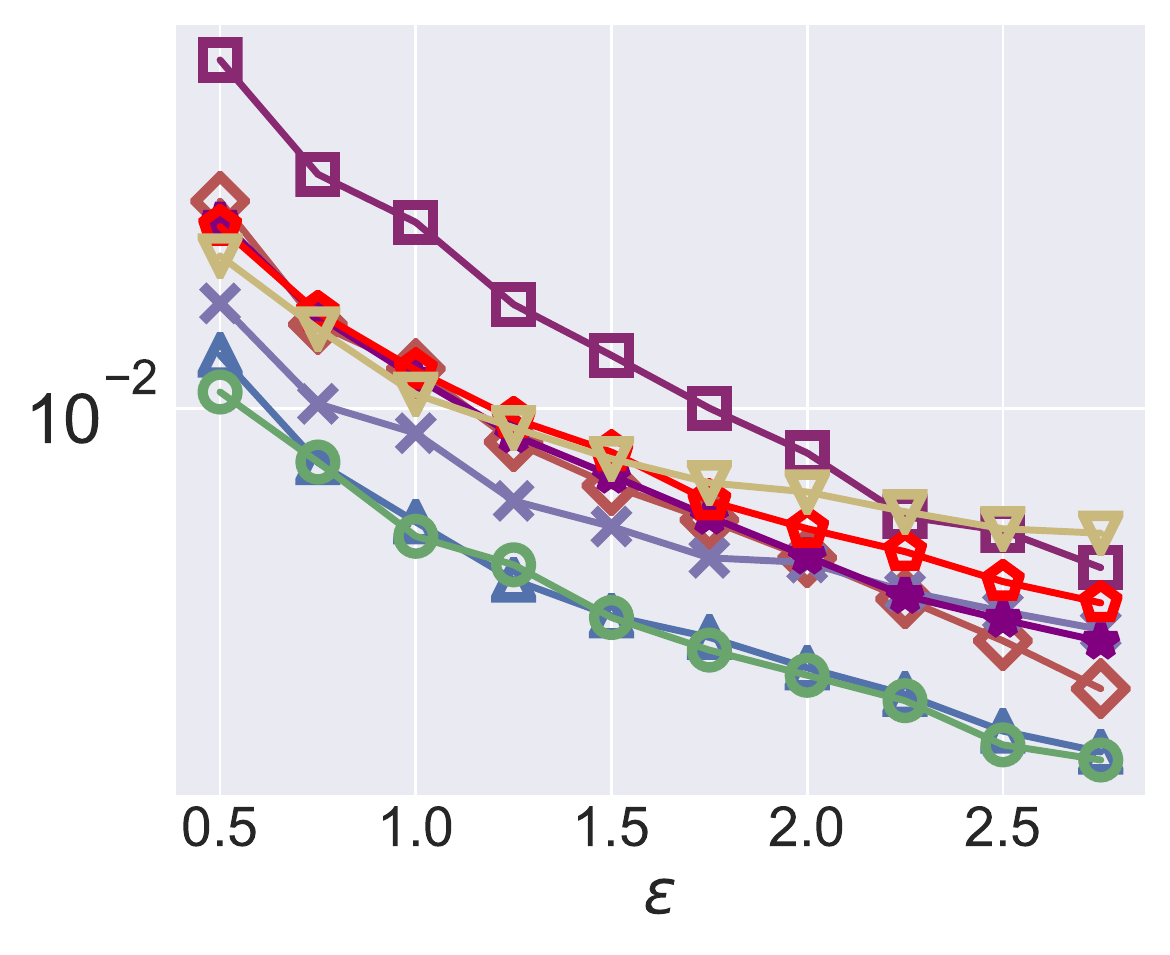}
		\vspace{-0.7cm}
		\caption{Retirement $\alpha= 0.4$}
		\label{rq_RT_0.25}
	\end{subfigure}\\
	\vspace{-0.3cm}
	\caption{
	MAE of random range query with range $\alpha=0.1$ (first row) and $\alpha=0.4$ (second row).
	}
	\label{rq}
	\vspace{-0.2cm}
\end{figure*}

\mypara{Range Query.}
The queries are randomly generated, but with fixed range sizes.  
Denote the left and right of the range as $i$ and $i+\alpha$, we randomly generate $i\in [0, 1-\alpha]$ with $\alpha = 0.1$ and $0.4$.  
The results in Figure~\ref{rq} shows that \SW with EMS outperforms HH and HaarHRR~\cite{pvldb:KulkarniCD18}.
In fact, it is the best in most cases, except when $\alpha=0.1$ in the taxi pickup time dataset and in low privacy region of income dataset. 
However, \SW with EMS has performance similar to \CFO-binning-64 when $\alpha=0.1$ and still outperforms all the hierarchy-based approaches in taxi pickup time dataset.
For the income dataset, EM and EMS performs well in high privacy range (i.e., $\epsilon \leq 2$), while HH-ADMM performs best in low privacy range, followed by EM and EMS.

\mypara{Mean Estimation.}
Results for mean estimation are showed in Figure~\ref{mean_beta}-\ref{mean_RT}.
\SR performs better than \PM when $\epsilon$ is small, but worse when $\epsilon$ is larger.  This is consistent with the analysis in~\cite{icde:WangXYHSSY18}.
Note that \SR and \PM devote all privacy budget to estimate mean.  While \SW with EMS can estimate the full distribution, it performs comparable to the best of \SR and \PM for estimating the mean.
We also see that HH-ADMM has better performances than all other \CFO-binning methods, but is still inferior to \SW with EMS.

\begin{figure*}[h]
    \centering
    \begin{subfigure}[b]{\textwidth}
		\includegraphics[width=1\textwidth]{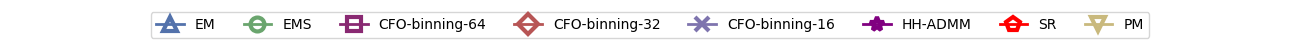}
		\vspace{0.00mm}
	\end{subfigure}\\
	\vspace{-0.4cm}
\begin{subfigure}[b]{0.23\textwidth}
		\includegraphics[width=\textwidth]{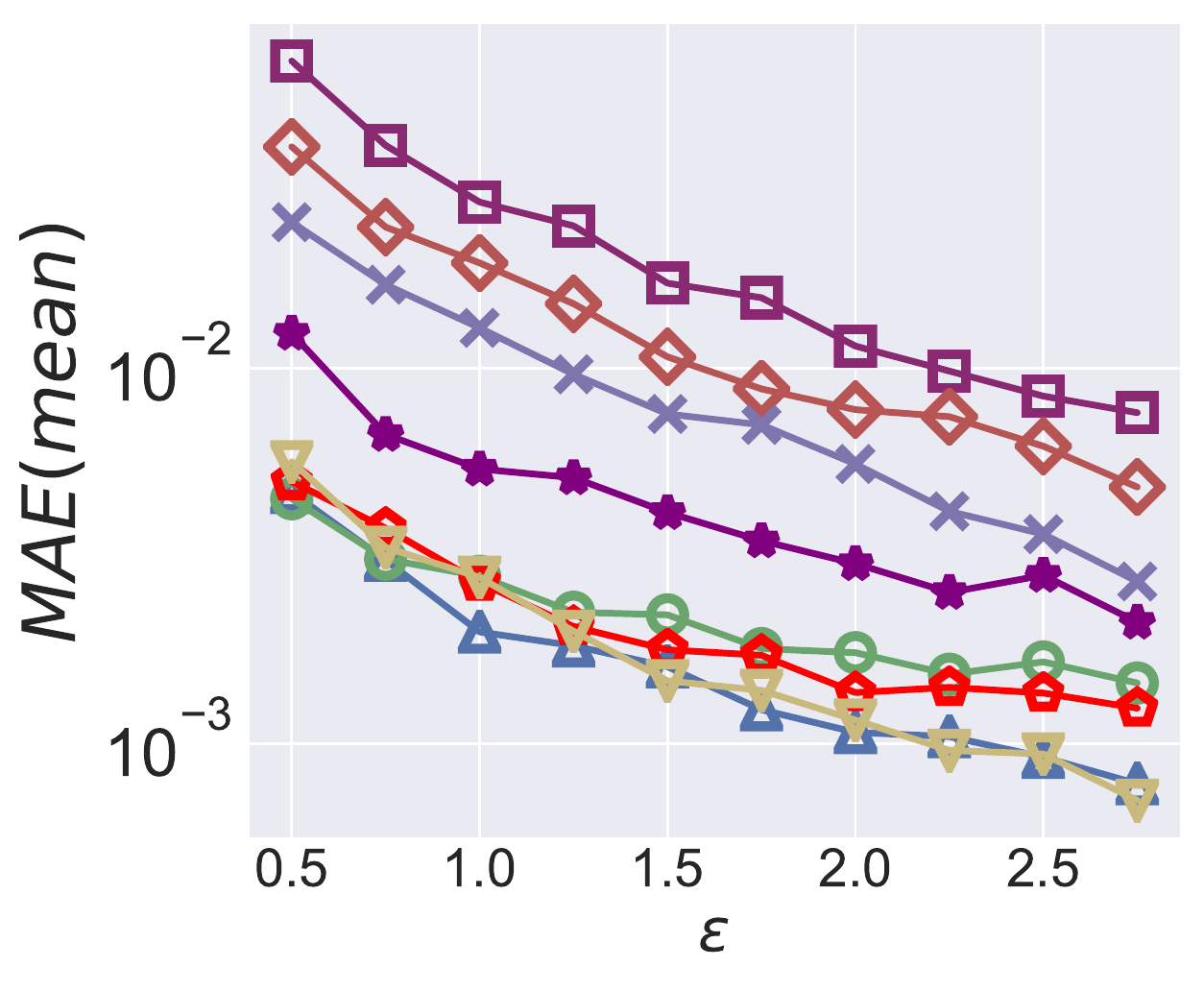}
		\vspace{-0.8cm}
		\caption{Beta(5,2)}
		\label{mean_beta}
	\end{subfigure}
\begin{subfigure}[b]{0.23\textwidth}
		\includegraphics[width=\textwidth]{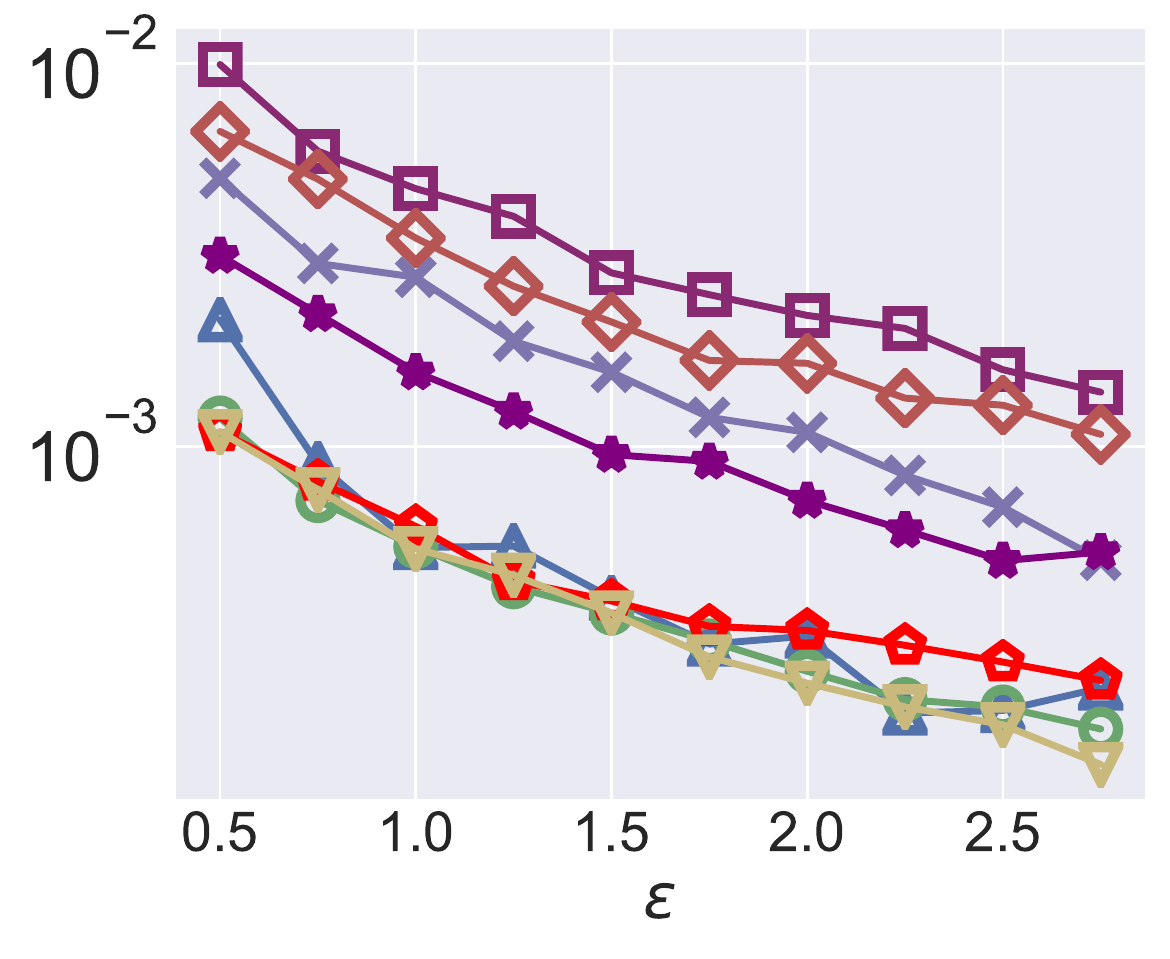}
		\vspace{-0.8cm}
		\caption{Taxi pickup time}
		\label{mean_PT}
	\end{subfigure}
\begin{subfigure}[b]{0.23\textwidth}
		\includegraphics[width=\textwidth]{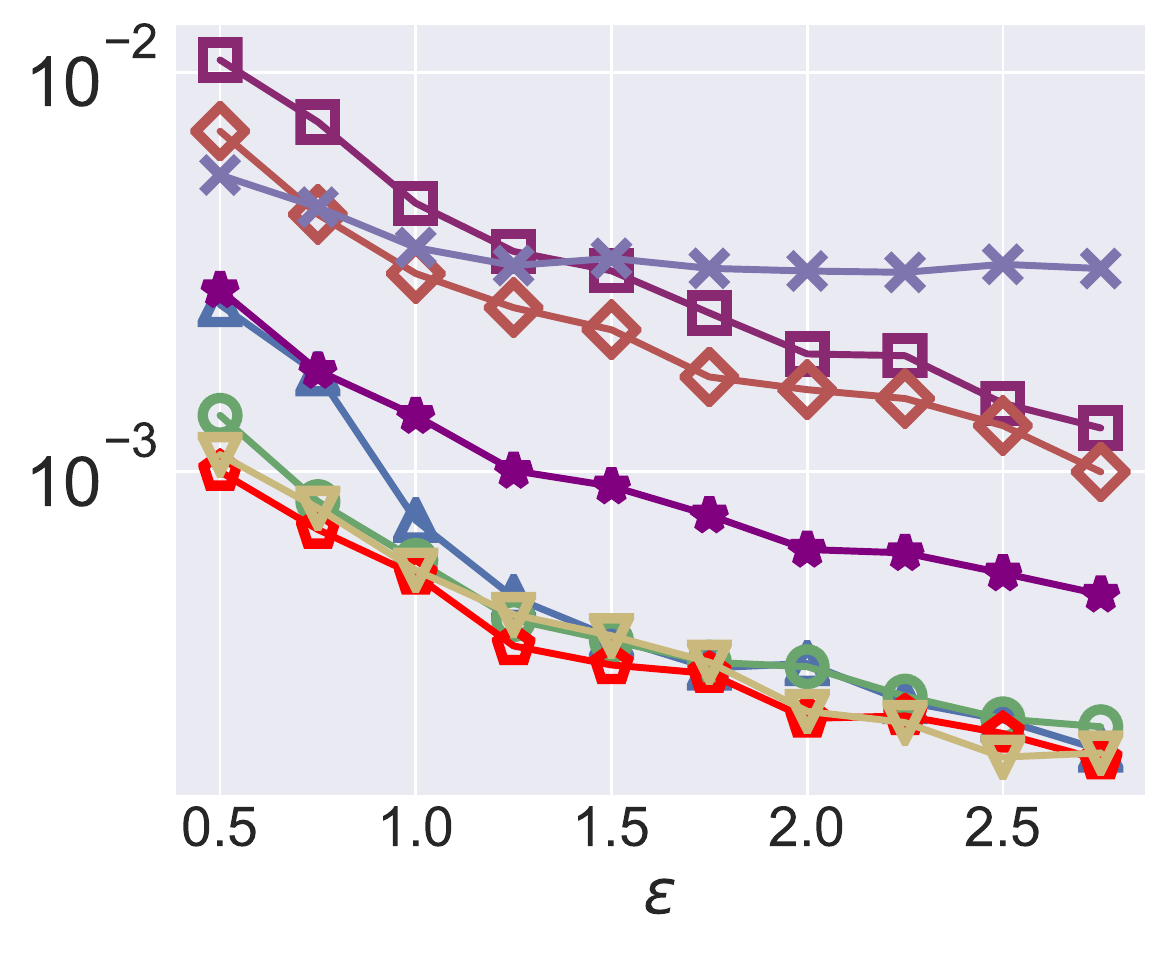}
		\vspace{-0.8cm}
		\caption{Income}
		\label{mean_INC}
	\end{subfigure}
\begin{subfigure}[b]{0.23\textwidth}
		\includegraphics[width=\textwidth]{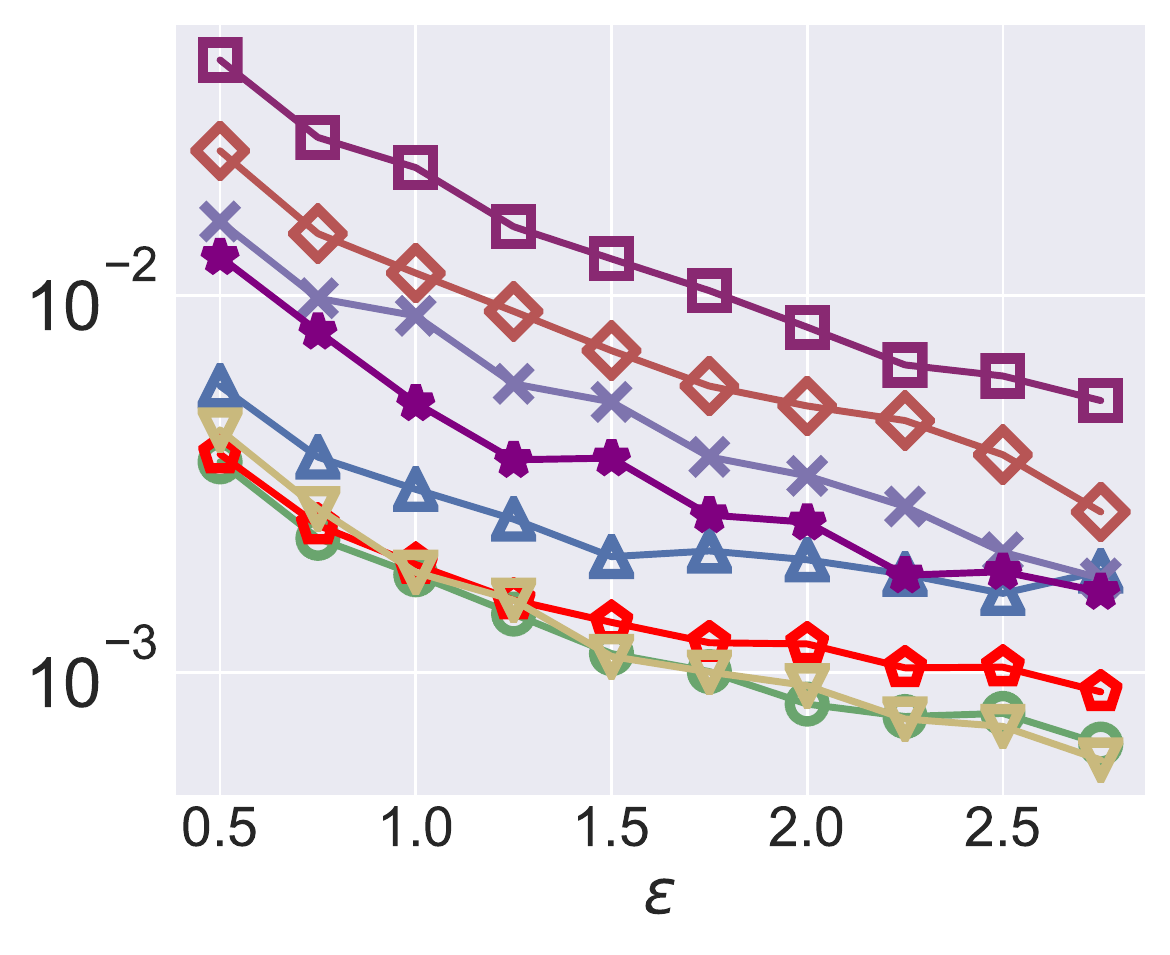}
		\vspace{-0.8cm}
		\caption{Retirement}
		\label{mean_RT}
	\end{subfigure} \\
	\vspace{-0.1cm}
\begin{subfigure}[b]{0.23\textwidth}
		\includegraphics[width=\textwidth]{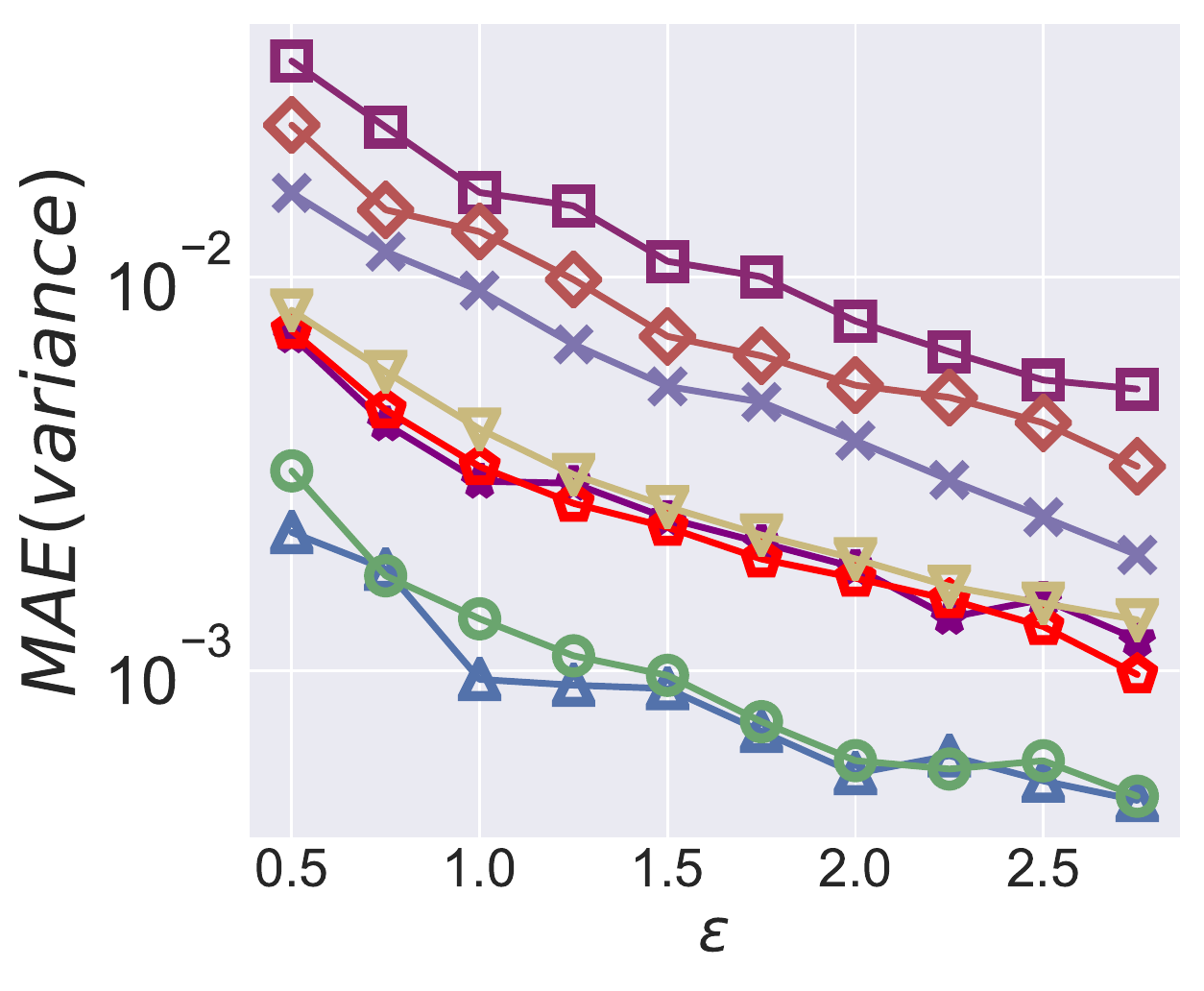}
		\vspace{-0.8cm}
		\caption{Beta(5,2)}
		\label{var_beta}
	\end{subfigure}
	 \begin{subfigure}[b]{0.23\textwidth}
		\includegraphics[width=\textwidth]{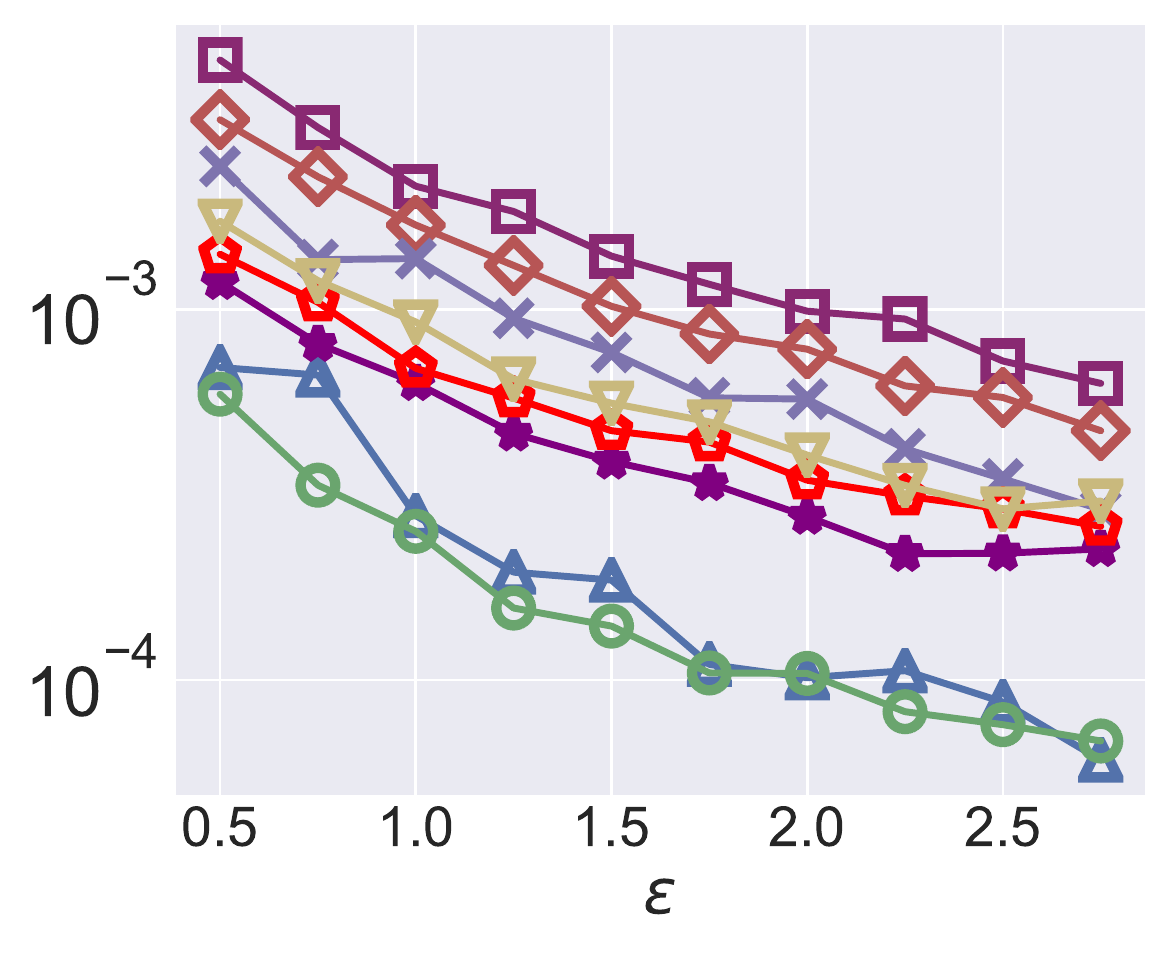}
			\vspace{-0.8cm}
		\caption{Taxi pickup time}
		\label{var_PT}
	\end{subfigure}
	\begin{subfigure}[b]{0.23\textwidth}
		\includegraphics[width=\textwidth]{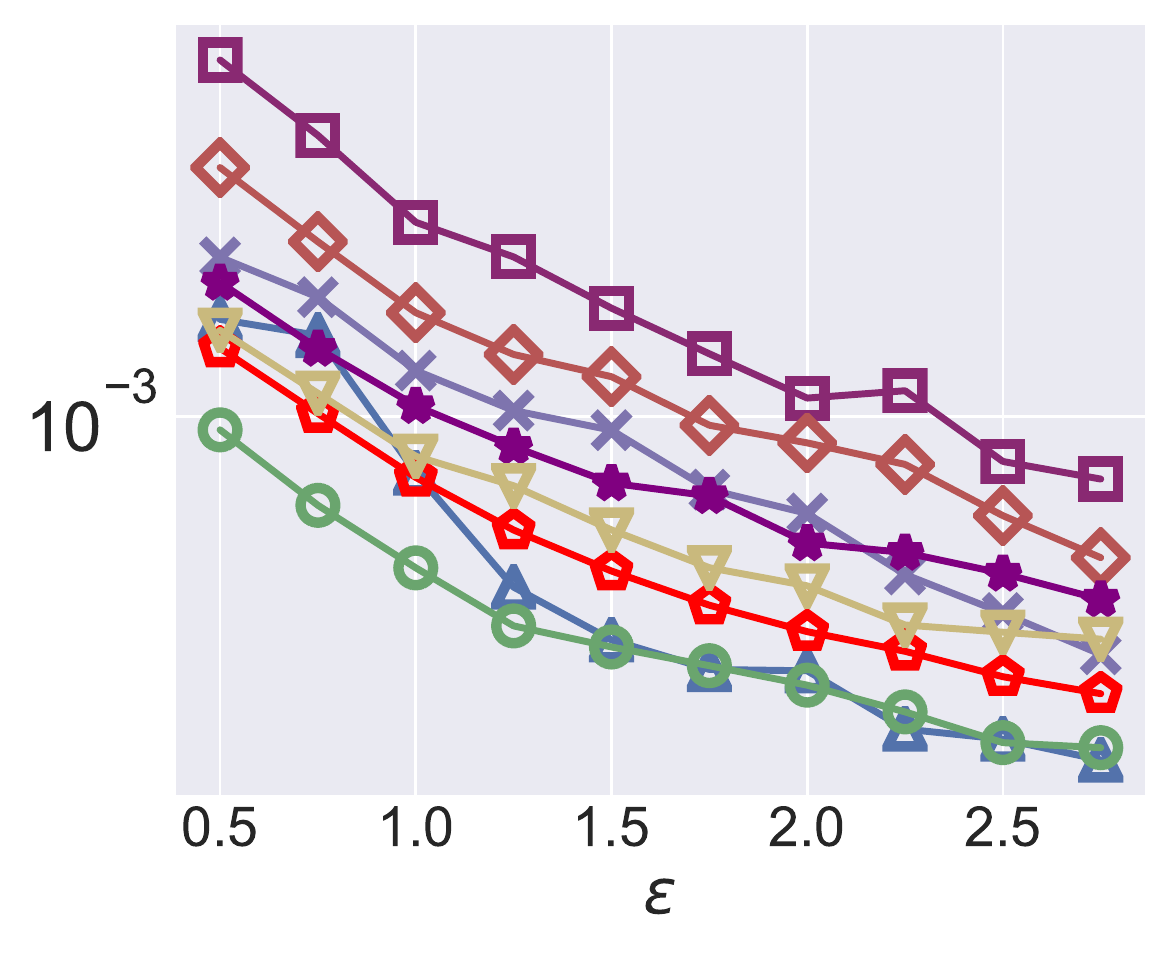}
			\vspace{-0.8cm}
		\caption{Income}
		\label{var_INC}
	\end{subfigure}
\begin{subfigure}[b]{0.23\textwidth}
		\includegraphics[width=\textwidth]{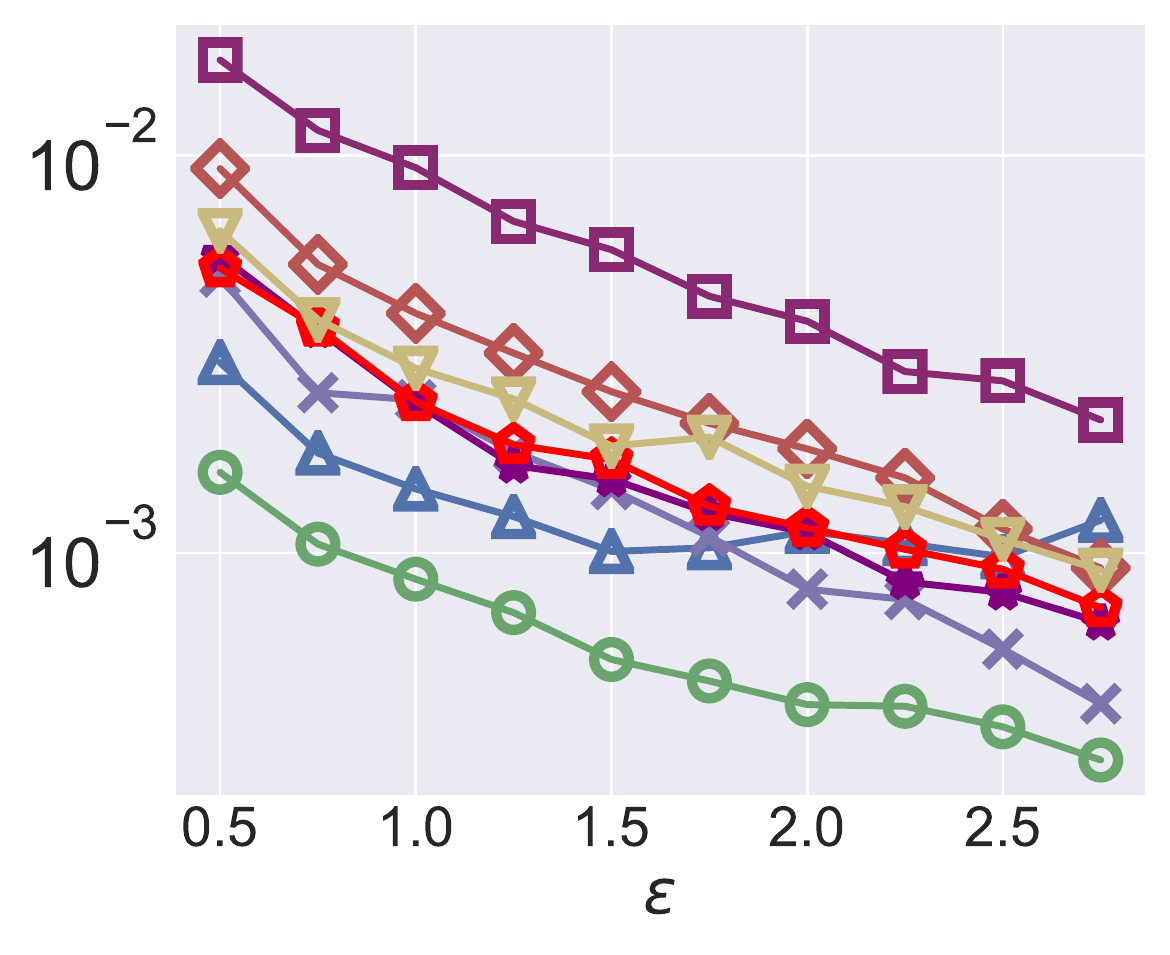}
			\vspace{-0.8cm}
		\caption{Retirement}
		\label{var_RT}
	\end{subfigure}\\
	\vspace{-0.1cm}
\begin{subfigure}[b]{0.23\textwidth}
		\includegraphics[width=\textwidth]{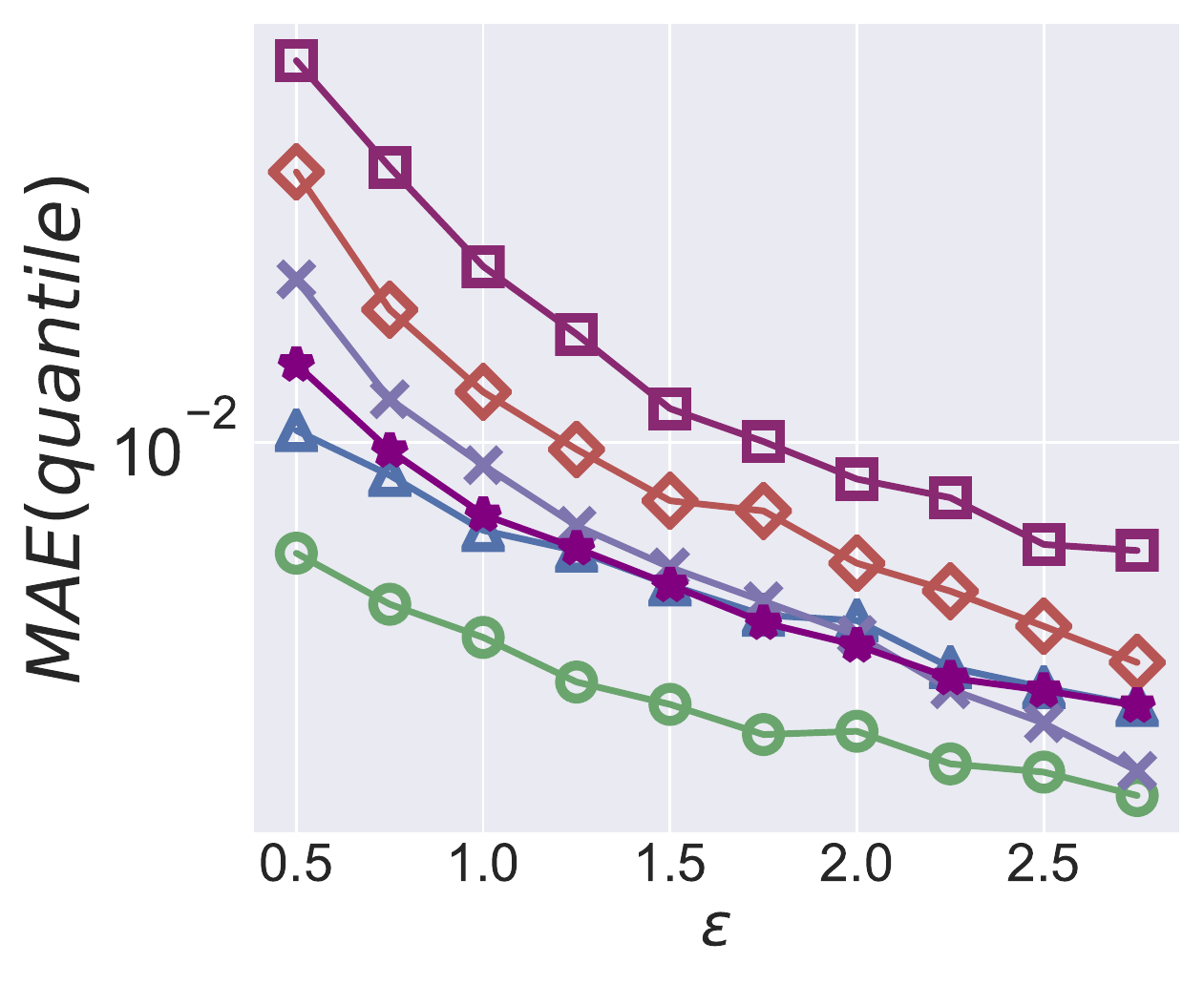}
		\vspace{-0.8cm}
		\caption{Beta(5,2)}
		 \label{quartile_beta}
	\end{subfigure}
    \begin{subfigure}[b]{0.23\textwidth}
		\includegraphics[width=\textwidth]{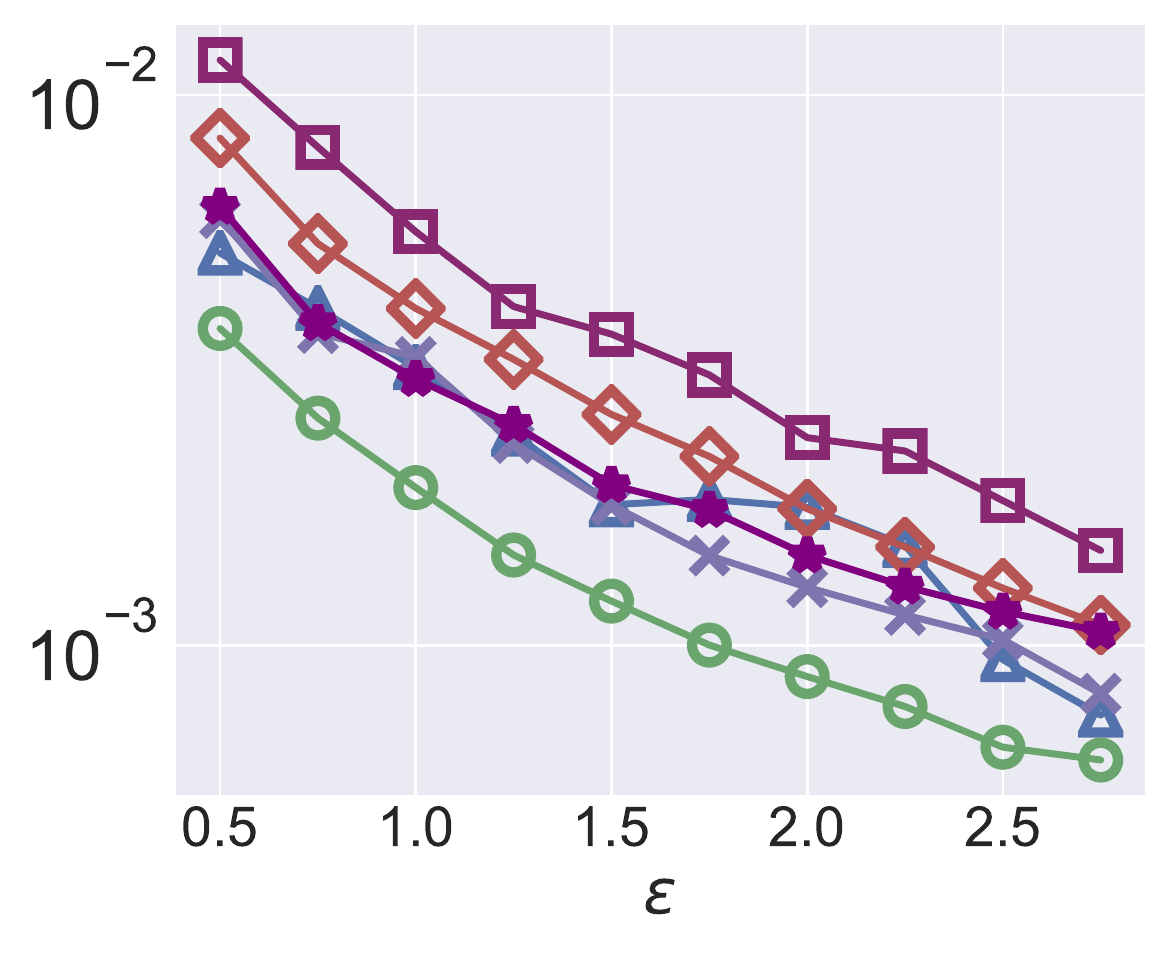}
		\vspace{-0.8cm}
		\caption{Taxi pickup time}
		\label{quartile_PT}
	\end{subfigure}
    \begin{subfigure}[b]{0.23\textwidth}
		\includegraphics[width=\textwidth]{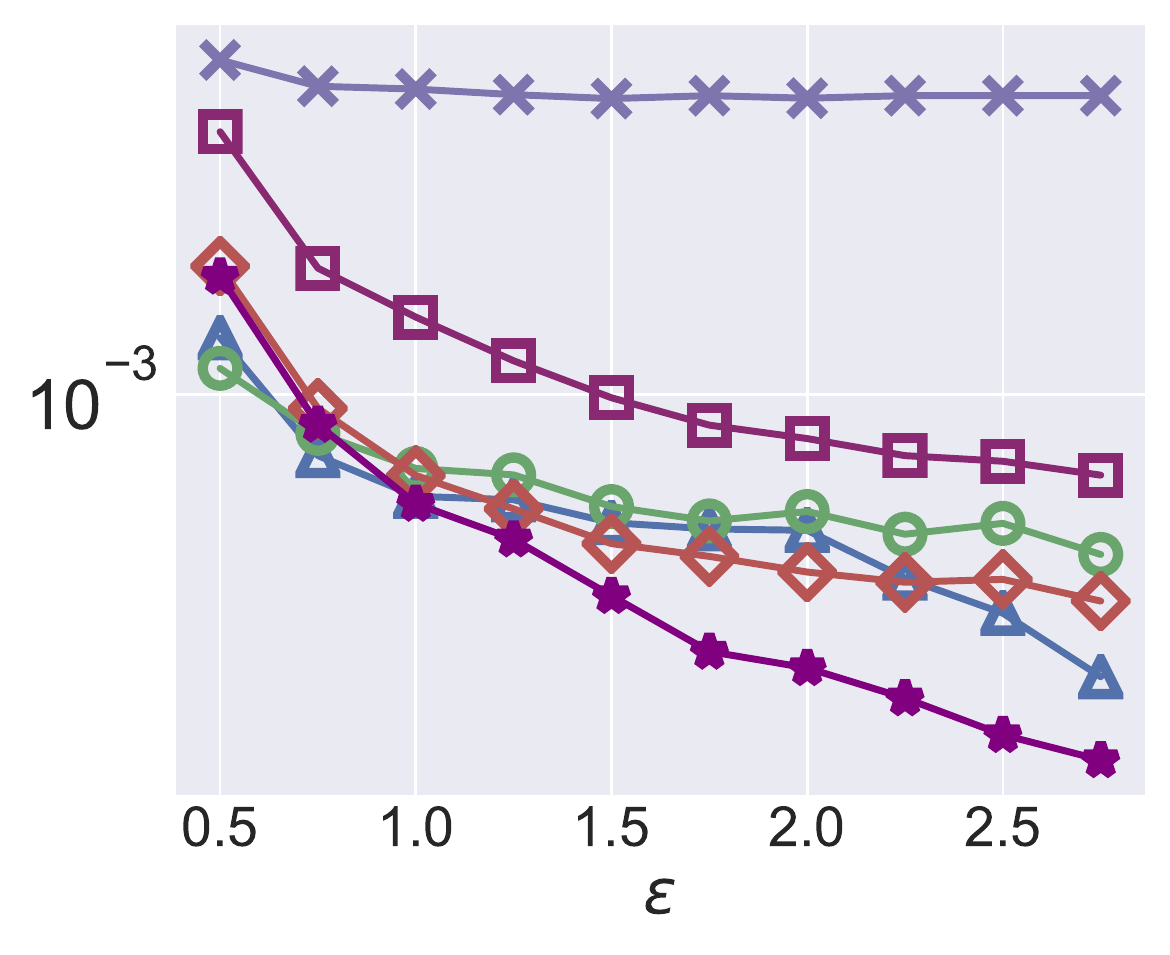}
		\vspace{-0.8cm}
		\caption{Income}
		\label{quartile_INC}
	\end{subfigure}
    \begin{subfigure}[b]{0.23\textwidth}
		\includegraphics[width=\textwidth]{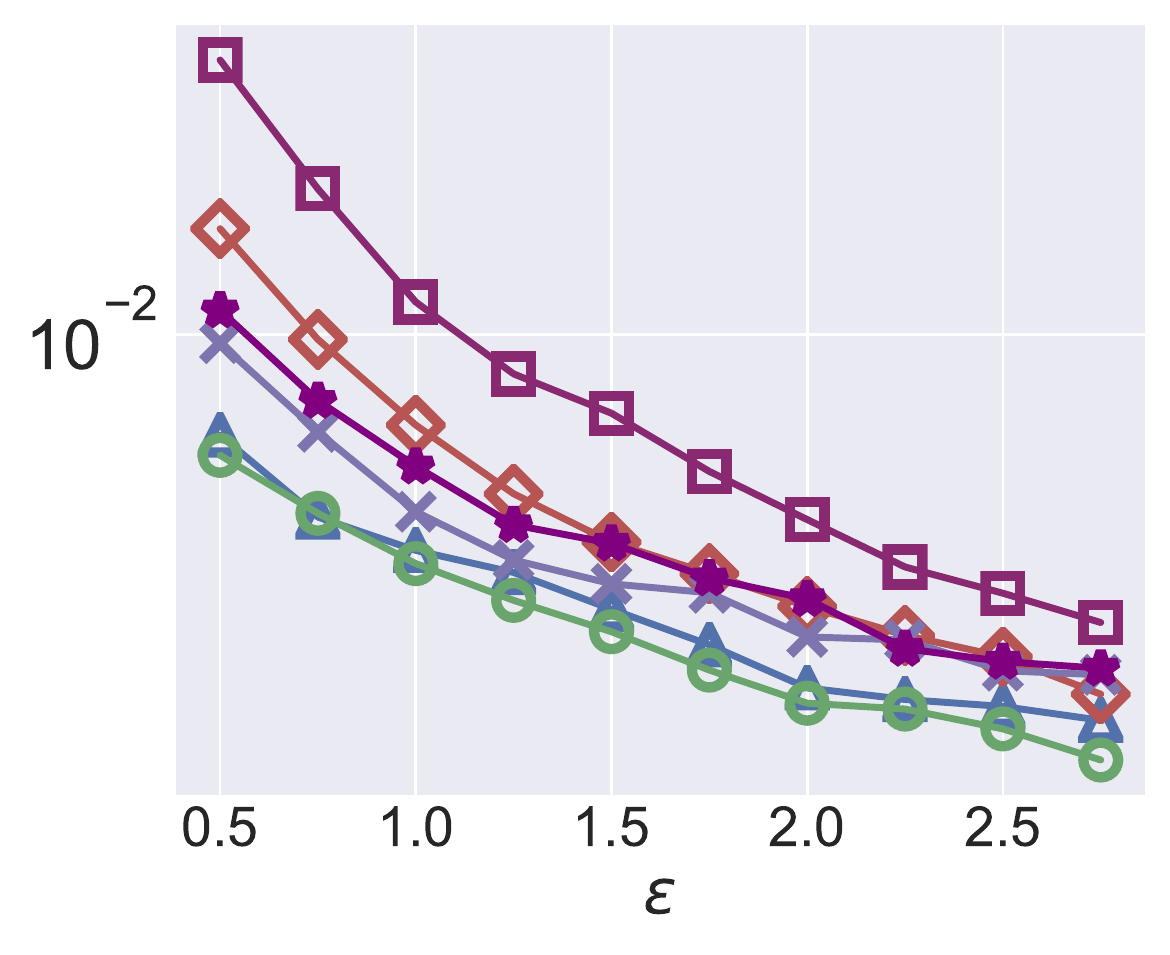}
		\vspace{-0.8cm}
		\caption{Retirement}
		\label{quartile_RT}
	\end{subfigure}
	\vspace{-0.4cm}
	\caption{MAE for estimating mean (first row), variance (second row), and quantiles (third row).  
}
	\label{exp:statistics}
	\vspace{-0.3cm}
\end{figure*}

\mypara{Variance Estimation.}
Although \SR and \PM are proposed for mean estimation, they can be modified to support variance estimation as well.  Specifically, we randomly sample $50\%$ of users to estimate mean first.  The estimated mean is then broadcast to the remaining users.  Then each user compares his secret value and the received estimated mean, and reports the squared difference (i.e., $(v_i - \tilde{\mu})^2$) to the server, who averages them to obtain variance.

The experimental results are showed in Figure~\ref{var_beta}-~\ref{var_RT}.
As we can see, the error of \SR and \PM is larger than EM or EMS in most cases.  One reason is that only half of the users is used for variance estimation (the other half is necessary for mean estimation).  
The relative performance of other methods are similar to previous experiments.

\mypara{Quantile Estimation.}
Experimental results are shown in Figure~\ref{quartile_beta}-\ref{quartile_RT}.  
Ignoring the spiky income dataset for now, our proposed \SW with EMS performs best.  Moreover, we observe that \SW with EM sometimes performs better but is not stable, because it is sensitive to parameters.  
HH-ADMM performs worse than \SW, but close to the best of \CFO with binning.
For \CFO with binning, because of the trade-off between estimation noise and the bias within the bins, larger bin sizes typically perform better in smaller $\epsilon$ ranges, while the smaller bin sizes narrows the gap as $\epsilon$ increases.

For the spiky income dataset, even for $\epsilon=0.5$, larger bin sizes give worse utility ($1$ to $2$ orders of magnitude) than other mechanisms.  This also demonstrates that the optimal bin size is data-dependent.  HH-ADMM successfully captures the spikiness of the dataset and thus performs the best.  

\subsection{Wave Shapes and Parameters}
\label{subsec:exp_different_shape}
Here we compare the different shapes of General Wave (\GW) with \SW, and different parameters of \SW.

\mypara{Different shapes of wave in \GW.} In Section~\ref{subsec:sw}, we analytically show that \SW is preferred because it maximizes the Wasserstein distance between output distributions.  We empirically compare \SW with other wave forms. 
We consider 5 other \GW mechanism with different shape, including 4 trapezoid shapes and one triangle shape. 
The upper side to bottom side length ratio of trapezoid wave are $0.2, 0.4, 0.6$ and $0.8$.
The experimental results in Figure~\ref{exp:vary_shape} show when $\epsilon = 1$, \SW gives the best estimated distributions in terms of Wasserstein distance, no matter how we change $b$.  
As the ratio decreases, the recovery accuracy also degrades in general.
The experimental results support our intuition in Section~\ref{subsec:sw}.

\mypara{\SW with different $b$.}
In Section~\ref{subsec:choose_b}, we propose to use $b_{\SW} = \frac{\epsilon e^\epsilon - e^\epsilon +1}{2e^\epsilon(e^\epsilon - 1 - \epsilon)}$.  
Figure~\ref{exp:vary_b} reports experimental results with different $b$.  Our choice of $b_{\SW}$, which is indicated as the vertical dotted line, is among the ones that provide best utility.  
We have also evaluated $b$ on other metrics; the results give similar conclusion, and are omitted because of space limitation.

\mypara{Bucketization granularity.}
\label{subsec:exp_bucketize}
To see what is the optimal bucketization granularity on different datasets, we choose $4$ different numbers of buckets ($256, 512, 1024 $ and $2048$) then compare the Wasserstein distance between the estimated distributions and the true distributions. 
For simplicity, we use same number of buckets for both $\Codomain$ and $\Domain$.
The experimental results in Figure~\ref{exp:bucketize} show different datasets have different optimal bucketization granularity.
For Beta(5,2), we have best result when the number of buckets is 256.
For the other 3 datasets, dividing $\Domain$ into 1024 buckets can give us best performance in most cases.

\begin{figure*}[t]
    \centering
	\begin{subfigure}[b]{\textwidth}
        \centering
		\includegraphics[width=1\textwidth]{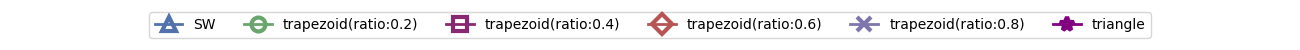}
	\end{subfigure} \\
	\begin{subfigure}[b]{0.22\textwidth}
		\includegraphics[width=\textwidth]{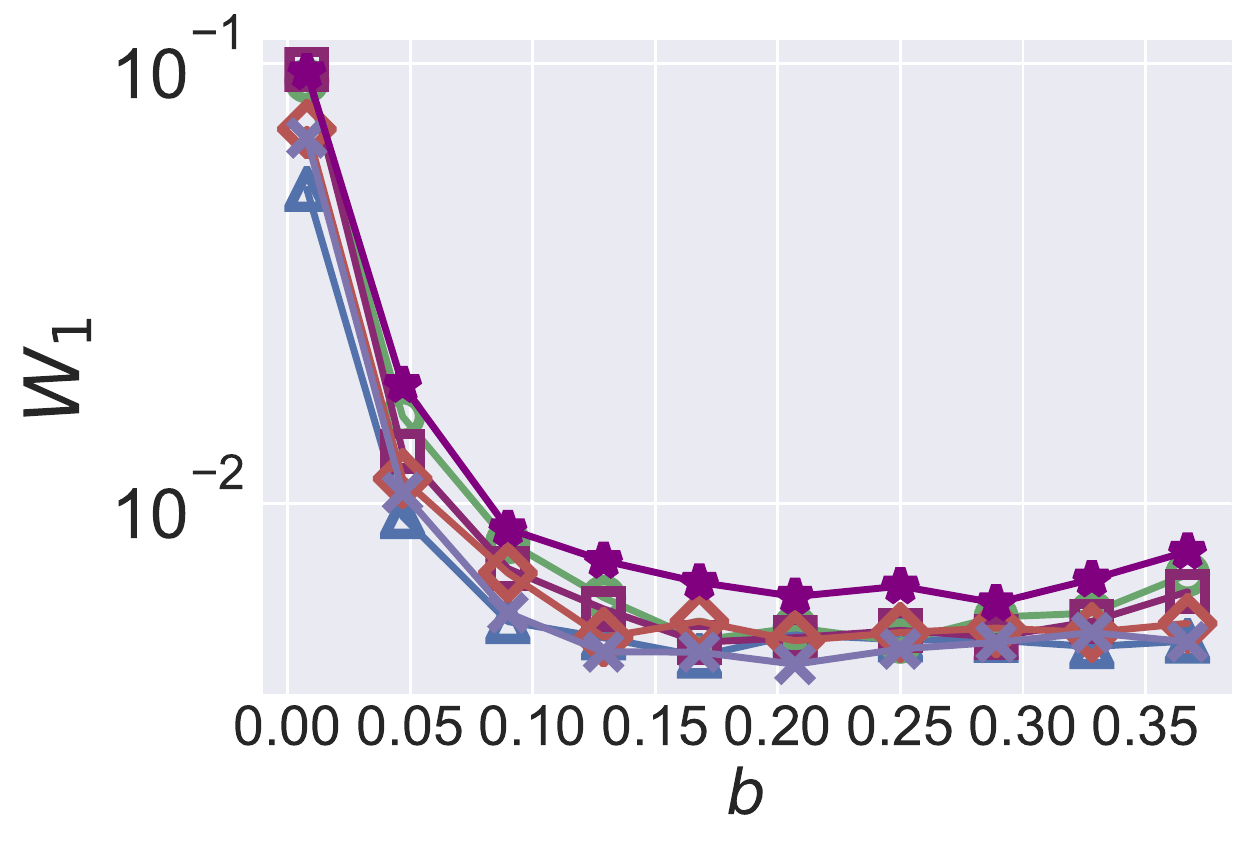}
	    \vspace{-0.6cm}
		\caption{Beta(5, 2) $\epsilon = 1.0$}
		\label{wass_different_shape_beta}
	\end{subfigure}
	\begin{subfigure}[b]{0.22\textwidth}
		\includegraphics[width=\textwidth]{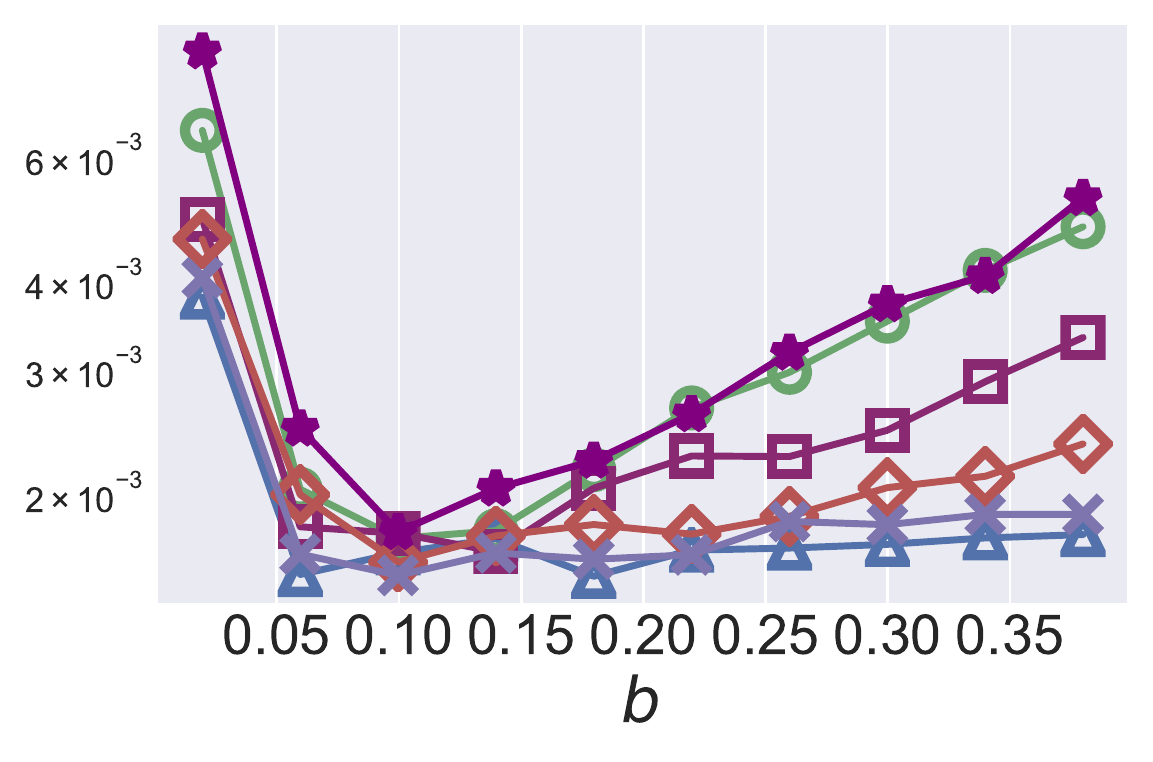}
	    \vspace{-0.6cm}
		\caption{Taxi pickup time $\epsilon = 1.0$}
		\label{wass_different_shape_PT}
	\end{subfigure}
	\begin{subfigure}[b]{0.22\textwidth}
		\includegraphics[width=\textwidth]{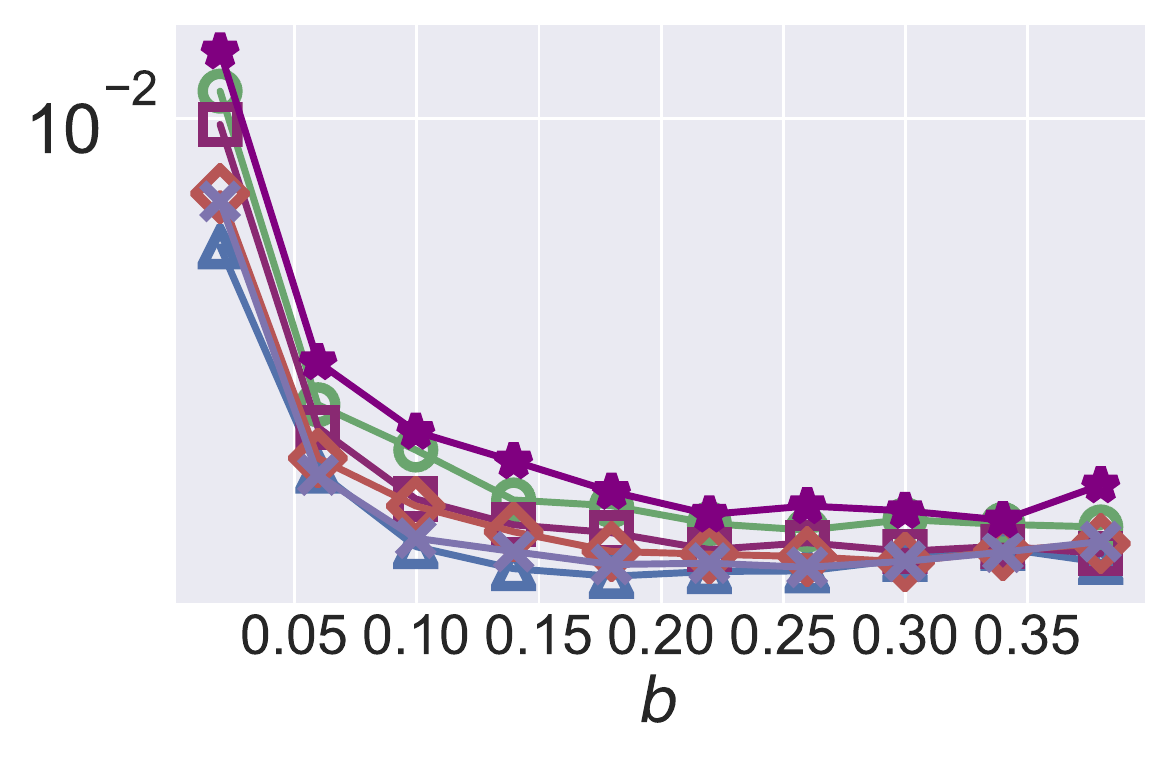}
	    \vspace{-0.6cm}
		\caption{Income $\epsilon = 1.0$}
		\label{wass_different_shape_INC}
	\end{subfigure}
	\begin{subfigure}[b]{0.22\textwidth}
		\includegraphics[width=\textwidth]{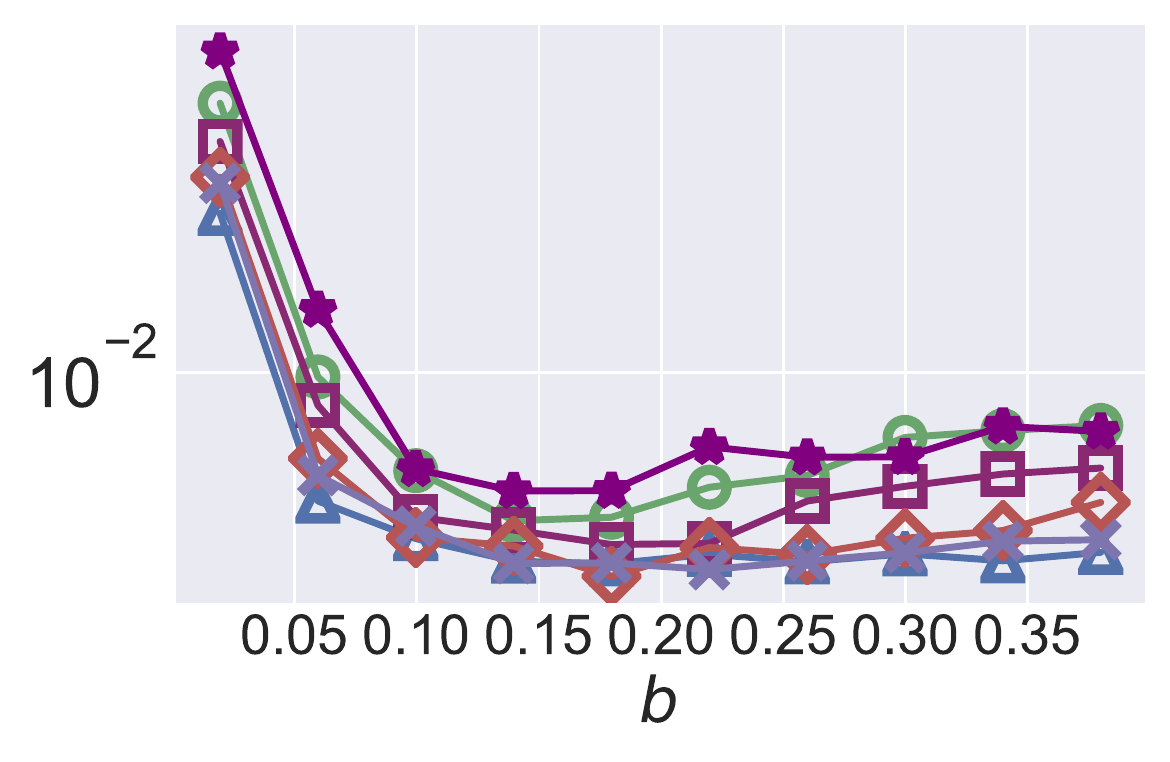}
	    \vspace{-0.6cm}
		\caption{Retirement $\epsilon = 1.0$}
		\label{wass_different_shape_RT}
	\end{subfigure}
	\vspace{-0.4cm}
	\caption{
	Comparison of different shapes of wave in \GW. 
	Ratios are the upper/lower length ratios for trapezoids.
	}
	\label{exp:vary_shape}
	\vspace{-0.3cm}
\end{figure*}

\begin{figure*}[t]
    \centering
    \begin{subfigure}[b]{\textwidth}
		\includegraphics[width=1\textwidth]{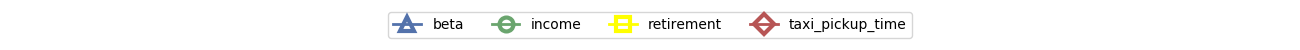}
	\end{subfigure} \\
	\vspace{-0.1cm}
\begin{subfigure}[b]{0.24\textwidth}
	    \label{wass_vary_b_1}
		\includegraphics[width=\textwidth]{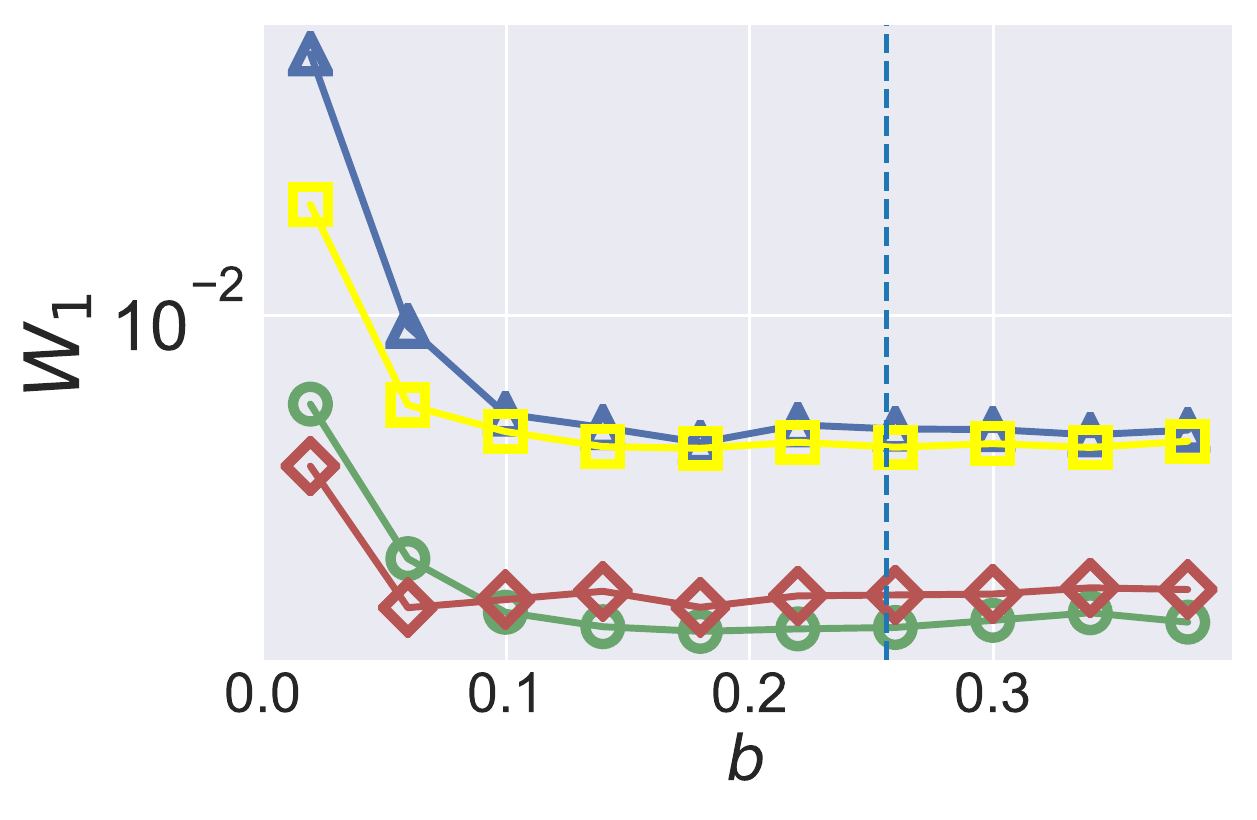}
		\vspace{-0.6cm}
		\subcaption{$\epsilon = 1.0$, $b_{\SW} =0.256$}
	\end{subfigure}
\begin{subfigure}[b]{0.24\textwidth}
	    \label{wass_vary_b_2}
		\includegraphics[width=\textwidth]{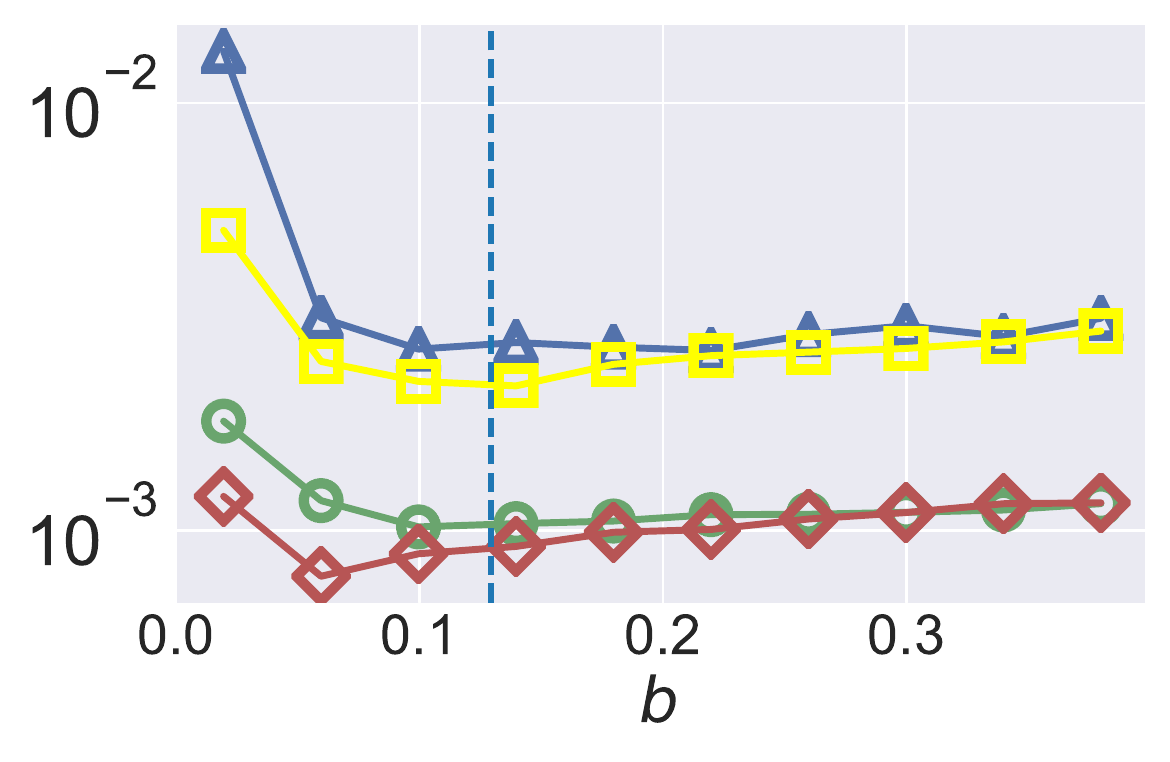}
	    \vspace{-0.6cm}
		\caption{$\epsilon = 2.0$,  $b_{\SW} =0.129$}
	\end{subfigure}
\begin{subfigure}[b]{0.24\textwidth}
	    \label{wass_vary_b_3}
		\includegraphics[width=\textwidth]{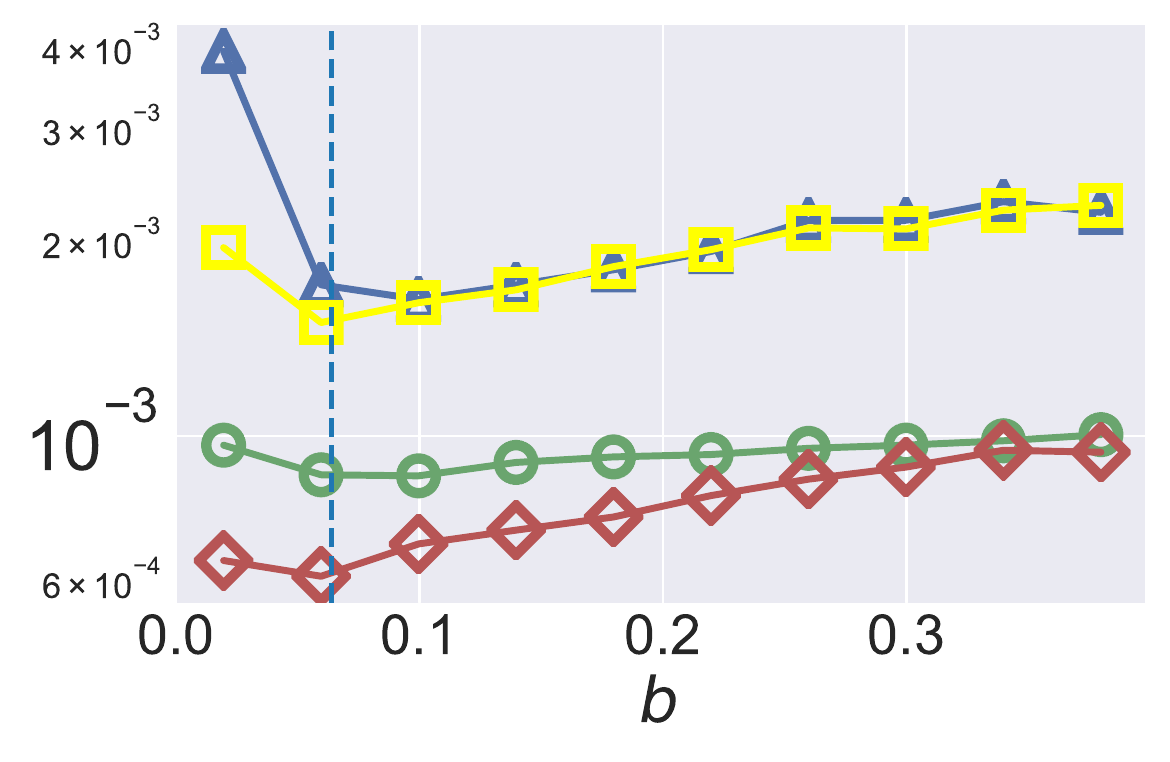}
		\vspace{-0.6cm}
		\caption{$\epsilon = 3.0$, $b_{\SW} =0.064$}
	\end{subfigure}
\begin{subfigure}[b]{0.24\textwidth}
	    \label{wass_vary_b_4}
		\includegraphics[width=\textwidth]{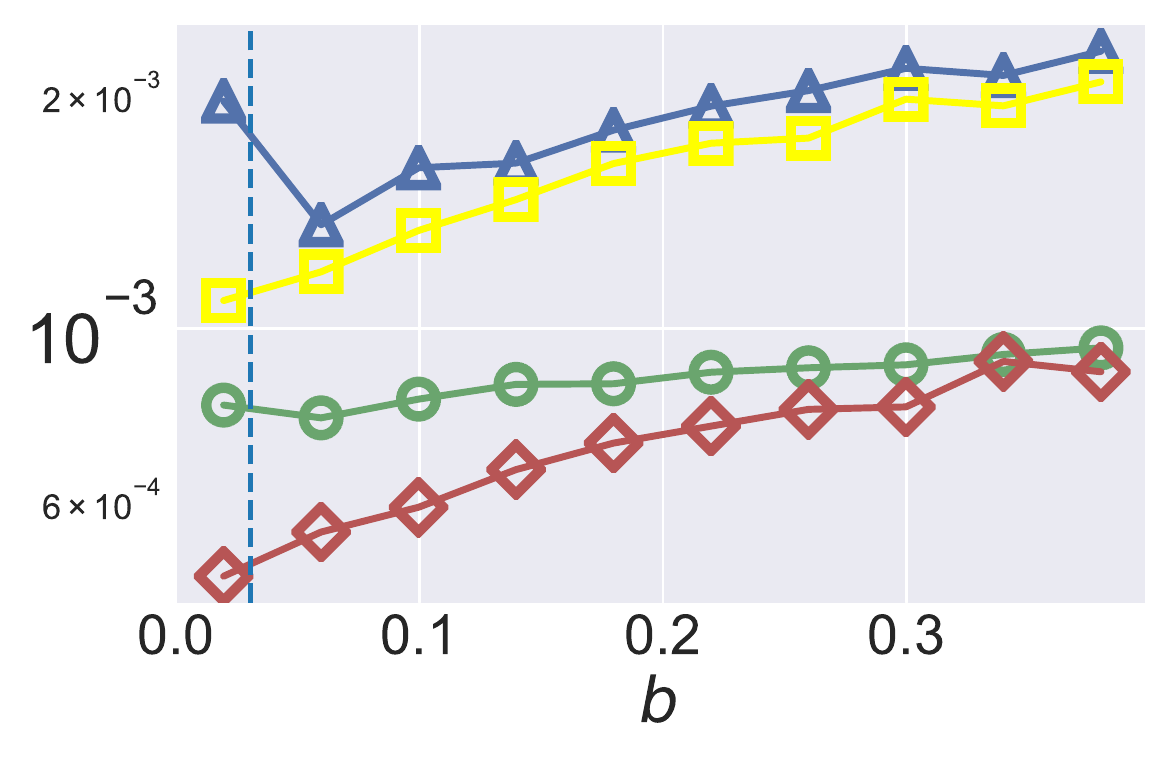}
		\vspace{-0.6cm}
		\caption{$\epsilon = 4.0$, $b_{\SW} =0.030$}
	\end{subfigure}
	\\
	\vspace{-0.4cm}
    \caption{
	Wasserstein distances between the true data and the estimation produced by EMS algorithm with fixed $\epsilon$ values and varying $b$ from $0.01$ to $0.38$.
	Dotted vertical lines means the used $b_{\SW}$ in Section~\ref{subsec:choose_b}.
	}
	\label{exp:vary_b}
	\vspace{-0.3cm}
\end{figure*}

\begin{figure*}[t]
    \centering
    \begin{subfigure}[b]{\textwidth}
		\includegraphics[width=1\textwidth]{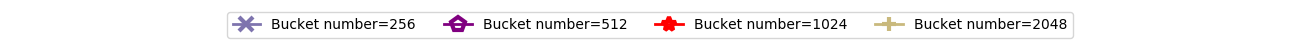}
		\vspace{0.00mm}
	\end{subfigure} \\
	\vspace{-0.5cm}
	\addtocounter{subfigure}{0}
\begin{subfigure}[b]{0.23\textwidth}
		\includegraphics[width=\textwidth]{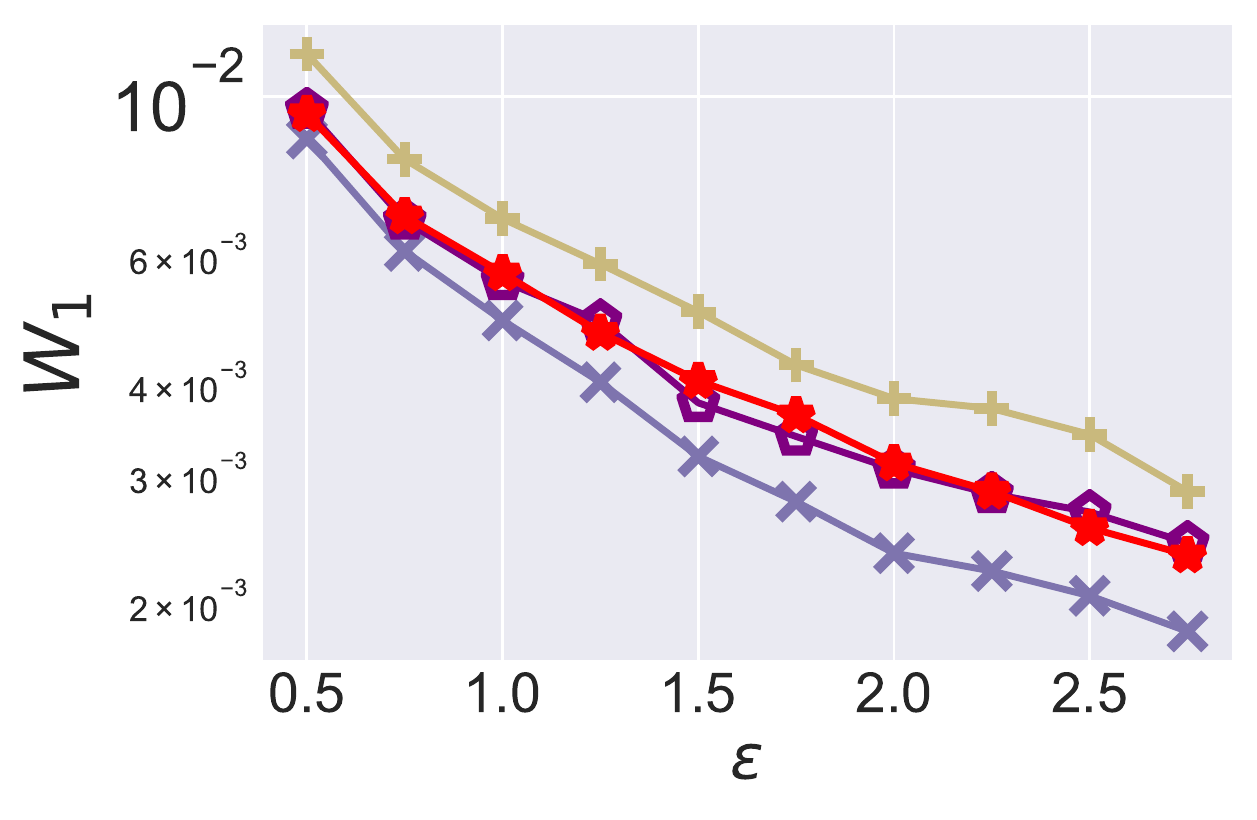}
		\vspace{-0.7cm}
		\caption{Beta(5,2)}
		\label{continuous_beta}
	\end{subfigure}
\begin{subfigure}[b]{0.23\textwidth}
		\includegraphics[width=\textwidth]{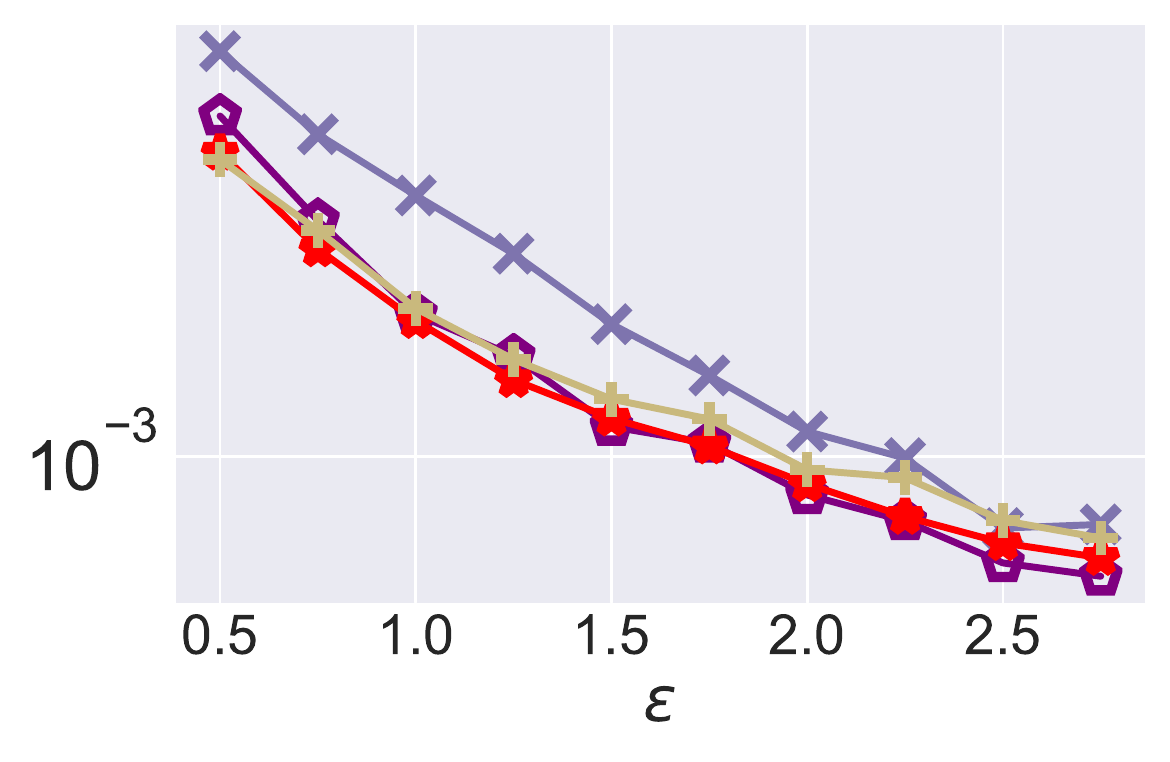}
		\vspace{-0.7cm}
		\caption{Taxi Pickup time}
		\label{continuous_PT}
	\end{subfigure}
\begin{subfigure}[b]{0.23\textwidth}
		\includegraphics[width=\textwidth]{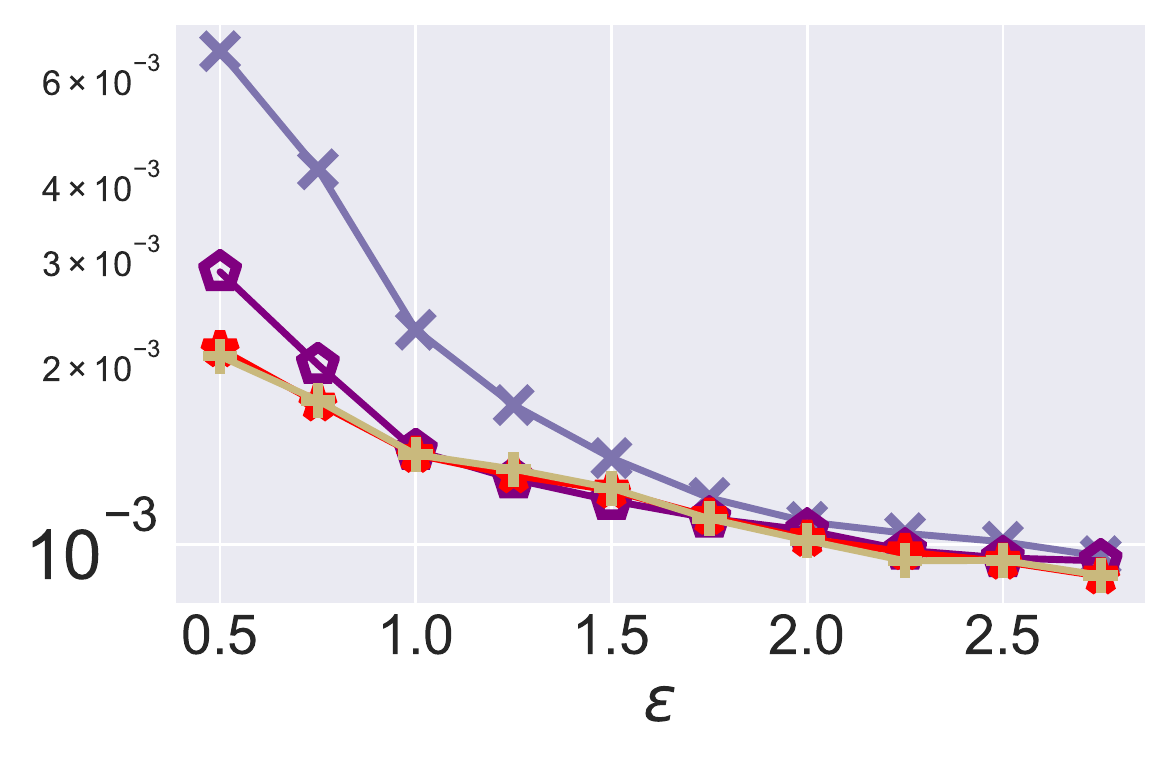}
		\vspace{-0.7cm}
		\caption{Income}
		 \label{continuous_INC}
	\end{subfigure}
\begin{subfigure}[b]{0.23\textwidth}
		\includegraphics[width=\textwidth]{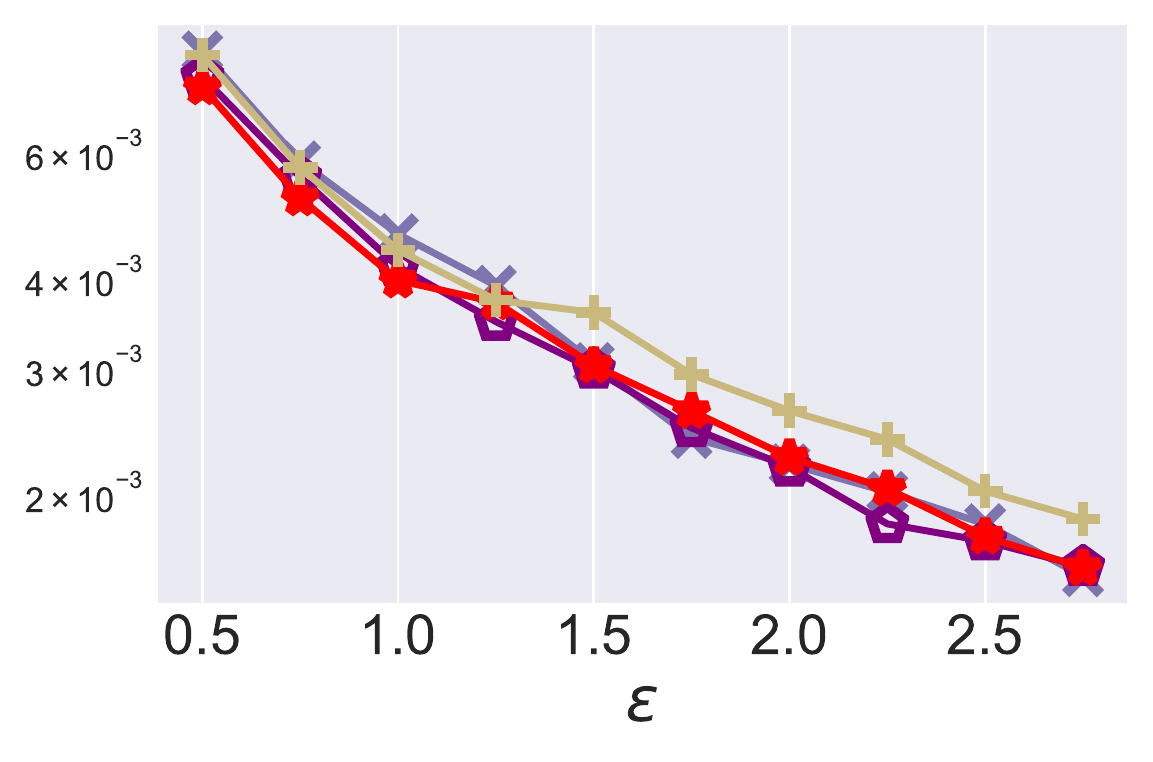}
		\vspace{-0.7cm}
		\caption{Retirement}
		\label{continuous_RT}
    \end{subfigure}
    \\
	\vspace{-0.4cm}
	\caption{
	Wasserstein distance between estimated and true distribution with different bucketization granularity.
	}
	\label{exp:bucketize}
	\vspace{-0.2cm}
\end{figure*}

\section{Related Work}
\label{sec:related}

Differential privacy has been the \textit{de facto} notion for protecting privacy.  In the local setting, we have seen real world deployments: Google deployed RAPPOR~\cite{ccs:ErlingssonPK14} as an extension within Chrome; Apple~\cite{url:apple} uses similar methods to help with predictions of spelling and other things; Microsoft also deployed an LDP system for telemetry collection~\cite{nips:DingKY17}.

\mypara{Categorical Frequency Oracle.}
One basic mechanism in LDP is to estimate frequencies of values.
There have been several mechanisms~\cite{arXiv:AcharyaSZ18, ccs:ErlingssonPK14, stoc:BassilyS15, nips:DingKY17, nips:BassilyNST17,uss:WangBLJ17} proposed for this task.  Among them, \cite{uss:WangBLJ17} introduces \olh, which achieves low estimation errors and low communication costs.  
Our paper develop new frequency oracles for numerical attribues.

\mypara{Handling Ordinal/Numerical Data.}
When the data is ordinal, the straightforward approach is to bucketize the data and apply categorical frequency oracles.   
\cite{infocom:wang2017local} considers distribution estimation, but with a strictly weaker privacy definition (intuitively, the more different between $x$ and $x'$, the more privacy budget).  There are also mechanisms that can handle numerical setting, but focusing on the specific task of mean estimation.  Specifically, \cite{jasa:DuchiJW18,nips:DingKY17} use \SR and \cite{icde:WangXYHSSY18} uses \PM.  
These two approaches have been discussed in Section~\ref{sec:background} and compared in the experiments. 

\mypara{Post-processing.}
Given the result of a privacy-preserving algorithm, one can utilize the structural information to post-process it so that the utility can be improved.  
In the setting of centralized DP, Hay et al.~\cite{pvldb:HayRMS10} propose an efficient method to minimize $L_2$ difference between the original result and the processed result.  This approach utilizes the hierarchy structure constraint.
Besides that, the authors of~\cite{kdd:lee2015maximum} also consider the non-negativity constraint and propose to use ADMM to obtain result that achieves maximal likelihood.
As ADMM is not efficient for high dimensional case, a gradient descent based algorithm is proposed~\cite{icml:mckenna2019graphical}.

In the LDP setting, \cite{icde:WangXYHSSY18} and \cite{pvldb:KulkarniCD18} also consider the hierarchy structure and apply the technique of~\cite{pvldb:HayRMS10}.  We propose to use ADMM instead of~\cite{pvldb:HayRMS10}, which improves utility. 

Without using the hierarchical constraint (only consider \CFO), Jia et al.~\cite{infocom:jia2018calibrate} propose to utilize external information about the dataset (e.g., assume it follows a power-law distribution), and Wang et al.~\cite{arXiv:Wang19LLLS} consider the constraints that the distribution is non-negative and sum up to $1$.  
Bassily~\cite{aistats:Bassily19} and Kairouz et al.~\cite{icml:KairouzBR16} study the post-processing for some \CFO with MLE.
Compared with those existing methods, our work is also a post-processing method but is applied to a new Square Wave reporting method and requires different techniques (such as EM with smoothing).  
 \section{Conclusions}
\label{sec:conc}

We have studied the problem of reconstructing the distribution of a numerical attribute under LDP.
We introduce HH-ADMM as an improvement to existing hierarchy-based methods.  
Most importantly, we propose the method of combining Square Wave reporting with Expectation Maximization and Smoothing.
We show that Square Wave mechanism has the best utility among general wave mechanisms, and introduce techniques to choose the bandwidth parameter $b$ by maximizing an upper bound of mutual information.  
Extensive experimental evaluations demonstrate that \SW with EMS generally performs the best under a wide range of metrics.  
We expect these protocols and findings to help improving the deployment of LDP protocols to collect and analyse numerical information.

\bibliographystyle{abbrv}
\bibliography{ref}

\appendix
\section{Derivation of EM} 
\label{app_EM}
In the post-processing phase, the aggregator receives $n$ reports from users, which are denoted as $\rdvv = \{\rdv_1, \ldots, \rdv_n\}$.  Let $\vv = \{v_1, \ldots, v_n\}$ be the true input values for square wave report mechanism.  
We also assume that the input value of each user is drawn independently from a fixed unknown probability distribution.
The aggregator want to make the estimate frequency histogram $\estx \in [0, 1]^{d}$ as close as possible to the true private frequency $\truex$.
It is equivalent to maximize the maximum log-likelihood $L(\estx)$. 
\begin{align*}
    L(\estx) = \ln \Pr{\rdvv | \estx}
\end{align*}

In this section, we derive the EM algorithm presented in Algorithm~\ref{algo:EM}. 
Note that existing work also uses EM to post-process results of \CFO ~\cite{popets:FantiPE16,tifs:RenYYYYMY18}, but our proposed EM algorithm takes aggregated results and is thus more efficient.

Let $\estx\iter{t}$ be the estimation of $\x$ after $t$ iterations of EM algorithm.
Then the difference of the log-likelihood between $\estx\iter{t}$ and $\x$ is
\begin{align*}
    L(\truex) - L(\estx\iter{t}) = \ln \Pr{\rdvv|\estx} - \ln \Pr{\rdvv|\estx\iter{t}}.
\end{align*}
With Bayesian rule, the difference can be written as: 
\begin{align*}
    &\quad L(\truex) - L(\estx\iter{t}) \\
    &= \ln \sum_{\vv} \Pr{\rdvv|\vv, \truex}\Pr{\vv|\truex} - \ln \Pr{\rdvv|\estx\iter{t}}\\
    &= \ln \sum_{\vv} \Pr{\rdvv|\vv, \truex}\Pr{\vv|\truex}\frac{\Pr{\vv|\rdvv, \estx\iter{t}}}{\Pr{\vv|\rdvv, \estx\iter{t}}} - \ln \Pr{\rdvv|\estx\iter{t}} \\
    &= \ln \sum_{\vv} \Pr{\vv|\rdvv, \estx\iter{t}}\frac{\Pr{\rdvv|\vv, \truex}\Pr{\vv|\truex}}{\Pr{\vv|\rdvv, \estx\iter{t}}} -\ln \Pr{\rdvv|\estx\iter{t}}\\
\end{align*}
By Jensen's inequality, 
\begin{align*}
    &\quad L(\truex) - L(\estx\iter{t}) \\
    &\geq \sum_{\vv} \Pr{\vv|\rdvv, \estx\iter{t}} \ln \left(\frac{\Pr{\rdvv|\vv, \truex}\Pr{\vv|\truex}}{\Pr{\vv|\rdvv, \estx\iter{t}}}\right) -\ln \Pr{\rdvv|\estx\iter{t}}\\
\end{align*}
With a few more steps,
\begin{align*}
    &\quad L(\truex) - L(\estx\iter{t}) \\
    &\geq \sum_{\vv} \Pr{\vv|\rdvv, \estx\iter{t}} \ln \left(\frac{\Pr{\rdvv|\vv, \truex}\Pr{\vv|\truex}}{\Pr{\vv|\rdvv, \estx\iter{t}}}\right) \\
    &\quad - \sum_{\vv} \Pr{\vv|\rdvv, \estx\iter{t}}\ln \Pr{\rdvv|\estx\iter{t}} \\
    &= \sum_{\vv} \Pr{\vv|\rdvv, \estx\iter{t}} \ln \left(\frac{\Pr{\rdvv|\vv, \truex}\Pr{\vv|\truex}}{\Pr{\vv|\rdvv, \estx\iter{t}}\Pr{\rdvv|\estx\iter{t}}}\right) \\
\end{align*}
Define  $l(\truex|\estx\iter{t})$ as the following:
\begin{align*}
  & l(\truex|\estx\iter{t}) = \\
  & \quad L(\estx\iter{t}) + \sum_{\vv} \Pr{\vv|\rdvv, \estx\iter{t}} \ln \left(\frac{\Pr{\rdvv|\vv, \truex}\Pr{\vv|\truex}}{\Pr{\vv|\rdvv, \estx\iter{t}}\Pr{\rdvv|\estx\iter{t}}}\right) 
\end{align*}
Based on the previous results, it is obvious that $L(\truex) \geq l(\truex|\estx\iter{t})$.
In order to maximize the log-likelihood, we can maximize the $l(\truex|\estx\iter{t})$ directly.
\begin{align*}
    \estx\iter{t+1} = \argmax_{\estx}\{l(\estx|\estx\iter{t})\}
\end{align*}
Following the our goal, we have 
\begin{align*}
    &\quad \text{maximize } l(\estx|\estx\iter{t})\\
    & \Leftrightarrow \text{maximize } \sum_{\vv} \Pr{\vv|\rdvv, \estx\iter{t}} \ln \left(\Pr{\rdvv|\vv, \estx}\Pr{\vv|\estx} \right)\\
    &\Leftrightarrow \text{maximize } \sum_{\vv} \Pr{\vv|\rdvv, \estx\iter{t}} \ln \Pr{\rdvv,\vv| \estx}\\
    &\Leftrightarrow \text{maximize } \E_{\vv|\rdvv,\estx\iter{t}}\left[\ln \Pr{\rdvv,\vv| \estx}\right]  .
\end{align*}
Next, notice that $\E_{\vv|\rdvv,\estx\iter{t}}[\ln \Pr{\rdvv,\vv| \estx}]$ can be rewritten in the following form as user samples are independent: 
\begin{align*}
    &\quad \E_{\vv|\rdvv,\estx\iter{t}}\left[\ln \Pr{\rdvv,\vv| \estx}\right] \\
    &= \E_{\vv|\rdvv,\estx\iter{t}} \ln \prod_{k=1}^{n}P(\rdv_k,v_k| \estx)\\
    &= \sum_{k=1}^{n}\E_{\vv|\rdvv,\estx\iter{t}}[\ln \Pr{\rdv_k,v_k | \estx}] \\
\end{align*}
If we consider the probability of all possible output values, then we can further derive the above equation into 
\begin{align*}
    &\quad \E_{\vv|\rdvv,\estx\iter{t}}[\ln \Pr{\rdvv,\vv| \estx}] \\
    &= \sum_{k=1}^{n} \sum_{i=1}^{d} \Pr{v_k \in B_i|\rdv_k, \estx\iter{t}}\ln \Pr{\rdv_k, v_k\in B_i | \estx}\\
    &= \sum_{k=1}^{n} \sum_{i=1}^{d} \\
    &\frac{\Pr{ \rdv_k|v_k \in B_i, \estx\iter{t}} \Pr{v_k \in B_i | \estx\iter{t}}}{\Pr{\rdv_k |\estx\iter{t}}}\ln \Pr{\rdv_k| v_k \in B_i, \estx}\estx_i\\
    &= \sum_{k=1}^{n} \sum_{i=1}^{d} \frac{\Pr{ \rdv_k|v_k \in B_i, \estx\iter{t}}\estx\iter{t}_i}{\Pr{\rdv_k |\estx\iter{t}}}\ln \Pr{\rdv_k| v_k \in B_i, \estx}\estx_i
\end{align*}
where $v_k \in B_i$ means value $v_k$ falls in the $i^{th}$ bucket of the \emph{input domain}.
Because $\Pr{\rdv_k| v_k\in B_i, \estx}$ is decided by the \SW reporting, we can ignore it and focus only on $\estx_i$.
\begin{align*}
  &\quad \argmax_{\estx}\{\E_{\vv|\rdvv,\estx\iter{t}}[\ln \Pr{\rdvv,\vv| \estx}] \} \\
&= \argmax_{\estx}\{  \sum_{i=1}^{d} \ln \estx_i  \sum_{k=1}^{n}\frac{\Pr{ \rdv_k|v_k\in B_i, \estx\iter{t}}\estx\iter{t}_i}{\Pr{\rdv_k |\estx\iter{t}}}\} \\
\end{align*}
Since we only need to estimate the frequencies, we can combine the randomized reports together to simplify the computation:
\begin{align*}
    &\quad \sum_{k=1}^{n}\frac{\Pr{ \rdv_k|v_k \in B_i, \estx\iter{t}}\estx\iter{t}_i}{\Pr{\rdv_k |\estx\iter{t}}} \\
    &= \sum_{k=1}^{n} \sum_{j \in [\tilde{d}]} \mathbf{1}[\rdv_k \in \tilde{B}_j] \frac{\Pr{ \rdv_k\in \tilde{B}_j|v_k\in B_i, \estx\iter{t}}\estx\iter{t}_i}{\Pr{\rdv_k \in \tilde{B}_j |\estx\iter{t}}} \\
\end{align*}
The indicator function $\mathbf{1}[\rdv_k\in \tilde{B}_j]$ equals $1$ if $\rdv_k$ falls in the $j^{th}$ bucket in the \emph{output domain}, and $0$ otherwise.
We can swap the inner and outer summations,
\begin{align*}
    &\quad \sum_{k=1}^{n}\frac{\Pr{ \rdv_k|v_k \in B_i, \estx\iter{t}}\estx\iter{t}_i}{\Pr{\rdv_k |\estx\iter{t}}} \\
    & =  \sum_{j \in [\tilde{d}]} \sum_{k=1}^{n} \mathbf{1}[\rdv_k \in \tilde{B}_j] \frac{\Pr{ \rdv_k\in \tilde{B}_j|v_k\in B_i, \estx\iter{t}}\estx\iter{t}_i}{\Pr{\rdv_k \in \tilde{B}_j |\estx\iter{t}}} \\
    & = \sum_{j \in [\tilde{d}]} n_j   \frac{\Pr{ \rdv\in \tilde{B}_j|v\in B_i, \estx\iter{t}}\estx\iter{t}_i}{\Pr{\rdv \in \tilde{B}_j |\estx\iter{t}}} \\
    & = \estx\iter{t}_i  \sum_{j \in [\tilde{d}]} n_j   \frac{\Pr{ \rdv\in \tilde{B}_j|v\in B_i, \estx\iter{t}}}{\Pr{\rdv \in \tilde{B}_j |\estx\iter{t}}} \\
    & = \estx\iter{t}_i \sum_{j \in [\tilde{d}]} n_j  \frac{\Pr{\rdv\in \tilde{B}_j|v\in B_i, \estx\iter{t}}}{\sum_{r=1}^{d} \Pr{\rdv\in \tilde{B}_j|v \in B_{r}, \estx\iter{t}}\estx\iter{t}_r} .\\
\end{align*}
Here $n_j$ means the count of reports that have value in $\tilde{B}_j$.
This is what we define and compute in the E-step in Algorithm~\ref{algo:EM}:
$$P_i = \estx\iter{t}_i \sum_{j \in [\tilde{d}]} n_j  \frac{\Pr{\rdv\in \tilde{B}_j|v\in B_i, \estx\iter{t}}}{\sum_{r=1}^{d} \Pr{\rdv\in \tilde{B}_j|v \in B_{r}, \estx\iter{t}}\estx\iter{t}_r}$$

Now, the maximization  problem becomes
\begin{align*}
  &\quad \argmax_{\estx}\{\E_{\vv|\rdvv,\estx\iter{t}}[\ln \Pr{\rdvv,\vv| \estx}] \} \\
  &= \argmax_{\estx}\{\sum_{i=1}^{d}P_i\ln \estx_i \} \\
  &=\argmax_{\estx}\{\sum_{i=1}^{d-1}P_i\ln \estx_i + P_d\ln(1 - \sum_{i=1}^{d-1}\estx_i)\}
\end{align*}
The last equation holds because of the consistency requirement $\sum_{i=1}^{d}\estx_i = 1$. 
Let $f(\estx) = \sum_{i=1}^{d-1} P_i \ln\estx_i  +  P_m \ln (1-\sum_{i=1}^{m-1}\estx_i)$, then the derivative of $f(\estx)$ is 
\begin{align*}
    \frac{\partial f(\estx)}{\partial \estx_i} &= \frac{P_i}{\estx_i} - \frac{P_d}{1-\sum_{j=1}^{d-1}\estx_j}
\end{align*}
Let $\frac{\partial f(\estx)}{\partial \estx_i} = 0$, we have
\begin{align*}
    &\estx_i = \frac{P_i(1-\sum_{j=1}^{d-1}\estx_j)}{P_d}\\
    \implies &\sum_{i=1}^{d-1}\estx_i = \frac{\sum_{i=1}^{d-1} P_i(1-\sum_{j=1}^{d-1}\estx_j)}{P_d}\\
    \implies & \sum_{j=1}^{d-1}\estx_j = \frac{\sum_{i=1}^d P_i - P_d}{\sum_{i=1}^d P_i }\\
    \implies &\sum_{i=1}^{d-1}\estx_i = 1 - \frac{P_d}{\sum_{j=1}^{d}P_j}\\
    \implies &\estx_i =  \frac{P_i}{\sum_{j=1}^{d}P_j}
\end{align*}
Notice that this is what we compute in the M-step, and we finish the derivation of EM algorithm.

% %
 \section{Derivation of HH-ADMM}
\label{app:derivation_hh_admm}
As mentioned in the main text, the optimization problem is the following:
\begin{align}
    \text{minimize} & \quad \frac{1}{2}\|\estx - \tilde{\x}\|_2^2 \\ 
    \text{subject to}  &\quad \mathbf{A} \estx = 0, \quad \estx \succcurlyeq 0, \quad \estx_0 = 1 \nonumber
\end{align}
It can be transformed into:
\begin{align}
    \text{minimize} &\quad \frac{1}{2}\|\y\|_2^2 + \mathbb{I}_{\mathcal{C}}(\z) + \mathbb{I}_{\mathcal{N}_+}(\w) \nonumber\\
    \text{subject to } & \quad \estx - \tilde{\x} - \y = 0 \nonumber\\
    & \estx - \z = 0 \nonumber\\
    & \estx - \w = 0,\nonumber
\end{align}
where $\mathbb{I}_{S}(\x)$ is an indicator function that equals $0$ if $\x\in S$, and $\infty$ otherwise.  Here $\mathcal{C} = \{\x | \mathbf{A}\x = 0 \}$ and $\mathcal{N}_+$ is the set of vectors that are non-negative and normalized (each level sum up to $1$).

With the scaled dual variables, the optimization problem can be transform into the following dual augmented problem
\begin{align*}
    \max_{\vmu, \vnu, \veta} \min_{\estx, \y, \z, \w}& \frac{1}{2}\|\y\|_2^2 + \mathbb{I}_{\mathcal{C}}(\z) + \mathbb{I}_{\mathcal{N}_+}(\w) \\
    & + \frac{\rho}{2}\|\estx - \tilde{\x} - \y + \vmu\|_2^2  + \frac{\rho}{2}\|\estx - \z +\vnu\|_2^2  \\
    & + \frac{\rho}{2}\|\estx - \w + \veta\|_2^2
\end{align*}
where $\rho$ is the penalty parameter in the dual augmented problem, and $\vmu, \vnu, \veta$ are dual variables.
We set $\rho=1$ in our algorithm.

When we use the ADMM algorithm to solve the dual augmented problem, we can update each variable iteratively,
\begin{align}
    &\y = \argmin_{\y}\left(\frac{1}{2}\|\y\|_2^2 + \frac{\rho}{2}\|\estx - \tilde{\x} - \y + \vmu\|_2^2 \right) \label{eq:update_y}\\
    &\z = \argmin_{z}\left(\mathbb{I}_{\mathcal{C}}(\z) + \frac{\rho}{2}\|\estx - \z +\vnu\|_2^2\right) \label{eq:update_z}\\
    &\w = \argmin_{\w}\left(\mathbb{I}_{\mathcal{N}_+}(\w) + \frac{\rho}{2}\|\estx - \w + \veta\|_2^2\right)\label{eq:update_w}\\
    &\x = \argmin_{\x}\left(\frac{1}{2}\|\estx - \tilde{\x} - \y + \vmu\|_2^2  + \frac{1}{2}\|\estx - \z +\vnu\|_2^2  \right. \nonumber\\
    &\left. \quad + \frac{1}{2}\|\estx - \w + \veta\|_2^2\right) \label{eq:update_x}\\
    &\vmu = \vmu+ \estx - \tilde{\x} - \y  \\
    & \vnu = \vnu + \estx - \z \\
    & \veta = \veta + \estx - \w 
\end{align}

Previous work~\cite{combettes2011proximal, kdd:lee2015maximum} introduced how the do the update mentioned above.
The sub-problem~\eqref{eq:update_y} and~\eqref{eq:update_x} are essentially least square optimization problems.
The other two sub-problem~\eqref{eq:update_z} and \eqref{eq:update_w} can be solved by 
\begin{align}
    \z = \Pi_{\mathcal{C}}(\estx +\vnu) \\
    \w = \Pi_{\mathcal{N}_+}(\estx + \veta)
\end{align}
where $\Pi_{S}$ is the operation that project the point to the closest point in set $S$.
For $\Pi_{\mathcal{C}}$, an efficient algorithm is given in~\cite{pvldb:HayRMS10}; and Norm-Sub is an efficient algorithm for $\Pi_{\mathcal{N}_+}$~\cite{arXiv:Wang19LLLS}.
Therefore, Algorithm~\ref{algo:HH-ADMM} is efficient and can converge.

\begin{algorithm}[H]
\begin{algorithmic}
\STATE \textbf{Input:} $ \tilde{\x}$
\STATE \textbf{Output:} $\estx$
\STATE $\estx \gets \tilde{\x}, t \gets 0$
\STATE Init $\y, \z, \w, \vmu, \vnu, \veta$ to be vectors of $0$
\WHILE{not converge} 
    \STATE $\y \gets  \frac{1}{2}(\estx - \tilde{\x} + \vmu) $
    \STATE $\z \gets \Pi_{\mathcal{C}}(\estx + \vnu)$
    \STATE $\w \gets \Pi_{\mathcal{N}_+}(\estx + \veta)$
    \STATE $\estx \gets  \frac{1}{3}\left((\y + \tilde{\x} -  \vmu) +  (\z - \vnu) + (\w - \veta) \right)$
    \STATE \COMMENT{update dual variables}
    \STATE $\vmu \gets \vmu + \estx - \tilde{\x} - \y$
    \STATE $\vnu \gets \vnu +  \estx - \z$
    \STATE $\veta \gets \veta + \estx - \w$
\ENDWHILE
\STATE Return $\estx$
\end{algorithmic}
\caption{HH-ADMM}
\label{algo:HH-ADMM}
\end{algorithm}

\end{document}